%% file: circConversion.tex
\pdfoutput=1
\documentclass{article}
\input{preamble}

\title{Unitary-circuit semantics for measurement-based computations}

\author{Niel de Beaudrap\footnote{%
	niel.debeaudrap@gmail.com}\\[1em]
	\large	Quantum Information Theory Group\\
 	\large 	Institut f\"ur Physik und Astronomie\\
 	\large	Universit\"at Potsdam
}

\date{}

\begin{document}
\maketitle

\begin{abstract}
	One-way measurement based quantum computations (\oneway) may describe unitary transformations, via a composition of CPTP maps which are not all unitary themselves.
	This motivates the following decision problems: Is it possible to determine whether a ``quantum-to-quantum'' \oneway\ procedure (having non-trivial input and output subsystems) performs a unitary transformation?
	Is it possible to describe precisely \emph{how} such computations transform quantum states, by translation to a quantum circuit of comparable complexity?
	In this article, we present an efficient algorithm for transforming certain families of measurement-based computations into a reasonable unitary circuit model, in particular without employing the principle of deferred measurement.
\end{abstract}



\section{Introduction}

The one-way measurement model of quantum computation (\oneway), presented by~\cite{RB01} and subsequently elaborated in~\cite{RBB03,DKP06}, allows descriptions of unitary transformations in terms of operations, such as measurements, which are not themselves unitary.
This may be regarded as being due to an appropriate combination of the observables which are measured on the one hand, and the stabilizer group of the state-space (in the ``quantum-to-quantum'' setting where we admit an input subsystem which may support an arbitrary input state) on the other hand.
This intuition may be used in certain cases to prove that particular measurement-based computations perform unitary transformations~\cite{BKMP07}, using a slight modification of the stabilizer formalism~\cite{GotPhD}.

As we are often interested in describing computation in terms of unitary evolutions, measurement-based computing presents us with natural decision problems.
Is it possible to efficiently determine whether a \oneway\ procedure performs a unitary transformation?
Is it possible to describe precisely \emph{how} such computations transform quantum states, by translation to \eg\ a unitary circuit of bounded complexity?
If not, is it at least possible to efficiently verify whether it performs a \emph{given} transformation?

These questions arise because simple unitary circuit models --- such as uniform circuit families generated over the gate-set $\ens{H, T, \cZ }$ consisting of the Hadamard gate, the $\pion 8$ phase-gate, and controlled-$Z$ gate respectively (with no classically controlled operations and measurement deferred to the end of a computation) --- are the \emph{de facto} standard models of quantum computation.
Such circuits provide the (thus far) best-developed set of idioms for uniformly expressing algorithms \emph{independently of architecture}.
Other mathematical representations of quantum computation may be useful as abstractions of proposals for implementations of quantum computers, as in the case of the measurement-based models~\cite{GC99,RB01} and the adiabatic model~\cite{FGGS00}; and ``purely physical'' idioms seemingly divorced from uniform circuit families may suggest new quantum algorithms, as in the case of the $O(\sqrt{N})$-time algorithm for evaluating \textsc{nand}-trees by quantum walks~\cite{FGG07}.
However, simple unitary circuit models are currently the preferred means for expressing quantum computations.

Because of the present uncertainty as to \emph{which} implementation proposal will prove successful, the unitary circuit model represents a natural candidate for an intermediate language for representing algorithms independently of the target architecture, even if alternative models suggest new algorithms.
Therefore, one may argue that how effectively one can translate from a given model of computation to unitary circuit models, and back, is an important measure of how well-understood the model is.
Unless some one proposal for implementation comes to dominate over the others, or an unlikely breakthrough (\eg\ the discovery of an efficient classical algorithm for simulating quantum circuits) is made, translation to a ``simple'' model of quantum computation such as the unitary circuit model is the best that can be reasonably hoped for as a \emph{uniform} way of obtaining low-level semantics of a quantum computation.

We are generally not interested in \emph{arbitrary} translations to and from unitary circuit models: we should require a translation scheme which does not substantially increase the complexity of the operation being translated.
For instance, provided a \oneway\ procedure $\fP$ for performing a unitary transformation, we may decompose each classically controlled measurement into a classically controlled change-of-basis rotation followed by a standard basis measurement; applying the principle of delayed measurement, we may then replace the classically controlled rotations with coherently controlled ones.
We may then decompose the gates of this circuit with respect to a desired gate-set, such as $\ens{H, T, \cZ }$.
However, in constructions for \oneway\ procedures such as those defined in~\cite{RBB03,DKP06}, each single-qubit transformation requires at least one additional qubit, one two-qubit entangling operation, and a classically controlled measurement which may depend on the results from several measurements.
If one applies the principle of deferred measurement to such a \oneway\ procedure $\fP$, which ``represents'' a unitary circuit $\fC$ using one of these constructions, the result is a circuit $\fC'$ which requires many more qubits than $\fC$, and has at least one additional two-qubit operation for each single-qubit operation in $\fC$.
This represents a constant factor inflation in both the number of operations of the circuit; and depending on the ratio of single-qubit to either the number of qubits or the number of two-qubit operations in the original circuit $\fC$, the asymptotic increase in the either the memory required or the number of two-qubit gates in $\fC'$ over $\fC$ is in principle unbounded.

In contrast, we may consider translation procedures $\cS$ from the \oneway\ model, possibly with a restricted domain, which yield circuits of ``reasonable'' complexity.
In this case, we may establish a reference for the complexity of the resulting circuits by considering (as in the previous paragraph) \oneway\ procedures which arise from a translation routine $\cR$ from unitary circuits to the \oneway\ model.
We may then require that $\img(\cR) \subset \dom(\cS)$, and that $\cS \circ \cR$ maps each circuit in $\dom(\cR)$ to one which is no more complex (and which performs the same unitary transformation).
In particular, we will be interested in the case where $\dom(\cS)$ contains a wide range of operations.
In this article, we require that $\dom(\cR)$ consists of the set of unitary bijections generated over $\ens{H,T,\cZ}$ (without measurements or classical control).
Then $\dom(\cS)$ must be similarly general, as it must ``recognize'' different representations of arbitrary unitary bijections in the \oneway\ model, and produce unitary circuits which are equivalent both in function and efficiency to the ``original''.
We encapsulate these requirements as follows:\footnote{%
	The following definition is a refinement of the conditions described in pages~81--82 of~\cite{BeaudPhD} to include pairs of maps $\cS$ and $\cR$ for which, for some \oneway\ procedures $\fP$, the procedure $(\cR \circ \cS)(\fP)$ differs in complexity from $\fP$.} 
\begin{definition}
	\label{def:semanticMap}
	Let $\cC_\cB$ be a model of unitary circuits generated over a gate-set $\cB$.
	A \emph{semantic map for the \oneway\ model in terms of unitary circuits} is an efficiently computable partial function $\cS$ from \oneway\ procedures to $\cC_\cB$, for which membership in $\dom(\cS)$ is efficiently decideable, and for which there exists a corresponding efficient translation routine $\cR$ from $\cC_\cB$ back to the \oneway\ model such that
	\begin{itemize}
	\item
		$\dom(\cR) = \cC_\cB$ and $\img(\cR) \subset \dom(\cS)$; and

	\item
		for any circuit $C \in \cC_\cB$, $C' = (\cS \circ \cR)(C)$ performs the same transformation as $C$, and has complexity bounded above by that of $C$: in particular,
		\begin{itemize}
		\item	$C'$ acts on at most as many qubits as $C$;
		\item	$C'$ has at most the same depth as $C$; and
		\item	the number of one-qubit and two-qubit gates in $C'$ are each at most the number of such gates in $C$.
		\end{itemize}
	\end{itemize}
	We say that $\cR$ is a \emph{representation map} of $\cC_\cB$ in the \oneway\ model, for which $\cS$ provides the semantics.
\end{definition}

\paragraph{Summary of results.}

We present a single semantic map $\cS$ with respect to two related constructions for \oneway\ procedures: one based on the work of Raussendorf, Browne, and Briegel~\cite{RBB03} (which we will refer to as the \emph{simplified RBB construction}) for producing cluster-state based procedures; and the simplified procedure defined by Danos, Kashefi, and Panangaden~\cite{DKP06} (which we will refer to as the \emph{DKP construction}) for producing \oneway\ procedures whose entanglement graphs are not required to be subgraphs of rectangular grids.
We will restrict our attention to the case of unitary bijections, \ie\ unitary circuits generated over $\ens{H,T,\cZ}$ which do not introduce fresh qubits.
In doing so, we present the following:
\begin{description}
\item[A simple alternative notation for quantum circuits.]\hfill\\
	In order to facilitate descriptions of the construction of \oneway\ procedures, as well as the task of determining when two circuits are congruent, we define in \autoref{sec:stableIndexNotation} a variant notation for unitary circuits (and products of linear operators generally) having convenient combinatorial properties.


\item[An explicit account of two constructions of \oneway\ procedures.]\hfill\\
	In order to explicitly describe a semantic map, we present in \autoref{sec:constructions} the construction of~\cite{DKP07} (``the DKP construction'') and a simplified version of the construction of~\cite{RBB03} (``the simplified RBB construction) for \oneway\ procedures, including an explicit account of the time complexity of these constructions.

\item[A succinct characterization of measurement dependencies in these constructions.]\hfill\\
	Using adjacency matrices of directed graphs, we characterize the classical measurement dependencies of the \oneway\ procedures arising from the DKP construction and the simplified RBB construction in \autoref{sec:dependencies}.

\item[Circuit constructions from combinatorial decompositions.]\hfill\\
	In \autoref{sec:starDecomp}, we illustrate how we may efficiently decompose the structures underlying a \oneway\ procedure obtained from the DKP or simplified RBB constructions.
	Using these decompositions, we show in \autoref{sec:candidateConstr} how such a decomposition may be used to describe unitary circuits which may perform the same operation as the original \oneway\ procedure.
\end{description}
These results allow us to define an efficient algorithm for computing a semantic map $\cS$ which produces circuits for unitary bijections from an efficiently identifiable class of \oneway\ procedures.

\vspace{-1ex}
\subsubsection*{Related work.}

A study into interpreting the effect of \oneway\ procedures in terms of unitary circuits is performed by Childs, Leung, and Nielsen~\cite{CLM2005} via single-qubit ``teleportation-based quantum computation''.
Their results may be described as providing alternative proofs of validity of the constructions of Raussendorf, Browne, and Briegel~\cite{RBB03} of particular unitary operations in the \oneway\ model, but do not present conditions for when such a reduction to the circuit model may be performed.

The first algorithm to convert \oneway\ procedures to unitary circuits was presented by Kashefi and Danos~\cite{DK06}, who investigated conditions under which \oneway\ procedures produced by the DKP construction perform unitary transformations.
These conditions are formulated in terms of combinatorial structures, which were shown to be efficiently detectable in~\cite{Beaudrap08,MP08}.
Using these structures, \cite{DK06} outlines a procedure to generate a unitary circuit using those structures.
Building on the results of~\cite{Beaudrap08,BP08,MP08}, we extend those results to include procedures produced by the simplified RBB construction.
We also characterize the measurement dependencies which arise in such \oneway\ procedures when they are in a particular canonical form, in order to verify whether such a \oneway\ procedure performs the unitary described by the constructed circuit; and demonstrate that the complexity of the circuits which are constructed are no more than the complexity of the original circuits.

The results here represent an extension of results in Chapter 3 of~\cite{BeaudPhD} concerning the DKP construction, by proving remarks made briefly there concerning the simplified RBB construction.
In the few places where the terminology here differs from \cite{BeaudPhD} (or from prior results), the differences are 	explicitly noted.

\section{Preliminaries}

We begin by presenting conventions and notational devices used throughout the article.
Among these are a concise operational notation for procedures in the \oneway\ model, and a transparent notation for unitary circuits which has convenient combinatorial properties.
These will prove instrumental to simplify the description of the constructions for \oneway\ procedures in the following section.

We take for granted a number of basic concepts in graph theory; the interested reader may refer to Diestel's text~\cite{Diestel} for elementary definitions and proofs.

\subsection{Unitary circuits}
\label{sec:unitaryCircuits}

We will consider a particular model of quantum circuits: circuits without classical control, and with measurement (and trace-out operations) deferred to the end.
We refer to such circuits as \emph{unitary} circuits.
As is common practise, we will ignore the terminal measurements and trace-out operations, and consider only the component involving unitary transformations and unitary embeddings.

\subsubsection{Sets of approximately universal gates}
\label{sec:gateSets}

As we noted in the introduction, we must consider a particular unitary circuit model with a fixed set of unitary gates.
We consider a model closely related to most common gate sets: the set $\ens{H, T, \cZ}$ consisting of the Hadamard gate $H$, the $\pion 8$ gate $T$, and the controlled-$Z$ gate $\cZ$.
These are expressed in matrix form by
\begin{align}
		H
	=&\,
		\sfrac{1}{\sqrt 2}
		\begin{bmatrix}
		 		\;1\;	&	\;1\;	\\	\;1\;	&	-1
		\end{bmatrix},
	&
		T
	=&\,
		\begin{bmatrix}
		 		\e^{-i\pi/8}\!	&	0	\\	0	&	\!\e^{i\pi/8}
		\end{bmatrix},
	&
		\cZ
	=&\,
		\begin{bmatrix}
		 		\;1\;	&	\;0\;	&	\;0\;	&	\;0\;	\\
				\;0\;	&	\;1\;	&	\;0\;	&	\;0\;	\\
				\;0\;	&	\;0\;	&	\;1\;	&	\;0\;	\\
				\;0\;	&	\;0\;	&	\;0\;	&	-1
		\end{bmatrix},
\end{align}
respectively.
As $H$ and $T$ generate (up to scalar factors) a dense subgroup of the unitary group $\U(2)$, and as $\textsc{cnot} = (\idop \ox H) \,\cZ\, (\idop \ox H)$, the set of gates $\ens{H, T, \cZ}$ is approximately universal for quantum computation.

To describe computations in the \oneway\ model, we will also be interested in two related sets of unitary gates, involving operators such as $J(\theta)$ and $\Zz$ given by
\begin{subequations}
\begin{gather}
	\label{eqn:Jtheta}
	 	J(\theta)
	\,=\;
		H \e^{-i Z \theta/2}
	\,=\,
		\sfrac{1}{\sqrt 2}
		\begin{bmatrix}
			\;\e^{-i\theta/2}\;	&	\;\e^{i\theta/2}\;	\\
			\;\e^{-i\theta/2}\;	&	-\e^{i\theta/2} 
		\end{bmatrix},	\quad\text{and}
	\intertext{}
	\label{eqn:Zz}
		\Zz
	\,=\;
		\e^{-i\pi Z \otimes Z /4}
	\,=\,
		\begin{bmatrix}
			\e^{-i\pi/4}	&	0	&	0	&	0	\\
			0	&	\e^{i\pi/4}		&	0	&	0	\\
			0	&	0	&	\e^{i \pi/4}	&	0	\\
			0	&	0	&	0	&	\e^{-i\pi/4} 
		\end{bmatrix}.
\end{gather}
\end{subequations}
Note in particular that $J(\frac{m\pi}{4}) = HT^m$, and $\Zz = \e^{i \pi / 4} \, (T^2 \ox T^2) \, \cZ$.
Therefore, the gate-sets
\begin{align}
	\label{eqn:altGateSets}
		\cB_\DKP
	\,=&\;
		\ens{\big.J(\theta), \cZ}_{\theta \in \frac{\pi}{4}\Z}
	&
		\text{and}
	&&
		\cB_\RBB
	\,=&\;
		\ens{\big.J(\theta), \Zz}_{\theta \in \frac{\pi}{4}\Z}
\end{align}
are also approximately universal for quantum computation.\footnote{%
	The gates $J(\pi/4)$ and $J(-\pi/4)$ alone generate a dense subgroup of $\U(2)$ up to scalar factors: this is in fact shown by \cite{NC00} (Section~4.5.2) in the proof that $H$ and $T$ are approximately universal for single-qubit operations.
	We include gates $J(\theta)$ for each $\theta$ an integer multiple of $\pi/4$ for both $\cB_\DKP$ and $\cB_\RBB$ in order to obtain a closer correspondence with circuits generated over $\ens{H,T,\cZ}$ in the construction of \oneway\ procedures in later sections.
}

In unitary circuit models, one usually permits the introduction of fresh qubits into the circuit. 
In order to prevent the preparation of computationally expensive initial states, we will require that all fresh qubits be initially in some fixed pure state independent of the other qubits.
Without loss of generality, to obtain a simpler correspondence to the operations of the \oneway\ model, we will suppose that all fresh qubits are initially prepared in the $\ket{+} \propto \sfrac{1}{\sqrt 2}\big( \ket{0} + \ket{1} \big)$ state.

\subsubsection{Parsimonious sets of gates}

As well as being approximately universal for quantum computation, the gate sets described above have another useful property:
\begin{definition}
	\label{def:parsimonious}
	A set $S$ of unitary operations is \emph{parsimonious} if every multi-qubit gate in $S$ is diagonal, and if distinct gates in $S$ acting on a common qubit commute if and only if they are diagonal.\footnotemark
\end{definition}\footnotetext{%
	The title of~\cite{DKP06} predates this definition of ``parsimonious'' by four years, and does not refer to the concept we have defined here.
	However, the sets of unitary operations $\ens{J(\theta), \cZ}_{\theta \in \R}$ and $\ens{J(\theta), \cZ}_{\theta \in \frac{\pi}{4}\Z}$ described there are both parsimonious in this sense.}
Each of the gate-sets described above are parsimonious.
For instance, it is immediately clear that $\ens{H, T, \cZ}$ is parsimonious; that $\cB_\DKP$ and $\cB_\RBB$ are parsimonious follows from the fact that
\begin{align}
		J(\alpha)\herm J(\beta)\herm J(\alpha)J(\beta)
	\;\;=&\;\;
		\e^{i Z \alpha/2} \,H\, \e^{i Z \beta/2} \,H^2 \;\! \e^{-i Z \alpha/2} \,H \;\! \e^{-i Z \beta/2}
	\notag\\=&\;\;
		\e^{i Z \alpha/2} \, \e^{i X (\beta - \alpha)/2} \, \e^{-i Z \beta/2}	\;,
\end{align}
which is equal to $\idop_2$ if and only if $\alpha = \beta$.

The motivation for this terminology is that, for a ``parsimonious'' gate-set $S$ and for any orthogonal basis $\ens{\ket{\phi_0}, \ket{\phi_1}} \in \cH$ other than the standard basis, there is at most one gate in $S$ for which the vectors $\ket{\phi_j}$ are eigenvectors. 
Requiring that a get-set be parsimonious is a strong constraint: by definition, operations in a unitary circuit over such a set commute if and only if they are of the same ``type'' (they are either equal or inverse to each other), if they act on disjoint sets of qubits, or if they are both diagonal.
This is evidently useful in determining when two quantum circuits are equivalent up to transposition of gates.

\subsubsection{Stable-index tensor notation for unitary circuits}
\label{sec:stableIndexNotation}

For a representation map $\cR$ from unitary circuits to \oneway\ procedures and a semantic map $\cS$, to show that $\cS \circ \cR$ maps unitary circuits to equivalent circuits of the same (or lesser) complexity, we may attempt to show that $\cS \circ \cR$ maps each circuit in $\dom(\cR)$ to one which is congruent\footnote{%
	For two unitary circuits, or two \oneway\ computations, we will say that the computations are congruent if they differ only by transpositions of commuting operations (and by a possible relabelling of the qubits operated upon).}
to the original (possibly up to a set of ``trivial'' simplifications).
For this purpose, it will be useful to have a notation for unitary circuits which is easy to produce by an algorithm, in which irrelevant details such as the order of commuting gates are essentially absent.
In general, this is a problematic task: however, we will be satisfied with a notation which solves this problem for parsimonious gate-sets, such as the sets of gates which we use in this article.

One solution is to describe a notation whose meaning is invariant under permutation of any of the terms.
An example of such a notation is Einstein summation notation (\cite{Einstein1916}, Section~B5), in which (for unitary circuits on qubits) operations $U$ are described as tensors, using their coefficients with respect to the standard basis on individual qubits involved in the circuit.
We use superscripts for column-indices, and subscripts for row-indices: thus, we have
\begin{subequations}
\begin{gather}
		H^j_k
	\;=\,
		\sfrac{1}{\sqrt 2} (-1)^{jk}	,
	\\[1ex]
		J(\theta)^a_b
	\;=\,
		\sfrac{1}{\sqrt 2} (-1)^{ab} \,\e^{i (2b - 1) \theta/2 },
	\\[1ex]
		\cZ^{r,s}_{t,u}
	\;=\,
		\delta_{r,t} \delta_{s,u} (-1)^{rs}	,
\end{gather}
\end{subequations}
and so on (where $\delta_{j,k}$ is the Kronecker delta).
In an expression involving multiple gates, two such tensors are identified as acting on a common qubit when the column-index of one agrees with the row-index of another: the coefficients of the compound tensor are computed by summing over the range of the repeated index.
The order of the multiplication of the tensors is then determined by the location of the index: an operation $U$ strictly precedes another operation $V$ if one of the row-indices of $U$ matches a column-index of $V$.
To represent tensor products of two operators, it suffices to juxtapose them with disjoint sets of indices.
For example, given
\begin{subequations}
\begin{align}
	\label{eqn:arbitraryOperatorProduct}
		C
	\;\;=\;\;
		\big(\idop \ox W\big) \, \big(V \ox \idop \ox \idop\big) \,U\, \big(\idop \ox T \ox \idop\big) \, \big(\idop \ox \ket{\phi} \ox \ket{\psi}\big)
\end{align}
for a generic three-qubit operator $U$, two-qubit $W$, single-qubit operators $V$ and $T$, and state-vectors $\ket{\phi}, \ket{\psi} \in \cH$, we may equivalently write
\begin{align}
	\label{eqn:genericEinsteinExample}
		C_{a_0}^{a_2,b_3,c_2}
	\;\;=\;\;
		W_{b_2,c_1}^{b_3,c_2}
		V_{a_1}^{a_2}
		U_{a_0,b_1,c_0}^{a_1,b_2,c_1}
		T_{b_0}^{b_1}
		\ket{\phi}^{\!b_0}
		\ket{\psi}^{\!c_0}	.
\end{align}
\end{subequations}
Furthermore, the coefficients described on the right-hand side remain the same under an arbitrary permutation of the factors, and the equality remains true after an application of any (one-to-one) relabelling of the indices.

However, while the order of the factors is not significant in Einstein notation, circuits which are equivalent up to rearrangements of commuting operations still give rise to tensor expressions which are not congruent via permutations of factors and index relabelling.
This is because the tensor indices themselves fix an order of multiplication: for instance, although the operations $\cZ_{a,b}$, $\cZ_{a,c}$, and $\cZ_{b,c}$ commute (where the subscripts here indicate that the operations are performed on some pair chosen from three qubits named $a$, $b$, and $c$), the products 
\begin{subequations}
\label{eqn:tripleCZproduct}
\begin{align}
		\cZ_{b,c} \; \cZ_{a,c} \; \cZ_{a,b}
	\;\;\equiv&\,\;\;
		\cZ_{b_1,c_1}^{b_2,c_2} \; \cZ_{a_1,c_0}^{a_2,c_1} \; \cZ_{a_0,b_0}^{a_1,b_1}	\,,
	\quad \text{and}
	\\[1ex]
		\cZ_{a,c} \; \cZ_{b,c} \; \cZ_{a,b}
	\;\;\equiv&\,\;\;
		\cZ_{a_1,c_1}^{a_2,c_2} \; \cZ_{b_1,c_0}^{b_2,c_1} \; \cZ_{a_0,b_0}^{a_1,b_1}
\end{align}
\end{subequations}
give rise to expressions in Einstein notation which cannot be made the same by a relabeling of indices and re-ordering of terms.

In the case that two gates commute if and only if they are both diagonal, as in circuits generated from parsimonious sets of gates, we may resolve this problem as follows.
We use a notation similar to Einstein notation, but in which in which an index can be repeated multiple times without summation, when the index represents a qubit in a part of a unitary circuit in which the standard basis is preserved for that qubit.\footnote{%
 	\label{fn:remarkRandomVariable}%
 	If one interprets tensor indices as a random variables ranging over $\ens{0,1}$, whose value may be realized by a standard basis measurement, this corresponds to re-using indices when the value of that variable would be unchanged (or ``stable'') under some number of subsequent operations.}

\paragraph{Stable index notation.}
	Consider a unitary operator $U$ on $N$ qubits, whose eigenvectors are all of the form $\ket{s_1} \ox \cdots \ox \ket{s_n} \ox \ket{\Psi}$ for $s_1, \ldots, s_n \in \ens{0,1}$ and $\ket{\Psi} \in \cH\sox{m}$, where $N = n + m$.
	We may represent the non-zero matrix coefficients of $U$ in the standard basis by
	\begin{subequations}
	\label{eqn:stableIndex}
	\begin{align}
		U\text{\large$\big.\pseu[ s_1, \cdots, s_n : y_{n+1}, \cdots, y_N / x_{n+1}, \cdots, y_N ]$}
	\;=&\,\;
		\bra{s_1, \ldots, s_n, y_{n+1}, \ldots, y_N} U \ket{s_1, \ldots, s_n, x_{n+1}, \ldots, x_N}
	\notag\\=&\,\;
		\text{\large$\text{\normalsize $U$}_{\raisebox{-0.1ex}{$\scriptstyle\!\!\: s_1, \ldots, s_n, x_{n+1}, \ldots, x_N$}}^{\raisebox{0.2ex}{$\scriptstyle\, s_1, \ldots, s_n, y_{n+1}, \ldots, y_N$}}$}	\;.
	\end{align}
	(By hypothesis, the coefficients of $U$ are zero when the values of the first $n$ row- and column-indices differ.)
	That is, for tuples $\vec s \in \ens{0,1}^n$ and $\vec a, \vec d \in \ens{0,1}^m$, we define $U\pseu[\vec s: \vec a / \vec d]$ by the operator equality
	\begin{align}
		\label{eqn:operatorDefnStableIndex}
			U
		\;\;=&
			\sum_{\substack{\vec s \in \ens{0,1}^n \\ \vec{a},\vec{d} \in \ens{0,1}^m}}		\!\!
				U\pseu[s_1, \cdots, s_n : a_1, \cdots, a_m / d_1, \cdots, d_m ]
				\;\; \Big( \ket{\vec s}\bra{\vec s} \ox \ket{\vec a}\bra{\vec d} \Big)	\;.
	\end{align}
	\end{subequations}

\paragraph{Terminology and composition conventions for stable-index notation.}
	In an expression $U\pseu[\vec s: \vec a / \vec d]$ as defined by \eqref{eqn:operatorDefnStableIndex}, we call the indices $s_j$ \emph{stable}, the indices $d_j$ \emph{deprecated}, and the indices $a_j$ \emph{advanced}.
	A product of such expressions is well-formed if each index is advanced at most once and deprecated at most once; there is no limit to the number of times an index may occur as a stable index.
	An index which is advanced in one term and deprecated in another is a \emph{bound} index; otherwise it is \emph{free}.
	We sum implicitly over all bound indices in a stable-index tensor expression, as in Einstein notation: we then form composite tensors $V \circ U$ by matching advanced \emph{or stable} indices of $U$ to corresponding deprecated \emph{or stable} indices of $V$.
	(If the indices of $U$ and $V$ are disjoint, this is a tensor product of $U$ and $V$.)

\begin{example}
\begin{subequations}
 	The $\pion 8$ operation $T$ is a single-qubit diagonal operator, and so has one stable index and no deprecated or advanced indices; the Hadamard operation $H$ is a single-qubit non-diagonal operator, and so has no stable indices and one deprecated/advanced index each.
	We may then write their coefficients as 
	\begin{align}
			T \pseu[ x ]	\,=&\;	\e^{(2x - 1) i \pi/8}
		&
			\text{and}
		&&
		H \pseu[: v / u ]	\,=&\;	\sfrac{1}{\sqrt 2} (-1)^{uv}.
	\end{align}
	The $J(\pion 4)$ operation may be defined by composing these two operators, while the inverse of $J(\mpion 4)$ is the other composition of these operators; we may write these compositions as
	\begin{align}
		J(\pion 4) \pseu[: y / x ]	\,=&\;	H\pseu[: y / x ] T\pseu[ x ]
		&
			\text{and}
		&&
		J(\mpion 4)\herm \pseu[: y / x ]	\,=&\;	T \pseu[ y ] H \pseu[: y / x ] .
	\end{align}
	Finally, the tensor product $P = T \ox H$ has eigenvectors of the form $\ket{z} \ox \ket{\psi}$ for $z \in \ens{0,1}$ and $\ket{\psi} \in \cH$ an eigenvector of $H$: its coefficients may then be written as $P \pseu[ z: y / x ] \,=\, H\pseu[: y / x ] T\pseu[ z ]$.
\end{subequations}
\end{example}

\paragraph{Identified systems.}
	For operators $U: \cH\sox{N} \to \cH\sox{N}$ acting on distinguishable qubits, we describe how $U$ acts on these qubits by fixing an order of the qubits, where the advanced indices and deprecated indices of the qubits are presented in the same order.

\paragraph{Extension to general linear operators.}
	For operators $A: \cH\sox{N} \to \cH\sox{M}$ where $N$ and $M$ may differ, a stable-index tensor expression for $A$ will have advanced indices which do not have corresponding deprecated indices if $M > N$, or deprecated indices which do not have corresponding advanced indices if $N > M$.
	Then $A$ describes an operation which either \emph{adds} qubits to the input space, or \emph{discards} qubits from the input space.
	(These may occur when $A$ describes an isometric embedding, or a single-qubit measurement with a selected outcome, respectively.)
	We may represent a qubit which is added by an advanced index with no corresponding deprecated index (padding the sequence of deprecated indices by a space, dot, or similar placeholder), and similarly for deprecated indices representing qubits which are discarded.
	
\begin{example}
	We may transcribe the operator expression given in \eqref{eqn:genericEinsteinExample} as
	\begin{align}
		\label{eqn:genericStableExample}
 		C\pseu[: a_2, b_3, c_2/a_0, \;\cdot\;, \;\cdot\;]
	\;\;=\;\;
		W\pseu[:b_3,c_2/b_2,c_1]
		V\pseu[:a_2/a_1]
		U\pseu[:a_1,b_2,c_1/a_0,b_1,c_0]
		T\pseu[:b_1/b_0]
		\ket{\phi}\pseu[:b_0/]
		\ket{\psi}\pseu[:c_0/]	.
	\end{align}
	The order of the indices is purely conventional (in the same way that the order of the indices of $C$ in \eqref{eqn:genericEinsteinExample} is purely conventional); we could equally well define $C$ above as $C\pseu[: b_3, a_2, c_2/\;\cdot\;, a_0, \;\cdot\;]$; the change in the order of the qubits causes a corresponding change in the matrix representation of $C$.
	The adjoint of the tensor in~\eqref{eqn:genericStableExample} is given by $C\herm\pseu[:a_0,\;\cdot\;,\;\cdot\;/a_2,b_3,c_2]$.
\end{example}

A formal procedure for translating tensor notations such as in~\eqref{eqn:arbitraryOperatorProduct}, or involving operators $U_{a,b,\cdots,q}$ labelled with the qubits that they act on, is easy to define: such a procedure is presented in \autoref{apx:stableIndexConstructionUnconstrained}.

\begin{example}
	Considering the same tensor $C$ as in \eqref{eqn:genericStableExample}, in the special case that $U$ has eigenvectors of the form $\ket{a}\ket{b}\ket{\psi}$ for $a,b \in \ens{0,1}$ and for some states $\ket{\psi} \in \cH$, we may simplify the expression of $C$ to 
	\begin{align}
		\label{eqn:simplifiedGenericStableExample}
 		C\pseu[: a_2, b_3, c_2/a_0, \;\cdot\;, \;\cdot\;]
	\;\;=\;\;
		W\pseu[: b_3,c_2/b_1,c_1]
		V\pseu[:a_2/a_0]
		U\pseu[a_0,b_1\!:c_1/c_0]
		T\pseu[:b_1/b_0]
		\ket{\phi}\pseu[:b_0/]
		\ket{\psi}\pseu[:c_0/]	.
	\end{align}
	The operators described by the expressions in~\eqref{eqn:genericStableExample} and~\eqref{eqn:simplifiedGenericStableExample} are identical in this case.
\end{example}

The following example illustrates the intuitive purpose of stable-index notation: by allowing stable indices to be repeated multiple times while remaining ``free'' variables, we also discard some redundant combinatorial information about the order of the product of operators.

\begin{example}
	\label{eg:commutingControls}
	\begin{subequations}
	For generic unitary operators $U \in \U(2)\sox{n}$, we may define the operator 
	\begin{gather}
		\label{eqn:controlFunctor}
		\ctrl U \;=\; \ket{0}\bra{0} \ox \idop_2^{\,\ox n} \;+\; \ket{1}\bra{1} \ox U \,.
	\end{gather}
	Let $V \in \U(2)$ and $W \in \U(2)\sox{2}$.
	For qubits $a$, $b$, $c$, and $d$, consider the circuits $\ctrl V_{a,d} \, \ctrl W_{a,b,c}$ and $\ctrl W_{a,b,c} \, \ctrl V_{a,d}$ (where the subscripts here denote the qubits on which the operators act).
	We may represent these two circuits with stable-index expressions
	\begin{align}
			\ctrl V\pseu[a:d_1/d_0] \, \ctrl W\pseu[a:b_1,c_1/b_0,c_0]
		&&
		\text{and}
		&&
			\ctrl W\pseu[a:b_1,c_1/b_0,c_0] \, \ctrl V\pseu[a:d_1/d_0] \;,
	\end{align}
	respectively.
	Note that these are equivalent expressions, up to permutations of the factors: this may be regarded as being due to the fact that the operations $\ctrl V_{a,d}$ and $\ctrl W_{a,b,c}$ commute (\ie\ the circuits are congruent).
	\end{subequations}
\end{example}

In contrast to Einstein notation, while the numerical meaning of a stable-index tensor expression is invariant under permutations of the terms, information about the order of consecutive operations is not represented when these operations preserve the standard basis of qubits on which they both act.
Circuits which are congruent up to the re-ordering of commuting gates may then give rise to ``equivalent'' stable-index tensor expressions in the following sense:
\begin{definition}
	An \emph{isomorphism} of stable-index tensor expressions $S_1$ and $S_2$ is a bijective relabelling $\lambda$ of the \emph{index labels} of one stable-index expression $S_1$ to those of $S_2$, and a bijective mapping $\tau$ of \emph{terms} of $S_1$ to those of $S_2$, such that for any term  $U\pseu[ s_1, \ldots, s_n : a_1, \ldots, a_m / d_1, \ldots, d_m ]$ in $S_1$ (with $n \ge 0$ stable indices, and $m \ge 0$ each of deprecated and advanced indices including placeholders), there is a term in $S_2$ given by
	\begin{gather}
			\tau\paren{\bigg.U\pseu[s_1,\, \ldots, s_n : a_1,\, \ldots, a_m/d_1,\, \ldots, d_m]}
		\;\;=\;\;
			U\pseu[\lambda(s_1),\, \ldots, \lambda(s_n) : \lambda(a_1),\, \ldots, \lambda(a_m)/\lambda(d_1),\, \ldots, \lambda(d_m)]
	\end{gather}
	(where placeholders on the left are mapped by $\lambda$ to placeholders on the right).\footnote{%
		Note that isomorphisms of stable-index representations may change the order of the terms as well as the index labels.
		That is, they are combinatorial homomorphisms, not just a relabelling of the indices in the terms of $S_1$ kept in the same sequence.}
\end{definition}

Isomorphic stable-index expressions represent not only equivalent unitary operations (which follows from the implicit summation convention), but congruent circuit decompositions as well:
\begin{theorem}
	\label{thm:stableIndexCongruence}
	Let $C_1$ and $C_2$ be two sequences of unitary operators over a common gate set, and acting on the same number of qubits.
	If the stable-index tensor representations of $C_1$ and $C_2$ are isomorphic, then $C_1$ and $C_2$ are congruent via a re-ordering of commuting operators and relabelling of the qubits.
	Furthermore, if $C_1$ and $C_2$ are generated from a parsimonious set of gates, the converse also holds.
\end{theorem}
Thus, for unitary circuits constructed from a parsimonious set of gates, isomorphism of stable-index tensor expressions is equivalent to congruence up to rearrangements of commuting gates.
This feature is the motivation both for this tensor notation, and for the interest in parsimonious sets of gates as in \autoref{def:parsimonious}.
We only explicitly use advantage of this feature of this tensor notation in \autoref{sec:semanticAlg}; however, other combinatorial properties of stable-index notation will frequently prove convenient in describing constructions of \oneway\ procedures from the circuit model.
The proof of \autoref*{thm:stableIndexCongruence} is provided in \autoref{apx:stableIndexCongruence}.

While stable-index expressions lack any information about the order of operations which preserve the standard basis of the qubits that they act on in common, it is possible to retrieve information about the order of non-commuting gates, in order (for example) to determine how to actually perform such an operation as a physical transformation.
As in Einstein notation, the order of non-commuting gates can still be determined from the indexing.
\begin{definition}
	\label{def:indexSuccessorFn}
	For a stable-index tensor expression $S$, the \emph{index successor function} $f$ for $S$ is the mapping defined on the indices $\delta$ of $S$ which are deprecated but which have a corresponding advanced index $\alpha$ (corresponding conventionally to the same qubit), such $f(\delta) = \alpha$ for each such $\delta$ and $\alpha$.
\end{definition}
Note that for a deprecated index $v \in \dom(f)$, any term $U$ which involves the index $f(v)$ must be performed after any operation in the circuit acting on $v$ (excepting $U$ itself, if it is the unique term in which $v$ and $f(v)$ both occur).
We may define:
\begin{definition}
	\label{def:interactionHypergraph}
	For a stable-index tensor expression $S$, the \emph{interaction hyper-graph} 
	for $S$ is the hypergraph $G$ on the set of indices in $S$, whose edges consist of those sets of indices which are acted on by some operation in $C$.
\end{definition}
For any $v \in \dom(f)$, the index $f(v)$ is adjacent to $v$ in this hypergraph.
We may then use $f$ and $G$ to establish a partial order $u \le v$ on the indices of $S$, which describes when an index $u$ is deprecated before another index $v$, given an ordering of the terms which is consistent with an order in which the operations may be performed.
We may define $\le$ as follows.
For two indices $f(v), f(w)$ which occur in a common term, each must be deprecated strictly after the indices $v$ and $w$ in order for the term to describe a well-defined (acyclic) product of operators.
If $v \sim w$ denotes the adjacency relation in $G$, we then have
\begin{align}
	\label{eqn:flowPremonition}
		v \le f(v)
	&&
		\text{and}
	&&
		w \sim f(v)	\;\implies\; v \le w
\end{align}
for any $v \in \dom(f)$.
We can then describe a pre-order\footnote{%
	\label{fn:preorder}%
		A \emph{pre-order} $\preceq$ is a binary relation which is reflexive ($x \preceq x$ for all $x$) and transitive (if $x \preceq y$ and $y \preceq z$ then $x \preceq z$), but with no further requirements.
 		Partial orders and equivalence relations are both examples of pre-orders.
		An example of a pre-order which is neither an equivalence relation nor a partial order is the relation $x \preceq y \;\iff\; \Re(x) \le \Re(y)$ for complex numbers $x$ and $y$.} 
$U \preceq V$ on the operations in $S$, where $U \preceq V$ if $U$ has an index $u$ and $V$ has an index $v$ such that $u \le v$.
We then have $U \preceq V \preceq U$ if and only if either $U$ and $V$ are the same term, or different terms which act on common qubits but commute; if $U \preceq V \not\preceq U$, then $U$ must occur strictly earlier than $V$.
We may then obtain an order for the operations in $S$ by finding a linear order which extends $\preceq$.

In the case of the gate-sets defined in \autoref{sec:gateSets}, the operations all act on single qubits or pairs of qubits.
The interaction hyper-graph $G$ of (a stable-index expression for) a circuit over such a gate-set will then be a conventional graph.
If we ignore single-qubit diagonal operations which add no information to the ordering of the indices, and note that the square of the two-qubit operations are in each case a product of local unitaries, we may even describe the circuit by a \emph{simple} graph (without repeated edges or loops).
This graph structure, and the ordering structures we have described above for circuits will prove convenient for analyzing constructions of \oneway\ procedures.

\subsection{The measurement calculus}
\label{sec:measurementCalculus}

The way in which computation is performed in the \oneway\ model is very different from unitary circuits as we describe them in \autoref{sec:unitaryCircuits}.
In the latter case, computation is performed with interleaving non-commuting single and two-qubit unitaries without any classical control or intermediate measurement.
In contrast, a \oneway\ computation consists of an initial stage in which all of the (mutually commuting) two-qubit operations are performed, followed by a sequence of classically controlled measurements and (for procedures which produce a residual quantum state as output) classically controlled single-qubit unitaries.
In this section, we briefly review the elementary operations in the \oneway\ model and describe standard forms for procedures in the \oneway\ model as a preliminary to describing \oneway\ constructions for unitary transformations, essentially following the terminology and notation introduced by Danos, Kashefi, and Panangaden~\cite{DKP07}.

\subsubsection{Conventions for writing CPTP maps}

In the \oneway\ model, classical information is obtained via measurements which may in practise be destructive (\eg\ as with photodetectors).
Thus, it is usually preferable to describe \oneway\ procedures by completely positive trace-preserving (CPTP) maps acting on density operators, where the input and output spaces of these operators may not be supported by the same set of qubits.

In order to better distinguish CPTP maps from linear operators on $\cH\sox{N}$, we will represent CPTP maps in either sans-serif (as in $\J\;\!$) or fraktur (as in $\fJ\;\!$) typeface.
We also adopt the following additional notational devices:
\begin{description}
\item[Measurement results.]\hfill\\
	When a qubit $v$ has been measured, we will represent the result by a bit $\s[v] \in \ens{0,1}$ whose value is only defined after the measurement.

\item[Operands of CPTP maps.]\hfill\\
	The qubits on which an operation $\Op{v,w,\ldots}{}$ acts, whether to introduce it to the system by preparation, discard it by measurements, or to transform it unitarily, are represented in the subscript.
	(The bit-registers $\s[v]$ storing measurement results are typically omitted from such lists of operands.)

\item[Parameterization of CPTP maps.]\hfill\\
	Operations $\Op{}{\alpha,\beta,\s[v], \s[w], \ldots}$ in a \oneway\ procedure will often be determined by fixed parameters $\alpha, \beta, \ldots$ and the results $\s[v], \s[w], \ldots$ of measurements on qubits $v, w, \ldots$.
	(Typically, fixed parameters $\alpha$ denote the values of angles in the range $(-\pi, \pi]$, and operations will depend on the parities $\s[v] + \s[w] + \cdots$ of measurement results.)
	Such parameters are always described in the superscript.
\end{description}%
\noindent
Finally: as there is little risk of confusion, we will typically abbreviate a composite CPTP map $\Op{}{(1)} \circ \Op{}{(2)} \circ \cdots \circ \Op{}{(n)}$ by $\Op{}{(1)} \Op{}{(2)} \cdots \Op{}{(n)}$.

\subsubsection{Elementary operations of the \oneway\ model}
\label{sec:elemOperationsOneWay}

As in unitary circuit models, we define \oneway\ computations as compositions of elementary operations.
We define these following~\cite{DKP07}.

The initial stage of a \oneway\ computation is the preparation of a large entangled state, which may be obtained by performing entangling operations on a collection of qubits prepared in the $\ket{+}$ state.
Furthermore, the entangling operation which is performed may be described as a product of $\cZ$ operations acting on distinct pairs of qubits.
It will then be useful to define the operations
\begin{align}
		\New{v}(\rho_S)	\;=&\;\,	\rho_S \ox \ket{+}\bra{+}_v
	&
		\text{and}
	&&
		\Ent{vw}(\rho_S)		\;=&\;\,	\cZ_{vw}\, \rho_S\; \cZ_{vw}
\end{align}
for a state $\rho_S$ on a set of qubits $S$, where the subscripts on the right-hand sides indicate the qubits on which the operators act.
For the map $\New{v}$\,, we require that $v \notin S$, and obtain the system $S \union \ens{v}$ as output; for $\Ent{vw}$\,, we require that $v,w \in S$, and obtain the same system $S$ as output.
We refer to $\New{v}$ as a \emph{preparation} map, and $\Ent{vw}$ as an \emph{entangling} map.

The single-qubit measurements we require will be with respect to orthonormal bases on ``the \XY\ plane'' (of the Bloch sphere): that is, with respect to uniform superpositions of $\ket{0}$ and $\ket{1}$ which we denote by
\begin{align}
		\ket{+_\theta}
	\;=&\;\,
		\sfrac{1}{\sqrt 2}\Big( \ket{0} + \e^{i \theta} \ket{1} \Big)
	&
		\text{and}
	&&
		\ket{-_\theta}
	\;=&\;\,
		\sfrac{1}{\sqrt 2}\Big( \ket{0} - \e^{i \theta} \ket{1} \Big)	
\end{align}
for angles $-\pi < \theta \le \pi$.
Using these states, we may define the measurement operators
\begin{gather}
	\label{eqn:xyPlaneMeas}
		\Meas{v}{\theta}(\rho_S)
	\;=\;
		\bra{+_\theta} \:\!\rho_S\:\! \ket{+_\theta}_{\!v}		\,\ox\, \ket{0}\bra{0}_{\s[v]}
		\;\,+\;\;
		\bra{-_\theta} \:\!\rho_S\:\! \ket{-_\theta}_{\!v}	\,\ox\, \ket{1}\bra{1}_{\s[v]} \;,
\end{gather}
acting on a density operator $\rho_S$ on a set of qubits $S$.
(Recall that $\s[v]$ is a classical bit storing the boolean measurement result.\footnote{%
	While we typically do not represent the classical bits used in operations in terms of density operators, we may do so when this is convenient for exposition.})
We require that $v \in S$; the output consists of the system $S \setminus \ens{v}$, ignoring the newly defined bit $\s[v]$.

We extend this notation in a simple way to describe how measurements are adapted according to previous measurements: we introduce the abbreviation
\begin{gather}
	\label{eqn:adaptedMeasurementNotation}
		\Meas{v}{\theta; \beta}
	\;\;=\;\;
		\Meas{v}{(\minus 1)^\beta \!\!\;\cdot \theta}
\end{gather}
for a classically controlled measurement performing either $\Meas{}{\alpha}$ or $\Meas{}{\minus \alpha}$, depending on a boolean expression $\beta = \s[a] + \s[b] + \cdots$ representing the parity of some set of measurement results.\footnote{%
	The calculus defined in~\cite{DKP07} also admits the possibility of a measurement angle being modified by a possible addition of $\pi$, depending on a boolean expression (as with sign dependencies above).
	As this elaboration is not necessary for measurement-based simulation of unitary circuits, we omit it in order to present a  streamlined description of how to construct a \oneway\ procedure in ``standard'' form.}
\label{discn:defaultMeasAngle}
We refer to $\theta$ as the \emph{default (measurement) angle} of the operation, and $\beta$ a \emph{sign-dependency} of $v$; we say that $v$ has a \emph{sign-dependency on $w$} when such an expression $\beta$ depends non-trivially on $\s[w]$.

For \oneway\ procedures which produce a residual quantum state as output (as we consider in this article), it is also necessary to adapt the final state depending on the outcomes of the measurements by using classically controlled single-qubit operations.
For boolean expressions $\beta$ describing the parities of some measurement outcomes, we may decompose these corrections into $X$ and $Z$ operations with the operators
\begin{align}
		\Xcorr{v}{\beta}(\rho_S)
	\;=&\;\,
		X^\beta_v\,\rho_S\,X^\beta_v
	&
		\text{and}
	&&
		\Zcorr{v}{\beta}(\rho_S)
	\;=&\;\,
		Z^\beta_v\,\rho_S\,Z^\beta_v	\;,
\end{align}
acting on a density operator $\rho_S$ on a set of qubits $S$ which contains $v$.

Finally, a purely classical operation which is useful in descriptions of \oneway\ procedures is a classical bit-flip operation on a measurement result $\s[v]$:
\begin{gather}
	\label{eqn:shiftOptor}
		\Shift{v}{\beta}(\rho)
	\;\;=\;\;
		\ket{\beta}\bra{0}_{\s[v]} \rho \;\ket{0}\bra{\beta}_{\s[v]} \,\;+\;\; \ket{1 - \beta}\bra{1}_{\s[v]} \rho \;\ket{1}\bra{1 - \beta}_{\s[v]}	\;,
\end{gather}
which acts on the bit $\s[v]$ which stores the result of a previous measurement.\footnote{%
	Such ``shift'' operators are usually not noted in accounts of measurement-based computation (with~\cite{DKP07} being an exception), probably because they are simple operations on classical bits.
	However, in any implementation of quantum computers based on single-qubit measurements, shift operators may in practise be used to describe how parity expressions $\beta$ are computed in the midst of computations.
	Thus, while much of the analysis of this article follows the convention of eliminating shift operations where possible, it seems productive to extend the sense of ``standard forms'' found in the literature to include shift operations between measurements.}
This operation is also known as a \emph{signal shift} operation, as it can be used to remove the dependency on the result (or \emph{signal}) of the measurement of some qubit $a$ from the measurement of a later qubit $v$; the influence of the measurement on $a$ is then accumulated into the classical bit $\s[v]$.
As a result, shift operators $\Shift{v}{\beta}$ only occur after the measurement of the corresponding qubit $v$.

\subsubsection{Elementary transformations of \oneway\ procedures}

A typical \oneway\ procedure is a composition of many operations as described above.
In composing them, we may simplify them through simple transformations of the operations involved.
At the same time, we wish to bring these expressions into a canonical form in which the entangling and measuring operations are separated into different phases of the computation (one of the motivations of the \oneway\ model being to separate the tasks of creating entanglement, and maintaining control of the states of individual qubits).
We now describe these transformations.

A well-formed \oneway\ procedure prepares each qubit at most once, in which case that preparation is the first operation performed.
We may then accumulate all preparation operations $\New{v}$ to the beginning of a measurement-based procedure.\footnote{%
 	Composition is performed from right to left, as with standard notations for unitary circuits; thus, the beginning of a \oneway\ procedure is on the right-hand side of a sequence of operations.}
Thus, for an operation $\Op{S}{}$ acting on a set of qubits $S$, we have
\begin{align}
	\label{eqn:commutePrep}
		\New{v} \Op{S}{}	\,=&\;	\Op{S}{} \New{v}	\;,
	&&
		\text{when $v \notin S$}.
\end{align}
(When $v \in S$, the expression on the left-hand side is not well-defined.)
We may collect these preparation maps together and describe them as a single operation on many qubits: to this effect, we define the shorthand
\begin{gather}
 		\New{v_N} \cdots \New{v_2} \New{v_1}
	\;\;=\;\;
		\New{\ens{v_1, \ldots, v_N}}	\;.
\end{gather}
Similarly, because every operation aside from the preparation maps and entangling maps are either measurements (which are the last operation performed on any qubit), Pauli corrections (which we may easily transform by commutation with $\cZ$, a Clifford operation), or shift operations (which do not act on qubits), we may apply the following commutation operations to accumulate all of the entangling maps to the beginning of a \oneway\ procedure (just after the preparation maps):
\begin{subequations}
\label{eqn:commuteEnt}
\begin{align}
		\label{eqn:commuteEntMeas}
		\Ent{vw} \Meas{p}{\theta;\beta}
	\,=&\;
		\Meas{p}{\theta;\beta} \Ent{vw}	\;,
	&&
		\text{when $p \notin \ens{v,w}$};
	\\
		\label{eqn:commuteEntShift}
		\Ent{vw} \Shift{p}{\beta}
	\,=&\;
		\Shift{p}{\beta} \Ent{vw}	\;,
	&&
		\text{when $p \notin \ens{v,w}$};
	\\
		\label{eqn:induceZcorr}
		\Ent{vw} \Xcorr{v}{\beta}
	\,=&\;
		\Zcorr{w}{\beta} \Xcorr{v}{\beta} \Ent{vw}	\;,
	&&
		\text{and similarly for $\Ent{vw} \Xcorr{w}{\beta}$};
	\\
		\Ent{vw} \Zcorr{v}{\beta}
	\,=&\;
		\Zcorr{v}{\beta} \Ent{vw}	\;,
	&&
		\text{and similarly for $\Ent{vw} \Zcorr{w}{\beta}$}.
\end{align}
(When $p \in \ens{v,w}$, the expression on the left-hand side of~\eqref{eqn:commuteEntMeas} is not well-defined, and the bit $\s[p]$ on the left-hand side of~\eqref{eqn:commuteEntShift} does not have a defined value.)
As $\cZ$ is a diagonal operation, we also have
\begin{align}
	\Ent{uv} \Ent{vw} \,=&\;	\Ent{vw} \Ent{uv}	\;.
\end{align}
\end{subequations}
As with preparation maps, we may collect these entangling maps together and describe them as a single operation on many qubits.
Controlled-$Z$ operators act symmetrically on pairs of qubits: in stable-index notation (\autoref{sec:stableIndexNotation}), we have $\cZ\pseu[v,w]\! = \cZ\pseu[w,v]\!$.
Then the order of the entangling operations and the order of the qubits acted on for any single entangling map is unimportant.
As $\cZ$ is self-inverse, we have $\Ent{uv} \Ent{uv} = \idop$; then we may represent a product of entangling maps simply by an \emph{entanglement graph}
\begin{gather}
		\Ent{G}
	\;\;=
		\prod_{uv \in E(G)}	\! \Ent{uv}	\;,
\end{gather}
where the edge-set $E(G)$ consists of all of the pairs of qubits on which entangling maps act.

We may also consider how to commute correction operations past other operations.
First, note that commuting $X^\beta$ and $Z^\gamma$ operations on a state vector accrues at most a sign factor.
Then, the corresponding CPTP maps commute:
\begin{gather}
 		\Xcorr{u}{\beta} \Zcorr{v}{\gamma}
	\;=\;
 		\Zcorr{v}{\gamma} \Xcorr{u}{\beta}	\,.
\end{gather}
Similarly, $\Xcorr{}{}$ operations on any two qubits commute, as do $\Zcorr{}{}$ operations.
We may accumulate corrections of the same type on a common qubit:
\begin{align}
			\Xcorr{u}{\beta} \Xcorr{u}{\beta'}
		=&\,
			\Xcorr{u}{\beta + \beta'}	\,,
		&
			\Zcorr{u}{\gamma} \Zcorr{u}{\gamma'}
		=&\,
			\Zcorr{u}{\gamma + \gamma'}	\,.
\end{align}
With respect to measurement operations, corrections on a qubit $v$ may be interpreted as (classically controlled) changes of basis for the measurement on $v$; then, we can consider ``absorbing'' correction operations on qubits to be measured into the measurement operations.
Consider a measurement in the $\ket{+_\theta}, \ket{-_\theta}$ basis for some $-\pi < \theta \le \pi$.
We have
\begin{gather}
		X \ket{+_\theta}
	\;=\;\,
		\sfrac{1}{\sqrt 2}
			\Big( \e^{i\theta} \ket{0} + \ket{1} \Big)\!
	\,\;\propto\;\,
		\sfrac{1}{\sqrt 2}\Big( \ket{0} + \e^{-i\theta} \ket{1} \Big)
	\;=\;\,
		\ket{+_{(\minus\:\!\theta)}}	,
\end{gather}
and similarly $X \ket{-_\theta} \propto \ket{-_{(\minus\:\!\theta)}}$.
We may then incorporate $\Xcorr{v}{\beta}$ corrections into the basis of a subsequent measurement on $v$ via the equation
\begin{gather}
		\Meas{v}{\theta} \Xcorr{v}{\beta}
	\;\,=\;\,
		\Meas{v}{(\minus 1)^\beta \cdot\:\! \theta}
	\;\,=\;\,
		\Meas{v}{\theta;\beta}
\end{gather}
for a boolean expression $\beta$, using the notation defined in \eqref{eqn:adaptedMeasurementNotation}: more generally, we have $\Meas{v}{\theta;\beta} \Xcorr{v}{\gamma} = \Meas{b}{\theta;\beta+\gamma}$.
We may describe the effect of $Z$ corrections on measurements in a similar fashion:
\begin{gather}
		Z \ket{+_\theta}
	\;=\;\,
		\sfrac{1}{\sqrt 2} \Big( \ket{0} - \e^{i\theta} \ket{1} \Big)
	\;\propto\;\,
		\ket{-_\theta}	\,,
\end{gather}
and similarly $Z \ket{-_\theta} = \ket{+_\theta}$.
For an arbitrary state $\rho_S$ on a set of qubits $S$ including $v$, we then have
\begin{subequations}
\begin{align}
		\bra{+_\theta} \Big[ Z_v \rho_S Z_v\herm \Big] \ket{+_\theta}_v
	\;\,=&\,\;
		\bra{-_\theta} \rho_S  \ket{-_\theta}_v	\;,
	\quad
		\text{and}
	\\
		\bra{-_\theta} \Big[ Z_v \rho_S Z_v\herm \Big] \ket{-_\theta}_v
	\;\,=&\,\;
		\bra{+_\theta} \rho_S  \ket{+_\theta}_v	\;.
\end{align}
\end{subequations}
Thus, $Z$ corrections effectively exchange the roles of the states $\ket{\pm_\theta}$ for measurements made in that basis.
As this exchange is independent of any classically-controlled change of sign in $\theta$, we then obtain
\begin{gather}
		\Meas{v}{\theta;\beta}	\Zcorr{v}{\gamma}
	\;\,=\;\,
		\Shift{v}{\gamma}	\Meas{v}{\theta;\beta}
\end{gather}
for arbitrary boolean expressions $\beta$ and $\gamma$.
(That is, we may toggle the measurement result $\s[v]$ instead of performing a $Z$ correction on $v$ prior to measurement.)
For such a shift operator $\Shift{v}{\gamma}$ which occurs after a measurement $\Meas{v}{\ast}$ and before any operations using $\s[v]$ for classical control, we call $\gamma$ a \emph{bit-dependency} of $v$.\footnote{%
	This corresponds essentially to what are called $\pi$-dependencies in~\cite{BeaudPhD}, and to ``dependencies caused by $Z$-actions'' in~\cite{DKP07}.}

There are some further simplifications to measurement operations which are possible when the measurement angle $\theta$ is a multiple of 
$\pion 2$.
We refer to these as \emph{Pauli simplifications}, as such measurements are with respect to the eigenbases of the $X$ or $Y$ Pauli operators.
For $\theta \in \ens{0,\pi}$, we have $\ket{+_{(\minus\:\!\theta)}} = \ket{+_\theta}$ and $Z \ket{-_\theta} \propto \ket{-_{(\minus\:\!\theta)}}$; thus, we have
\begin{gather}
		\Meas{v}{m\pi;\beta}
	\;\,=\;\,
		\Meas{v}{m\pi}	\;,
\end{gather}
so that we may ignore sign-dependencies in this case.
For $\theta \in \ens{\mpion2,\pion2}$, we have $X \ket{+_\theta} \propto Z \ket{+_\theta}$ and similarly for $\ket{-_\theta}$, so that
\begin{gather}
		\Meas{v}{\pion[\tpm]2;\beta}
	\;\,=\;\,
		\Meas{v}{\pion[\tpm]2}	\Xcorr{v}{\beta}
	\;\,=\;\,
		\Meas{v}{\pion[\tpm]2}	\Zcorr{v}{\beta}
	\;\,=\;\,
		\Shift{v}{\beta}	\Meas{v}{\pion[\spm]2}	\;.
\end{gather}
Thus, measurements in so-called Pauli bases may be performed without dependencies on previous measurements.
In summary, we have the following relations which allow correction operations to be commuted past or absorbed into measurements:
\begin{subequations}
\label{eqn:absorbMeas}%
\begin{align}
	\label{eqn:commuteMeasA}%
		\Meas{v}{\theta;\beta} \Xcorr{u}{\gamma}
	\;=&\;\,
		\Xcorr{u}{\gamma} \Meas{v}{\theta;\beta}	\;;
	\\
	\label{eqn:commuteMeasB}%
		\Meas{v}{\theta;\beta} \Xcorr{u}{\gamma}
	\;=&\;\,
		\Xcorr{u}{\gamma} \Meas{v}{\theta;\beta}	\;,
	&&
		\mspace{-100mu}\text{when $u \ne v$};
	\\
	\label{eqn:commuteMeasC}%
		\Meas{v}{\theta;\beta} \Zcorr{u}{\gamma}
	\;=&\;\,
		\Zcorr{u}{\gamma} \Meas{v}{\theta;\beta}	\;,
	&&
		\mspace{-100mu}\text{when $u \ne v$};
	\\
	\label{eqn:absorbMeasX}%
		\Meas{v}{\theta;\beta} \Xcorr{v}{\gamma}
	\;=&\;\,
		\Meas{v}{\theta;(\beta + \gamma)}	\;;
	\\
	\label{eqn:absorbMeasZ}%
		\Meas{v}{\theta;\beta} \Zcorr{v}{\gamma}
	\;=&\;\,
		\Shift{v}{\gamma}	\Meas{v}{\theta;\beta}	\;;
	\\
		\label{eqn:pauliSimplifX}
		\Meas{v}{m\pi;\beta}
	\;=&\;\,
		\Meas{v}{m\pi}	\;;
	\\
		\label{eqn:pauliSimplifY}
		\Meas{v}{\pion[\tpm]2;\beta}
	\;=&\;\,
		\Shift{v}{\beta}	\Meas{v}{\pion[\tpm]2}	\;.
\end{align}
\end{subequations}
We will occasionally be interested in describing transformations of \oneway\ procedures ``without Pauli simplifications'', in which case we do not use the relations given in~\eqref{eqn:pauliSimplifX} and~\eqref{eqn:pauliSimplifY}.

Finally, we consider the transformations which are possible by commuting shift operators.
We may accumulate multiple signal shifts acting on a single measurement result:
\begin{gather}
	\label{eqn:accumShifts}
		\Shift{v}{\beta} \Shift{v}{\gamma}
	\;=\;
		\Shift{v}{\beta + \gamma}	\;,
\end{gather}
where this sum is evaluated modulo $2$.
If desired, we may also eliminate shift operations entirely from a \oneway\ procedure, by incorporating the bit-flip which is performed on the measurement signal $\s[v]$ into every expression which depends on $\s[v]$.
This has the effect of increasing the complexity of the classical control expressions, but when performed for all measurement results $\s[v]$, also makes explicit the dependency of any operation on prior measurement results.
For an arbitrary operation $\Op{}{f(\ldots,\s[v],\ldots)}$ which is classically controlled by some function $f$ taking $\s[v]$ as a parameter, we have
\begin{gather}
	\label{eqn:signalShift}
		\Op{}{f(\ldots,\s[v],\ldots)} \Shift{v}{\beta}
	\;=\;
		\Shift{v}{\beta} \Op{}{f(\ldots,\s[v] + \beta,\ldots)} \;;
\end{gather}
that is, we substitute $\s[v] + \beta$ wherever $\s[v]$ occurs, again evaluated modulo $2$.
We may use this commutation relation to accumulate all shift operators at the right-hand side of a \oneway\ procedure, after every operation acting on qubits: we refer to this as \emph{signal shifting}.
Conventionally, we may discard any classical measurement results when they are no longer required; thus, at the right-hand end of a \oneway\ procedure, we may remove any signal shift operators which remain.

\subsubsection{Standard forms for \oneway\ procedures}

Using the simplifications above, we may transform \oneway\ procedures into ``standard'' orderings of the operations which are comparatively convenient for discussion, and are also of low operational depth: we now describe two different such forms.

As we have noted above, in any well-formed \oneway\ procedure, it is possible to commute all preparation and entangling maps to the beginning of the procedure (on the right), and to propagate all correction operations to the end (on the left), possibly absorbing corrections into measurement operations while doing so.
Some of the corrections will also induce shift operations just after the measurement into which they are absorbed; and we may do this in a uniform way without explicit knowledge (or despite uncertainties in) the values of the default measurement angles.
This motivates the following:
\begin{definition}
	\label{def:standardize}
	A \oneway\ procedure $\mathfrak P$ is said to be in \emph{standard form}\footnote{%
		This terminology is similar to, but differs slightly from, the usage in~\cite{DKP07} and~\cite{BeaudPhD}.}
	if it can be decomposed as
	\begin{gather}
	 		\mathfrak P
		\;=\;
			\sfC \circ \Meas{}{} \circ \Ent{G} \circ \New{P}	\;,
	\end{gather}
	where \,$\New{P}$ is a product of preparation maps acting on some set of qubits $P$, \,$\Ent{G}$ is a product of entangling operations acting on pairs $uv$ connected by an edge in the graph $G$, $\Meas{}{}$ consists entirely of measurement operations and shift operations (where any shift operation $\Shift{v}{\ast}$ immediately follows the measurement $\Meas{v}{\ast}$), and $\sfC$ consists entirely of correction operations.
	The procedure of transforming a \oneway\ procedure into such a form using equations \eqref{eqn:commutePrep}, \eqref{eqn:commuteEnt}, 
	and \mbox{\eqref{eqn:commuteMeasA}--\eqref{eqn:absorbMeasZ}} is called \emph{standardization}.
\end{definition}
Note in particular that the process of ``standardization'' does not apply Pauli simplifications or signal shifting.
For practical applications, it is probably preferable to perform Pauli simplifications, at least; but for the purposes of discussion it is useful to consider the case where such simplifications of expressions for \oneway\ procedures are not applied.

By applying Pauli simplifications on measurements and signal shifting, it is possible to have equivalent \oneway\ procedures which are each in standard form.
Rather than attempt to characterize when such expressions are equal, it is useful to consider a form in which any variation of measurement dependencies has been removed.
We may define such a form as follows:
\begin{definition}
	\label{def:normalForm}
	A \oneway\ procedure $\mathfrak P$ is said to be in \emph{normal form} if it is in standard form, does not contain any shift operations, and any measurement operator with an angle $\theta \in \frac{\pi}{2} \Z$ has no measurement dependencies.
	We call the a procedure $\fP$ in normal form the \emph{normalization} of another $\fP'$ if we may obtain $\fP$ from $\fP'$ by standardizing $\fP'$, applying Pauli simplifications \mbox{\eqref{eqn:pauliSimplifX}--\eqref{eqn:pauliSimplifY}}, and then performing signal shifting~\eqref{eqn:signalShift}, removing the shift operators from the left after doing so.
\end{definition}

While it may not be desirable to compute such expressions for the practical application of \oneway\ procedures, obtaining the normalizations of \oneway\ procedures will be useful for the analysis of this article.

\subsubsection{Composability of \oneway\ procedures}

To perform universal computation, it is necessary to \emph{compose} \oneway\ computations: that is, to apply a second \oneway\ procedure $\fP_2$ to the output system produced by another \oneway\ procedure $\fP_1$.
Therefore, when composing \oneway\ procedures $\mathfrak P = \mathfrak P_2 \circ \mathfrak P_1$, we will be interested in restricting the qubits which are allocated in $\mathfrak P_2$ to be distinct from any qubits involved in $\mathfrak P_1$, except those which $\fP_1$ produces as output.

It will be convenient at this point to introduce a combinatorial structure which partially describes a \oneway\ procedure in order to define composability conditions (as well as to visually convey the structures which will be important in later sections).
For a \oneway\ procedure $\fP$, we may consider the information contained in $\fP$ which is independent of ordering or of the particular measurement angles or outcomes:
\begin{definition}
	\label{def:geometry}
	The \emph{geometry underlying a \oneway\ procedure $\mathfrak P$} is a triple $(G,I,O)$, where $G$ is the entanglement graph on the qubits of $\fP$ specified by the entangling maps $\Ent{vw}$ in $\fP$, $I \subset V(G)$ is the set of qubits which are not prepared by $\fP$, and $O \subset V(G)$ is the set of qubits which are not measured by $\fP$.
	We refer to $I$ and $O$ as the \emph{input subsystem} and \emph{output subsystem} of $\mathfrak P$, respectively.
	(When a given graph $G$ may be inferred from context, we may write $I\comp = V(G) \setminus I$ and $O\comp = V(G) \setminus O$ for the complements of these subsystems.)
\end{definition}
We will often illustrate geometries in the style shown in \autoref{fig:genericGeom}.
\begin{figure}[t]
	\begin{center}
		\includegraphics{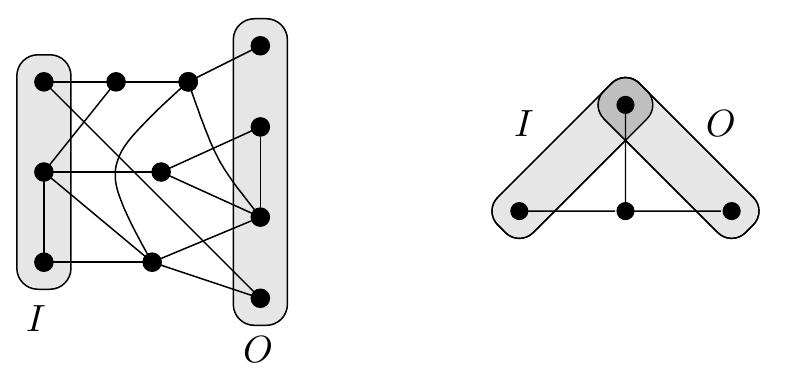}
 		\vspace{-1em}
	\end{center}
	\caption[Illustration of two geometries for \oneway\ procedures.]{\label{fig:genericGeom}%
		Illustration of two geometries for \oneway\ procedures.
		Labels for the vertices have been omitted.
		Note in particular that the sets $I$ and $O$ may overlap (indicated by the more darkly shaded region in the geometry on the right).}
\end{figure}
Such diagrams give a partial description of a \oneway\ procedure, in which each qubit $v \in I\comp$ is prepared, the entangling operations $\Ent{vw}$ implied by the edges $vw \in E(G)$ are performed, and the qubits $w \in O\comp$ are measured.
However, the measurement order, measurement angles and dependencies, and any correction/shift operations are left unspecified.

Using geometries, we can express conditions for the composability of two \oneway\ procedures in terms of a sense in which the geometries of these patterns may be ``composed'':
\begin{definition}
	For two geometries $\cG_1 = (G_1, I_1, O_1)$ and $\cG_2 = (G_2, I_2, O_2)$, we say that $\cG_1$ and $\cG_2$ are \emph{composable} if $\big[V(G_1) \inter V(G_2)\big] \subset \big[I_2 \inter O_1\big]$\,.
	The \emph{composition} $\bar\cG = \cG_2 \circ \cG_1$ is then defined by $\bar\cG = (\bar G, \bar I, \bar O)$, where
	\begin{subequations}
	\label{eqn:composibilityConditions}%
	\begin{align}
			V(\bar G)
		\,\;=&\;\;\;
			V(G_1) \union V(G_2)	\,,
		\\
			E(\bar G)
		\,\;=&\;\;\;
			E(G_1) \symdiff E(G_2)	\,,
		\\
			\bar I
		\;\;=&\;\;\;
			I_1 \,\,\union (I_2 \setminus O_1) \,,
		\\
			\bar O
		\,\;=&\;\;\;
			O_2 \union (O_1 \setminus I_2) \;,
	\end{align}
	\end{subequations}
	where $A \symdiff B = [A \setminus B] \union [A \setminus B]$ is the symmetric difference of $A$ and $B$.
	Two \oneway\ procedures $\mathfrak P_1$ and $\mathfrak P_2$ are \emph{composable} if their underlying geometries  $\cG_1$ and $\cG_2$ are composable; the underlying geometry of the procedure $\mathfrak P_2 \circ \mathfrak P_1$ is then $\cG_2 \circ \cG_1$.
\end{definition}
\autoref{fig:geomCompose} illustrates the composition of two geometries (and, by way of synecdoche, represents the composition of two \oneway\ procedures).
\begin{figure}[t]
	\begin{center}
		\includegraphics{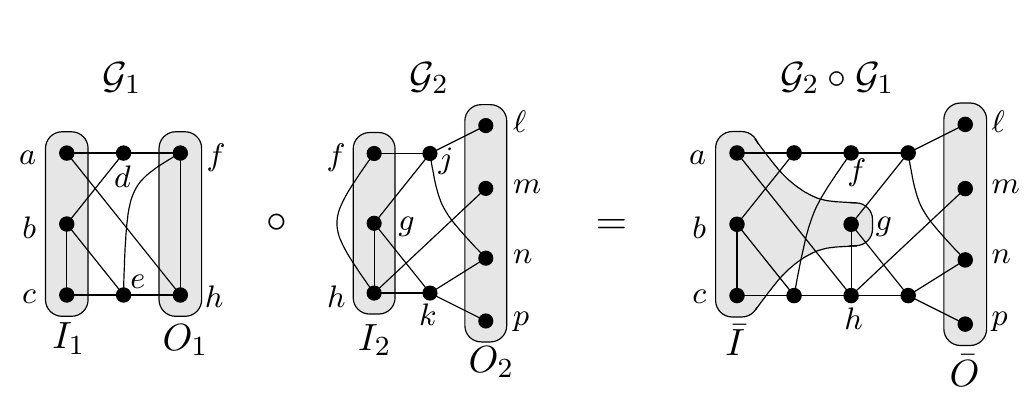}
		\vspace{-2ex}
	\end{center}
	\caption[Illustration of the composition of two generic geometries for \oneway\ procedures.]{\label{fig:geomCompose}%
		Illustration of the composition of two generic geometries for \oneway\ procedures.
		Notice that composition in the diagrams is left-to-right, while the composition in written notation is right-to-left.}
		\vspace{2ex}
\end{figure}

In order to obtain a composition of two \oneway\ procedures $\fP_2 \circ \fP_1$, we may relabel the qubits involved in $\fP_2$ in order for the conditions of \eqref{eqn:composibilityConditions} to hold.
From this point onwards, when we wish to compose two generic \oneway\ procedures $\mathfrak P_1$ and $\mathfrak P_2$, we will assume that they are composable in the sense above, eliding over the procedure of relabelling the qubits to be prepared which may be necessary.

\subsection{The DKP and simplified RBB constructions}
\label{sec:constructions}

Having established stable-index tensor notation (\autoref{sec:stableIndexNotation}) and a calculus of the operations used in a \oneway\ procedure (\autoref{sec:measurementCalculus}), we may describe two constructions for approximately universal quantum computation in the \oneway\ model.\footnote{%
	The presentation of such construction presented here differs from the approach of~\cite{DKP07}: while the latter proceeds by defining the composition of formal denotational semantics; we will adopt a somewhat more informal approach involving stable-index tensor notation.}
The first construction we will refer to as the \emph{DKP construction}, and is based on the work of Danos, Kashefi, and Panangaden~\cite{DKP07} for the construction of \oneway\ procedures in the absence of topological constraints on the interactions of qubits.
The second construction is derived from the elementary constructions among those presented by Raussendorf, Browne, and Briegel~\cite{RBB03} for producing \oneway\ procedures, under the topological constraint that the entanglement graph be an induced subgraph of a suitably large rectangular grid; we call this the \emph{simplified RBB construction}.\footnote{%
	The simplification implied in this title is twofold.
	Firstly, we use the $\Zz$ operation given in~\eqref{eqn:Zz} to describe the construction of arbitrary operations, rather than the $\textsc{cnot}$ operation as in \cite{RBB03}, page~5.
	While this will lead in practise to a constant-factor reduction in the complexity of some \oneway\ procedures, this construction is clearly implicit in~\cite{RBB03}, and so we may attribute this construction to that work.
	Secondly, the ``simplified'' RBB construction omits the special-purpose constructions exhibited in Sections~IV.A--F of~\cite{RBB03}, such as compact patterns for reversing an block of consecutive qubits, non-nearest neighbor $\textsc{cnot}$ operations, the quantum Fourier transform over $\Z_{2^k}$ for $k \ge 0$, and addition circuits for $\Z_{2^k}$.}

\subsubsection{Construction of the gate-sets $\cB_\DKP$ and $\cB_\RBB$}
\label{sec:constructGates}

The first step in both of the constructions we describe is to identify \oneway\ procedures for the gate-sets $\cB_\DKP = \ens{J(\theta), \cZ}_{\theta \in \frac{\pi}{4}\Z}$ and $\cB_\RBB = \ens{J(\theta), \Zz}_{\theta \in \frac{\pi}{4}\Z}$.
We shall do this in overview, leaving the details for \autoref{apx:constructGatesApx}.

By definition, the CPTP entangling map $\Ent{vw}$ performs controlled-$Z$ operations on pairs of qubits: for a system $S$ including qubits $v,w$, we have
\begin{gather}
	\label{eqn:cZpattFormula}
		\Ent{vw}(\rho_S)
	\;=\;
		\cZ_{vw} \;\rho_S\, \cZ_{vw}	\;.
\end{gather}
We may implement $\Zz$ and $J(\theta)$ operations in the \oneway\ model with somewhat more elaborate procedures:
\begin{gather}
	\label{eqn:ZzPattFormula}
		\fZz_{v,w}
	\;=\;
		\Zcorr{v}{\s[a]} \Zcorr{w}{\s[a]} \Meas{a}{\pion2} \Ent{av} \Ent{aw} \New{a}	\;,
	\\[1ex]
  	\label{eqn:JpattFormula}
		\fJ^\theta_{w/v}
	\;=\;
		\Xcorr{w}{\s[v]} \Meas{v}{-\theta} \Ent{vw} \New{w}	\;,
\end{gather}
where the subscript $w/v$ in the latter denotes that it discards $v$ from the input system and introduces a new qubit $w$.
It is possible to show that the effects of these two \oneway\ procedures are given by
\begin{align}
 	\label{eqn:ZzCompute}
		\fZz_{v,w}\big( \rho_S \big)
	\;=&\,\;
		\Zz_{v,w} \; \rho_S \; \Zz_{v,w}\herm	\;,
\end{align}
for arbitrary operators $\rho_S$ acting on a system including $v$ and $w$, and
\begin{align}
		\fJ^\theta_{w/v}\Big(\ket{\Psi}\bra{\Psi}_S\Big)
	\;=&\;\,
		J(\theta)\big._{w/v} \ket{\Psi}\bra{\Psi}_S J(\theta)\herm_{v/w}	\;,
\end{align}
for a system $S$ including $v$ but not including $w$, where by $J(\theta)_{w/v}$ we denote the linear operator $\cH[v] \to* \cH[w]$ which maps state-vectors $\ket{\psi}_v \,\mapsto\, \big[\:\! J(\theta) \ket{\psi} \:\!\big]_w$\,.
Thus, the procedures $\Ent{vw}$, $\fZz_{v,w}$, and $\fJ^\theta_{w/v}$ given in \eqref{eqn:cZpattFormula}, \eqref{eqn:ZzPattFormula}, and \eqref{eqn:JpattFormula} respectively may be used to represent each of the operations in $\cB_\DKP$ and $\cB_\RBB$.
The geometries underlying these three \oneway\ procedures are illustrated in \autoref{fig:primitiveGeom}.
\begin{figure}[t]
	\begin{center}
		\includegraphics{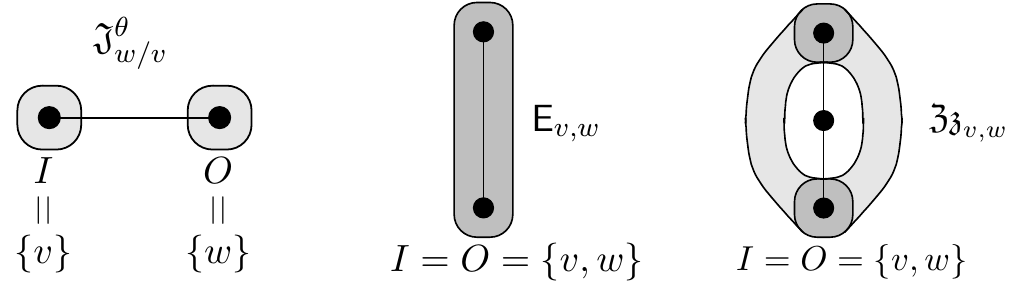}
		\vspace{-2ex}
	\end{center}
	\caption[Illustration of the geometries of two primitive \oneway\ patterns.]{\label{fig:primitiveGeom}%
		Illustration of the geometries of the three elementary \oneway\ procedures $\fJ_{w/v}^\theta$, $\Ent{vw}$, and $\fZz_{v,w}$ for generic qubits $v$ and $w$, given by \eqref{eqn:cZpattFormula}, \eqref{eqn:ZzPattFormula}, and \eqref{eqn:JpattFormula} respectively.
		Note that the input and output subsystems coincide for $\Ent{vw}$ and $\fZz_{v,w}$.}
		\vspace{2ex}
\end{figure}

\subsubsection{Corresponding stable-tensor indices to qubits in the \oneway\ model}
\label{sec:onewayStableIndexCorresp}

By the analysis of \autoref{sec:constructGates}, we have a simple correspondence between unitary circuit operations described in terms of stable-index tensor expressions, and simple \oneway\ procedures for effecting them: we have
\begin{subequations}
\label{eqn:onewayCorresp}
\begin{align}
		J(\theta) \pseu[:w/v]
	\:\!\to&\;\,
		\fJ^\theta_{w/v}	\;;
	&
		\cZ\pseu[v,w]
	\:\!\to&\;\,
		\Ent{vw}	\;;
	&
		\Zz\pseu[v,w]
	\:\!\to&\;\,
		\fZz_{v,w}	\;.
\end{align}
We may extend this correspondence also to account for the preparation of fresh qubits: as a preparation map $\New{v}$ simply prepares a qubit $v$ in the $\ket{+}$ state, we may also write
\begin{align}
 		\ket{+}\pseu[:v/]
	\,\to&\;\;
		\New{v}	\;.
\end{align}
\end{subequations}
These parallels between operations on indices in stable-index expressions, and CPTP maps operations on qubits in the \oneway\ model, also extend to syntactical considerations:
\begin{itemize}
\item 
	If we order terms in a stable-index expression from right to left according to an order in which the operations may actually be performed, by construction there cannot be any terms acting on an index $v$ to the right of an operation such as $J(\theta)\pseu[:w/v]$ which deprecates $v$; nor operations to the left of $v$ before an operation such as $J(\theta)\pseu[:v/u]$ in which $v$ is advanced.
	Similarly, in a \oneway\ procedure, operations acting on a \emph{qubit} $v$ are not permitted after it has been measured, nor before it has been prepared.

\item
	In a stable-index expression for unitary circuits, stable indices (which are neither advanced nor deprecated) are those indices representing the a qubit whose configuration over the standard basis 
	are unchanged by the operator acting on it.
	Similarly, in a \oneway\ procedure, a qubit which is neither added to the system nor discarded is in both the input and output subsystems of the procedure, and so is only acted on by diagonal operations (setting aside possible $\Xcorr{}{}$ operations
),
	which similarly do not affect the diagonal terms of their reduced density operators.
\end{itemize}
Thus, the role of tensor indices (in a stable-index expression) in describing quantum states as distributions over the standard basis is analogous to the role of qubits in a \oneway\ procedure which actually support quantum states.
We may then obtain \oneway\ procedures for unitary embeddings $U$ by simply applying the correspondences in \eqref{eqn:onewayCorresp} to a stable-index expression of a unitary circuit.
While the the constructions of~\cite{DKP07} or~\cite{RBB03} are not presented in precisely such terms, the DKP construction and simplified RBB constructions may be described as transliterating a sequence of gates in a unitary circuit into \oneway\ procedures of the form $\fJ^\theta$, $\Ent{}$, and $\fZz$, yielding a \oneway\ procedure for the circuit as a whole.
We present an explicit algorithm to do so in \autoref{sec:constructions}.

In the case of the DKP construction, in which the qubits of the resulting procedures are not subject to any topological constraints, application of the correspondence in \eqref{eqn:onewayCorresp} is an essentially complete summary of the construction.
The case of the simplified RBB construction is similar, but is subject to topological constraints which we describe below.

\subsubsection{Topological constraints in the simplified RBB construction}
\label{sec:topologicalConstraints}

The simplified RBB construction is subject to the constraint that the entangled states which are used must be obtained from a \emph{cluster state}: that is, a state in which qubits involved are arranged in a two-dimensional grid, and two qubits $v,w$ are acted on with an entangling map $\Ent{vw}$ if and only if $v$ and $w$ are nearest neighbors in the grid.
We will describe circuit constructions in which the indices of the stable-index representation satisfies the same constraints as required for the qubits in the cluster-state-based variant of the \oneway\ model, which will allow us to employ the correspondence of \eqref{eqn:onewayCorresp} directly.

In the cluster-state-based (and original) variant of the \oneway\ model, the only tool provided for obtaining graphs \emph{other} than a rectangular grid is the removal of qubits from such a grid, via measurement in the $\ket{0},\ket{1}$ basis.
Temporarily extending the set of operations defined in~\autoref{sec:elemOperationsOneWay}, we may define a $\Meas{v}{\sfz}$ operation which performs a destructive single-qubit measurement of a qubit $v$ in the standard basis, and produces a bit $\s[v]$ storing the result in the obvious manner.
We may then show
\begin{gather}
		\Meas{v}{\sfz} \Ent{vw}
	\;\;=\;\;
		\Zcorr{w}{\s[v]} \Meas{v}{\sfz}	\;,
\end{gather}
which follows from $\ket{x}_{\!\s[v]} \bra{x}_v \cZ_{vw} =\, Z^{x}_w \ket{x}_{\!\s[v]} \bra{x}_v$ for $x \in \ens{0,1}$.
We may use this to remove all of the entangling maps $\Ent{vw}$ incident to a qubit $v$, propagating the measurement to the right until it is performed immediately after the preparation of $v$.
If we remove the redundant preparation and measurement of $v$ (and the corrections induced by that measurement), this procedure has the effect of removing the vertex $v$ from the grid.
This is illustrated in \autoref{fig:removeVtx}.
\begin{figure}[t]
	\begin{center}
		\includegraphics{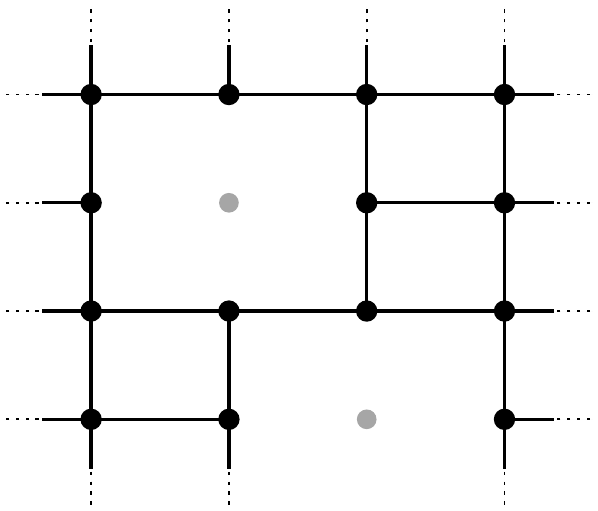}
		\vspace{-2ex}
	\end{center}
	\caption[Illustration of the effect of removing vertices from a two-dimensional grid.]{\label{fig:removeVtx}%
		Illustration of the effect of removing vertices from a two-dimensional grid.
		This corresponds to the effect of performing measurement of select qubits in a cluster state, in the standard basis.
		Vertices correspond to qubits, with vertices in grey corresponding to qubits which have been measured; edges (or absence of edges) correspond to the presence (resp. absence) of entangling maps on adjacent qubits.}
		\vspace{2ex}
\end{figure}

We may instead suppose (equivalently) that we may prepare any subset of the qubits in a two-dimensional lattice (omitting those qubits which we would subsequently remove as described above), and that two qubits are subject to an entangling map if and only if they are nearest neighbors in the grid.
That is: we will merely require that the entanglement graph $G$ is an induced subgraph of a rectangular lattice.
Given this restriction on \oneway\ procedures, we impose the corresponding restriction that the circuits which are accepted as input must have a linear nearest-neighbor topology. 

To represent products of single-qubit gates as a \oneway\ procedure, in the DKP construction as well as in the simplified RBB construction, we may compose patterns of the form $\mathfrak J_{w/v}^\alpha$ as in
\begin{gather}
	\mathfrak J_{v_N/v_{N \minus 1}}^{\alpha_{N \minus 1}} \cdots\;\, \mathfrak J_{v_2/v_1}^{\alpha_1} \, \mathfrak J_{v_1/v_0}^{\alpha_0}	\;.
\end{gather}
We may refer to such a procedure as a \emph{chain pattern}.
Note that in such a procedure, we require that the qubits $v_j$ for $0 < j < N$ to have degree $2$ in the corresponding entanglement graph, having no neighbors aside from $v_{j-1}$ and $v_{j+1}$\,.
In the simplified RBB construction, the degrees of these qubits must be ``enforced'' by a representation in the grid, where every neighbor of these qubits $v_j$ in the grid (for $0 < j < N$) are removed except for the neighbors $v_{j-1}$ and $v_{j+1}$. 
An illustration of such an embedding of a chain in the grid is illustrated in \autoref{fig:isolatedChain}.
\begin{figure}[t]
	\begin{center}
		\includegraphics{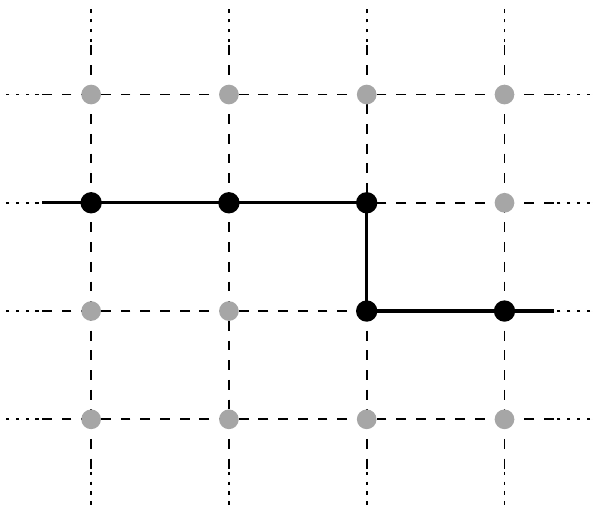}
	\end{center}
	\caption[Illustration of an embedding of a chain in the grid.]{\label{fig:isolatedChain}%
		Illustration of an embedding of a chain in the grid.
		Grey dots and broken lines represent vertices and edges which are removed from the grid by vertex deletions, to produce a graph containing the illustrated sequence of vertices with degree $2$.
	}
\end{figure}
(Typically, such chain patterns would be embedded as a horizontal path through the grid, but as \autoref{fig:isolatedChain} also shows, we may also employ more general paths in the grid.)

Two logical qubits in a circuit which do not interact must be represented by (chains of) qubits in the grid which are separated by at least two edges; otherwise there will be entangling relations between them.
Thus, we will represent products of single-qubit operations on independent qubits as non-adjacent horizontal rows of the grid.
We may then use a vertical column of three qubits (as illustrated in \autoref{fig:primitiveGeom}) to implement a $\Zz$ operation between these non-adjacent rows, using a single qubit between the two rows to mediate the interaction.
However, there are two technicalities which must be addressed:

\paragraph{Adjusting for uneven chain lengths.} 

Consider a decomposition of single-qubit operators $U$ and $V$ on distinct qubits, as products of $J(\theta)$ operators.
If the number of terms in the two products differ, the ends of the corresponding chain patterns may lie in different columns.
To operate on two chains with a $\Zz$ operation as described above, we require that the ends of the chains representing each qubit lie in the same column.
To achieve this, we may extend each chain by two or more vertices, as necessary, representing a decomposition of $\idop$ as a product of $J(\theta)$ gates.
We may define \oneway\ procedures
\begin{align}
	\label{eqn:identityPatterns}
		\mathfrak {Id}^2_{v_2/v_0}
	\;=&\;\,
		\mathfrak J_{v_2/v_1}^{\,0} \, \mathfrak J_{v_1/v_0}^{\, 0}		\;,
	&
		\mathfrak {Id}^3_{v_3/v_0}
	\;=&\;\,
		\mathfrak J_{v_3/v_2}^{\pion 2} \, \mathfrak J_{v_2/v_1}^{\pion 2} \, \mathfrak J_{v_1/v_0}^{\pion 2}	\;;
\end{align}
the geometries underlying these procedures are illustrated in \autoref{fig:idopChain}.
\begin{figure}[t]
	\begin{center}
	\vspace{1em}
		\includegraphics{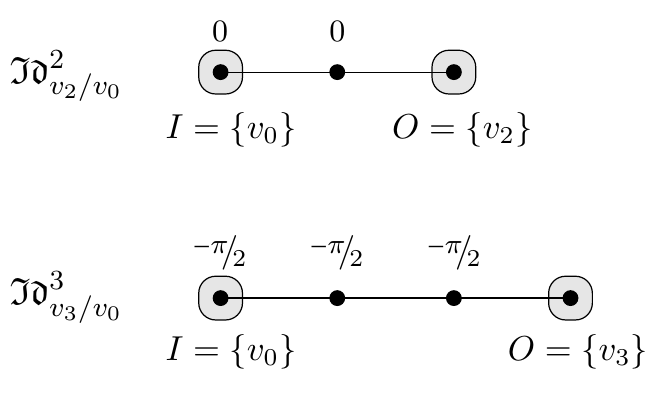}
	\vspace{-1em}
	\end{center}
	\caption[Illustration of an embedding of a chain in the grid.]{\label{fig:idopChain}%
		Illustration of the geometries underlying the procedures $\mathfrak{Id}^2$ and $\mathfrak{Id}^3$ as defined in	\eqref{eqn:identityPatterns}.
		The default measurement angles for qubits in $O\comp$ are indicated above each qubit.
	}
\end{figure}
In terms of the correspondence between stable-index expressions and \oneway\ procedures described in \autoref{sec:onewayStableIndexCorresp}, we have
\begin{subequations}
\begin{align}
		\mathfrak {Id}^2_{v_2/v_0}
	\;\to&\;\,
		J(0)\pseu[:v_2/v_1] J(0)\pseu[:v_1/v_0] \;,
	\\[1ex]
		\mathfrak {Id}^3_{v_3/v_0}
	\;\to&\;\,
		J(\pion 2)\pseu[:v_3/v_2] J(\pion 2)\pseu[:v_2/v_1] J(\pion 2)\pseu[:v_1/v_0] \;. 
\end{align} 
\end{subequations}
One can verify the stable-index products given on the right evaluate to $\idop\pseu[:v_2/v_0]$ and $\idop\pseu[:v_3/v_0]$ respectively; then both \oneway\ procedures above simply transmit the state presented as input to its respective output system, using a chain of length either $2$ or $3$.
By composing these procedures repeatedly, we can then implement a path of any length $\ell \ge 2$ in the grid, representing a sequence of operations which performs the identity $\idop_2$ on a given qubit.
Thus, whenever required, we may suppose that the paths in the grid corresponding to any two qubits are of the same length, regardless of the number of (non-trivial) unitary transformations performed on them.

\paragraph{``Forbidden'' qubits adjacent to the mediating qubit in $\Zz$.}
Just as we require a distance of $2$ between qubits in a \oneway\ procedure representing non-interacting qubits in a unitary circuit, we require that the qubit $a$ which mediates a $\Zz_{vw}$ operation is adjacent only to $v$ and $w$; that is, we require that the qubits immediately to the left and right in the grid are not in the entanglement graph.
Then we require a distance of two edges between any two instances of $\fZz$ performed on the same pair of rows. 
For instance, to represent a circuit $\Zz\pseu[v_f, w_f] J(\beta)\pseu[:w_f/w_i] J(\alpha)\pseu[:v_f/v_i] \Zz\pseu[v_i, w_i]$ on a pair of logical qubits $v$ and $w$, in which two $\Zz$ operations are separated only by a layer of single-qubit operations of depth $1$, we may perform the substitution
\begin{align}
		\Zz\pseu[v_f, w_f] J(\beta)\pseu[:w_f/w_i] & J(\alpha)\pseu[:v_f/v_i] \Zz\pseu[v_i, w_i]
	\notag\\=\;\;
		\Zz\pseu[v_f, w_f]	& J(0)\pseu[:w_f/w_2] J(0)\pseu[:w_2/w_1] J(\beta)\pseu[:w_1/w_i]
	\notag\\						& J(0)\pseu[:v_f/v_2] J(0)\pseu[:v_2/v_1] J(\alpha)\pseu[:v_1/v_i] \Zz\pseu[v_i, w_i]	\;.
\end{align}
Mapping each $J(\theta)\pseu[:a/d]$ term to a procedure $\fJ^\theta_{a/d}$ and each $\Zz\pseu[s,t]$ term to a procedure $\fZz_{s,t}$\,, the distance between the qubits $v_i$ and $v_f$ resulting in the right-hand side is then at least $2$ (and similarly for $w_i$ and $w_f$), as required.

\vspace{\baselineskip}
These observations describe ways in which a \oneway\ procedure for performing transformations must be ``padded'' by operations which perform the identity, in order to simulate a unitary circuit in the \oneway\ model.
We illustrate these observations in \autoref{fig:zzCompose} as sample constructions of \oneway\ procedures.
\begin{figure}[pt]
	\begin{center}
		\includegraphics[scale=0.8]{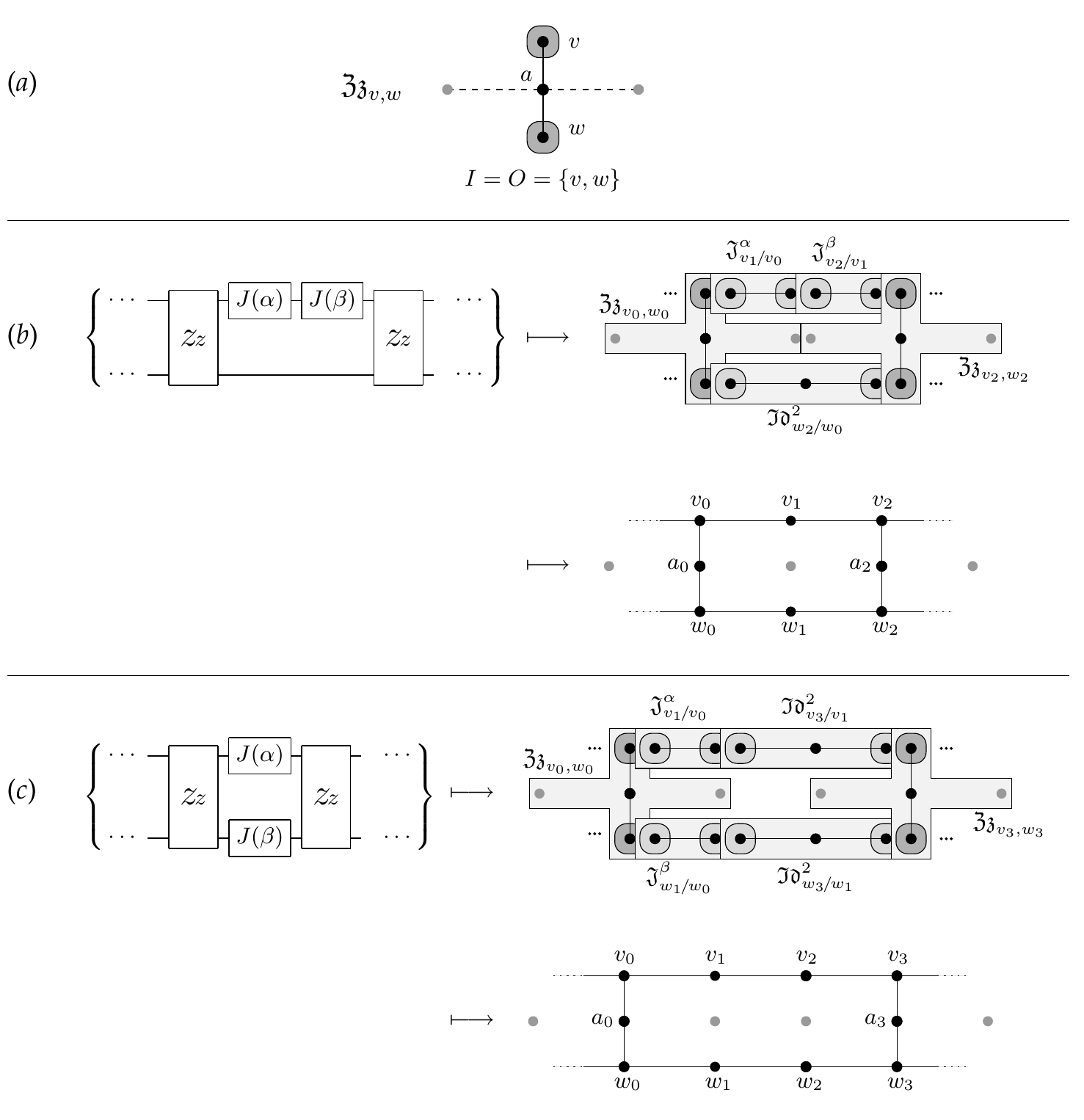}
 	\vspace{-0.5em}
	\end{center}
	\caption[\ldots]{%
		Illustration of the constraints on composition of the pattern $\fZz$ in the setting of an induced subgraph of the grid.
		Grey dots represent vertices which have been removed from the grid.
		Composition is illustrated with overlapping tiles, whose shapes represent spatial constraints imposed by the operations.
	
		\vspace{1ex}\bgroup
		\addtolength\textwidth{-0.5em}
		\begin{small}
		\begin{subfiglist}
		\item
			The geometry itself, together with the ``forbidden'' neighbors of the mediating qubit $a$.
			(Broken lines between the dots represent the edges removed as a result of removing these vertices.)
		\item
			A composition of $\fZz$ procedures with $\fJ^\theta$ procedures, corresponding to a circuit in which two $\Zz$ gates are separated by two $J(\theta)$ operations.
			The bottom row of the corresponding \oneway\ procedure is padded with a $\mathfrak {Id}^2$ procedure to synchronize the lengths of the rows.
			In this case, a single vertex is a common forbidden neighbor of each mediating qubit.
		\item
			\label{item:subfigPadBoth}
			A composition of $\fZz$ procedures with $\fJ^\theta$ procedures, corresponding to a circuit in which two $\Zz$ gates are separated by one $J(\theta)$ operation on either wire.
			The forbidden neighbors of the mediating qubits force both rows to be padded with identity operations.
		\end{subfiglist}
		\end{small}
		\egroup
	}
	\addtocounter{figure}{-1}\refstepcounter{figure}\label{fig:zzCompose}
	\vspace{2em}
\end{figure}
However, in order to obtain a uniform description of both the DKP and simplified RBB constructions, we may instead pad the stable-index representations for circuits with the corresponding products of single-qubit unitaries as a pre-processing stage.
By construction, the indices of the resulting stable-index expressions will then satisfy the same constraints which are required for the qubits in a cluster-based \oneway\ computation.
We may perform such padding as part of an alternative procedure for constructing stable-index expressions for circuits: such a procedure is described in \autoref{apx:stableIndexConstructionConstrained}.

\subsubsection{Unified description of the DKP and simplified RBB constructions}
\label{sec:constructionAlgorithms}

Using the above techniques for the simplified RBB construction, we may obtain a unified procedure describing both the DKP and simplified RBB constructions, as follows.

\paragraph{Procedure for both constructions.}

Suppose we are given a unitary circuit $C$ over the gate-set $\ens{H,T,\cZ}$\,: as described in the previous section, we further require that $C$ have a linear nearest-neighbor topology in the case of the simplified RBB construction.

\vspace{1em}\begin{descenum}
\item[Obtain a circuit in normal form.]\hfill\\
	Given a unitary circuit $C$ generated over the gate-set $\ens{H, T, \cZ}$, we obtain an equivalent circuit as follows.
	We cancel pairs of $\cZ$ gates which act on common pairs of wires, and which are separated only by gates with which they commute (\ie\ $T$ operations and operations acting on other wires).
	In the case of the simplified RBB construction, we convert the circuit to the (essentially equivalent) gate-set $\ens{H,T,\Zz}$ by applying the substitution
	\begin{align}
		\label{eqn:substZzForCz}
			\cZ
		\,=&\,\,
			\Zz (T\herm \ox T\herm)^2
	\end{align}
	to all $\cZ$ gates; no such translation is required for the DKP construction.
	We then commute all $T$ and $T\herm$ operations to the left (simplifying them using the identity $T^8 = \idop$), collecting them in each case either at the end of the circuit, or immediately preceding a Hadamard operation the same qubit.

	We call the resulting circuit $C'$ the \emph{normal form} of $C$; this circuit consists of a product of operations $\cZ$ and $H T^r$ (for various $r \in Z$) for the DKP construction, and $\Zz$ and $H T^r$ for the simplified RBB construction.

\item[Convert the circuit to a stable-index expression in the gate-set $\cB_\DKP$ or $\cB_\RBB$.]\hfill\\
	Having obtained $C'$, we apply the substitutions
	\begin{align}
	 		H T^m
		\;=&\;\,
			J(\tfrac{m\pi}{4})	\;,
		&
			T^m
		\;=&\;\,
			J(0) J(\tfrac{m\pi}{4})
	\end{align}
	for each consecutive block of $T$ operations, applying the right-hand equality for blocks of $T$ operations at the end of the circuit.
	These substitutions yield an equivalent 
	representation of the same operation using the gate-set $\cB_\DKP$ or $\cB_\RBB$, depending on whether the circuit makes use of $\cZ$ or $\Zz$ gates.
	As we do so, we convert the resulting circuit into stable-index notation: we use an algorithm such as that of \autoref{apx:stableIndexConstructionUnconstrained} to do so for the DKP construction; for the simplified RBB construction, we may use instead the procedure of \autoref{apx:stableIndexConstructionConstrained} to impose additional topological constraints on the resulting circuit.
	We call the resulting stable-index expression for the circuit $C''$.

\item[Translate individual operations into the \oneway\ model.]\hfill\\
	In $C''$, note that each single-qubit $J(\theta)$ operation deprecates one index and advances another, and that all of the two-qubit $\cZ$ or $\Zz$ operations act ``stably'' on their indices.
	We may then define a mapping $\Phi$ from stable-index expressions of operations in the gate-set $\cB_\DKP \union \cB_\RBB$ (and introduction of fresh qubits in the $\ket{+}$ state) to \oneway\ procedures involving $\fJ^\theta$, $\Ent{}$, and $\fZz$ operations as follows:
	\begin{subequations}
	\label{eqn:jointHomphm}%
	\begin{gather}
			\Phi\Big(\,	\ket{+}\pseu[:v/]	\Big)
		\;=\;\,
			\New{v} 	\;;
		\\[1ex]
			\Phi\Big(\,	J(\theta)\pseu[:w/v]	\Big)
		\;=\;\,
			\fJ^\theta_{w/v}	\;;
		\\[1ex]
			\Phi\Big(\,	\cZ\pseu[v,w]	\Big)
		\;=\;\,
			\Ent{vw}	\;;
		\\[1ex]
			\Phi\Big(\, \Zz\pseu[v,w]	\Big)
		\;=\;\,
			\fZz_{v,w}	\;,
	\end{gather}
	\end{subequations}
	where we identify the labels of indices of the stable-index tensor expression with labels for qubits in the \oneway\ procedure.
	We then extend this map $\Phi$ homomorphically to products of such terms.

\item[Obtain a normal form for the resulting \oneway\ procedure.]\hfill\\
	We standardize (\autoref{def:standardize}) the \oneway\ procedure $\Phi(C'')$, and then apply Pauli simplifications and signal shifting to obtain a procedure in normal form (\autoref{def:normalForm}).
\end{descenum}
By construction of $C''$ in the above procedure, each index is advanced at most once and deprecated at most once; thus, each qubit in $\Phi(C'')$ is produced in the output of an operation (or ``allocated'') at most once, and removed from the input of an operation (or ``discarded'') at most once.
Furthermore, by order of the terms in $C''$, the operation which allocates a qubit $v$ is the first which acts on $v$, and the operation which discards $v$ is the last acting on $v$.
The resulting \oneway\ procedure is therefore well-formed.

By the analysis of \autoref{sec:constructGates}, we may easily show in either case that $\Phi( C'' )$ is a CPTP map which performs the same unitary transformation as described by $C''$, modulo an identification of each input qubit with a corresponding output qubit.
Then the final procedure produced performs the same transformation as the circuit $C$.
We refer to this procedure as the \emph{simplified RBB construction} for \oneway\ procedures when we require that $C$ be linear nearest-neighbor and when we perform the substitution of \eqref{eqn:substZzForCz}, and \emph{the DKP construction} for \oneway\ procedures when we do not impose the added constraint or perform the substitution.
(For the sake of brevity, we will often refer to this construction simply as ``the RBB construction'', as we do not consider any constructions using the special-purpose \oneway\ procedures in Section~IV of~\cite{RBB03}.)

\paragraph{Run-time complexity.}

Both the DKP construction and the (simplified) RBB construction can be performed efficiently in the size of the input circuit.
In both cases, let $k$ be the number of qubits which $C$ acts upon, $N$ be the number of one-qubit gates, and $M$ be the number of two-qubit gates: if we consider only those circuits which act non-trivially on each qubit, we may assume that $k \le N + M$ for simplicity.
We may then show that the run-time complexity of the DKP construction is $O((N + k)M)$, and of the RBB construction is $O(NM(N+M))$, in each case bounded by the complexity of obtaining a \oneway\ procedure in normal form.
(The difference in complexity between the two constructions may be traced to the difference in the procedures to obtain the stable-index expression $C''$ for each, and hence to the single-qubit gates which are added in the RBB construction in order to impose the desired topological constraints on the indices of the stable-index expression.)
An explicit analysis of the complexity of these constructions is presented in \autoref{apx:constructionEfficiency}.

\subsubsection{Partial constructions}
\label{sec:partialConstructions}

As a final remark on constructions of \oneway\ procedures, it will be useful to describe partial versions of the DKP and (simplified) RBB constructions.
We define the \emph{DKP construction without normalization}\footnote{%
	This construction is referred to in~\cite{BeaudPhD} as the ``simplified'' DKP construction; we use different terminology here in order to be more uniform with our description of the RBB construction.}
and the \emph{RBB construction without normalization} to be, simply, the \oneway\ procedure constructions which result from omitting the Pauli simplifications and signal shifting stages from the phase of obtaining a \oneway\ procedure in normal form.
A procedure produced by such a partial procedure may then indicate sign dependencies for a measurements whose default angle is an integral multiple of $\pion 2$, and they may have signal-shift operators immediately following measurements.

In a practical setting, provided sufficiently fast control of the classical memory storing the measurement results $\s[v]$ for all qubits $v$ in a \oneway\ procedure, it is likely that a more useful procedure would be an intermediate construction to the constructions  with or without normalization, in which one performs Pauli simplifications but leaves the shift operations in place.
The purpose of identifying the partial constructions without normalization above is to obtain a simplified analysis of measurement dependencies, which may be used as a starting point for the analysis of the ``complete'' constructions.

\section{Semantic maps}
\label{sec:semanticMaps}

We may now consider the roles of the DKP and (simplified) RBB constructions as defining representations of unitary circuits in the \oneway\ model.
Both constructions make use of a simple process of translating unitary circuits from a gate set including $J(\theta)$ gates and a single two-qubit diagonal operation into corresponding \oneway\ procedures: this will induce a graph structure which is dominated by vertex-disjoint paths, with supplemental vertices and edges linking pairs of paths.
After standardizing the \oneway\ procedures in the course of these constructions, the additional structure which might arise from measurement dependencies or the grouping of operations is lost or obfuscated; preparation maps and entangling operations become dissociated from individual measurement operations, and corrections will in some cases be absent entirely after absorbing them into measurements and performing Pauli simplifications.

Given a \oneway\ procedure, in which no additional information is encoded \eg\ in the labels for the qubits, we may ask whether we can efficiently identify whether a \oneway\ procedure is one that may be produced by the DKP or the RBB constructions.
In the case where the input and output subsystems of the procedure are of the same size, we answer this question in the affirmative.
We do so by considering the combinatorial structures which are present in the geometries (\autoref{def:geometry}) of the \oneway\ procedures which result from these constructions.
We then define a map $\cS$ from \oneway\ procedures to circuits in the $\ens{H,T,\cZ}$ model which serves as a semantic map (in the sense of \autoref{def:semanticMap}), for the representation maps $\cR_\DKP$ and $\cR_\RBB$ corresponding to both the DKP and simplified RBB constructions restricted to input circuits which are unitary bijections.

\subsection{Measurement dependencies arising in the constructions}
\label{sec:dependencies}

The DKP and RBB constructions produce \oneway\ procedures from circuits by transforming gates such as $J(\theta)$, $\cZ$, and $\Zz$ into simple \oneway\ subroutines, via the mapping $\Phi$ defined in~\eqref{eqn:jointHomphm}.
It is instructive to consider the structures described by the sign- and bit-dependencies which arise from composing such \oneway\ subroutines, using the partial constructions without normalization described in \autoref{sec:partialConstructions}.

\subsubsection{Measurement dependencies and index successor functions}
\label{sec:dependenciesIndexSuccFns}


For a unitary circuit $C$ consisting of gates from the set $\cB_\DKP = \ens{J(\theta), \cZ}_{\theta \in \frac{\pi}{4}\Z}$, consider the last gate $J(\theta)$ performed in the circuit.
(If there is more than one such gate which may be performed in parallel, we may choose an arbitrary one.)
We may decompose $C$ into circuits $C_\star \tilde C$, where $C_\star$ consists of this final $J(\theta)$ gate acting on some qubit $u$, and any controlled-$Z$ operations acting on $u$ which follow it; $\tilde C$ consists of the remaining operations of $C$.
We may refer to $C_\star$ as a \emph{star circuit}\footnote{%
	These star circuits are similar to, but have a different orientation to, the circuits described in Section~III~B of\cite{DK06}.
 	This different choice of construction will facilitate the analysis of combinatorial structures later in the article: these are used in much the same way as the circuits arising from ``star patterns'' do in\cite{DK06}.}
Consider the translation of a star circuit $C_\star$ into the \oneway\ model via the map $\Phi$ defined in \eqref{eqn:jointHomphm}, and standardizing the resulting procedure (\autoref{def:standardize}).
Writing $C_\star$ in stable-index tensor notation (\autoref{sec:stableIndexNotation}), and letting $W$ be the set of qubits $w \ne u$ on which $C_\star$ acts, we have
\vspace{-1ex}
\begin{align}
	\label{eqn:starCircuitPremonition}
		\Phi\big(C_\star\big)
	\;\;=&\;\;
		\Phi\paren{\paren{\prod_w \cZ\pseu[w,v]} J(\theta)\pseu[:v/u]}
	\notag\\[1ex]=&\;\;
		\paren{\prod_{w \in W} \Ent{wv}{}} \fJ^\theta_{v/u}
	\notag\intertext{\iffalse\\[1ex]\fi}=&\;\;
		\paren{\prod_{w \in W} \Ent{wv}{}} \Xcorr{v}{\s[u]} \Meas{u}{-\theta} \Ent{uv} \New{v}
	\notag\\[1ex]\cong&\;\;
		\paren{\prod_{w \in W} \Zcorr{w}{\s[u]}} \Xcorr{v}{\s[u]} \Meas{u}{-\theta} \paren{\prod_{w \in W} \Ent{wv}{}} \Ent{uv} \New{v} \;,
\end{align}
where on the last line we commute the entangling operations to the right.
In the final \oneway\ procedure, the qubit $u$ is distinguished as the only qubit measured, and $v$ is the only qubit prepared; $v$ is also the only qubit subject to an $\Xcorr{v}{\s[u]}$ operation, with every other qubit (aside from $u$) being subject to a $\Zcorr{w}{\s[w]}$ operation.
The fact that $u$ and $v$ are measured and prepared qubits (respectively) in $\Phi(C_\star)$ above is a direct result of the fact that they are deprecated and advanced indices (respectively) in the stable-index expression given for $C_\star$.
In particular we have $v = f(u)$, where $f$ is the index successor function (\autoref{def:indexSuccessorFn}) of $C_\star$; thus, the corrections may be described in terms of the index successor function of the original circuit, and the entangling operations.

We may make similar observations for the (simplified) RBB construction as we have for the DKP construction above.
%
Motivated by the realization of the operator $\Zz_{vw}$ in the \oneway\ procedure $\fZz_{v,w}$, we may write
\begin{gather}
	\label{eqn:mediatorStarCircuitPremonition}
	 	\Zz\pseu[v,w]
	\;\;\propto\;\;
		P \pseu[:/a] \cZ\pseu[v,a] \cZ\pseu[w,a] \ket{+}\pseu[:a/] \,,
\end{gather}
where we define the operator $P = \bra{\smash{+_{\pi\!/\!2}}}$\,, and for a constant proportionality factor (as illustrated in \eqref{eqn:projectPlusZzz-a} in \autoref{apx:constructGatesApx-2qubit}).
In the stable-index expression on the right,
the index successor function $f$ remains the same (as the indices $a$ are advanced without corresponding to any deprecated index, and vice versa), but the interaction graph $G$ is transformed essentially by subdividing every edge arising from a $\Zz$ gate by introducing a new vertex in the middle (representing the mediating qubit with the index $a$).
Then, consider the translation of $\Zz_{v,w}$ to a \oneway\ procedure via $\Phi$,
\begin{align}
		\Phi\Big( \Zz(\theta)\pseu[u,v] \Big)
	\;=&\;\,
		\Zcorr{u}{\s[a]} \Zcorr{v}{\s[a]} \Meas{a}{\pion 2} \Ent{u a} \Ent{v a} \New{a}	\;:
\end{align}
if we extend the successor function $f$ to a function $\tilde f$ for which we define $\tilde f(a) = a$, while we have no $\Xcorr{w}{\s[a]}$ operations, we have a $\Zcorr{w}{\s[a]}$ for every qubit $w$ adjacent to $a = \tilde f(a)$ in the entanglement graph.

In each case, either the index successor function (or a modest extension of it) for a stable-index description of a simple unitary circuit can be used to characterize the corrections performed in a corresponding \oneway\ procedure.
We may then consider the result of translating a composition of such simple circuits,
\begin{gather}
		\Phi(C)
	\;=\;
		\Phi\Big( C\supp{\ell}_\star \cdots C\supp{1}_\star	\Big)
	\;=\;
		\Phi\Big( C\supp{\ell}_\star \Big) \circ \cdots \circ \Phi\Big( C\supp{1}_\star	\Big)	\;,
\end{gather}
where each $C\supp{j}_\star$ is a star circuit as above, extending the definition to also include individual $\Zz$ operations
as described in \eqref{eqn:mediatorStarCircuitPremonition}.
Commuting all of the preparation and entangling operations to the right, we obtain the same dependency of the corrections on the  successor function (extended for each $\Zz$ gate as described above): we obtain $\Zcorr{w}{\s[u]}$ corrections for each qubit in the interaction graph which is adjacent to $\tilde f(u)$, either due to the commutation of entangling maps past $\Xcorr{f(u)}{\s[u]}$ operations or from the $\Zcorr{w}{\s[u]}$ operations arising from $\fZz$ procedures.

\subsubsection{Efficiently verifying consistency of measurement dependencies}
\label{sec:verifyDependencies}

Based on the descriptions of the dependencies arising in the DKP and RBB constructions without normalization, we may characterize the dependencies of operations which arise in the two constructions \emph{with} normalization, and show that verifying that such dependencies hold in a \oneway\ procedure $\fP$ can be used as a subroutine to certify that $\fP$ performs a unitary transformation.

Following the remarks of the preceding sections about the constructions without normalization, we characterize the dependencies arising from the DKP and RBB constructions as follows.
\begin{definition}
	\label{def:extendedSuccFn}
	For a stable-index representation $C$ for a circuit the gate set $\cB_\DKP \union \cB_\RBB$, let $\tilde C$ be the stable-index expression obtained by performing the substitution of \eqref{eqn:mediatorStarCircuitPremonition} for each $\Zz$ operation in $C$.
	Let $f$ be the index successor function of $\tilde C$, as defined in \autoref{def:indexSuccessorFn}.
	The \emph{extended successor function} of $C$ is then the function 
	\begin{gather}
	 		\tilde f(v)
		\;=\;
			\begin{cases}
			 		f(v)	\;,	&	\text{if $v$ an index of $C$}	\\
					v		\;,	&	\text{if $v$ is an index deprecated in $\tilde C$ by a $P$ operator}
			\end{cases}	\;,
	\end{gather}
	whose domain is the set of indices of $\tilde C$ which are deprecated.
\end{definition}

\begin{theorem}
	\label{thm:verifyDependencies}
	Let $\fP_0 = \Phi(C)$, for a unitary circuit $C$ over either the gate set $\cB_\DKP \union \cB_\RBB$, and let $\bar\fP$ be the result of normalizing $\fP_0$.
	Let $(G,I,O)$ be the geometry underlying $\fP_0$, and let $f$ be the extended index successor function of $C$.
	For the sake of brevity, define $A_X \subset O\comp$ as the set of qubits to be measured with default angles $0$ or $\pi$ in $\fP_0$, and $A_Y \subset O\comp$ the set of qubits to be measured with default angles $\pion[{{\scriptscriptstyle \pm}}]2$ in $\fP_0$.
	We define square boolean matrices $F, T : V(G) \x V(G) \to \Z_2$ in terms of their coefficients, as follows:\footnote{%
		These operators are the transposes of the operators described in~\cite{BeaudPhD}, corresponding to changes in the representation in order to simplify the description of the effect of signal shifting; the statement of the result is also changed accordingly.}
	\begin{subequations}
	\begin{align}
			\label{eqn:signDependMatrix}
			F_{vw}
		\;=&\;\,
			\begin{cases}
			 	1	\,,	&	\text{if $v \in O\comp$ and $w = f(v) \notin A_X \union A_Y$}	\\
				0	\,,	&	\text{otherwise}
			\end{cases}	\;,
		\\[1ex]
			\label{eqn:shiftDigraphMatrix}
			T_{vw}
		\;=&\;\,
			\begin{cases}
			 	1	\,,	&	\text{if $w \ne v \in O\comp$ and $w \sim f(v)$}	\\
				1	\,,	&	\text{if $v,w \in O\comp$ and $w = f(v) \in A_Y$}		\\
				0	\,,	&	\text{otherwise}
			\end{cases}	.
	\end{align}
	\end{subequations}
	Then $(\idop - T)$ is invertible modulo $2$, and the dependencies in $\bar\fP$ between qubits $v, w \in V(G)$ may be characterized in terms of $F$ and $T$ as follows:
	\begin{itemize}
	\item 
		$w$ has a sign-dependency on $v$ if and only if $[(\idop - T)\inv F]_{vw} = 1$, for $w \in O\comp$; 

	\item
		there is an $\Xcorr{w}{}$ operation conditioned on $\s[v]$ if and only if $[(\idop - T)\inv F]_{vw} = 1$, and a $\Zcorr{w}{}$ operation conditioned on $\s[v]$ if and only if $[(\idop - T)\inv T]_{vw} = 1$, for $w \in O$.
	\end{itemize}
\end{theorem}
\begin{proof}
	By construction, for qubits $u,v \in V(G)$, $\fP_0$ contains a $\Zcorr{v}{\s[u]}$ operation if and only if $v \sim f(u) = u$; as $\Zcorr{v}{\s[u]}$ operations commute with entangling maps, we obtain no further corrections depending on $u$ if we commute the entangling maps of $\fP_0$ to the right.
	Similarly, $\fP_0$ contains an $\Xcorr{v}{\s[u]}$ operation if and only if $v = f(u) \ne u$, which must arise from a $\fJ^\theta_{v/u}$ procedure.
	Any operation on $v$ in $\fP_0$ which is not a part of the $\smash{\fJ^\theta_{v/u}}$ procedure must come after it, every entangling map acting on $v$ must occur after the $\Xcorr{v}{\s[u]}$ operation except for the $\Ent{uv}$ map occurring as a part of $\smash{\fJ^\theta_{v/u}}$.

	Commuting all of the entangling maps in $\fP_0$ to the right, we then induce $\Zcorr{w}{\s[u]}$ operations for all $w \sim v$ such that $w \ne u$, as described in \eqref{eqn:induceZcorr}.
	Thus, commuting all preparation and entangling maps in $\fP_0$ to the right yields an equivalent procedure $\fP_1$, in which for there is an $\Xcorr{w}{\s[v]}$ operation for every $w = f(v) \ne v$, and a $\Zcorr{w}{\s[v]}$ operation for every $w \sim f(v) \ne v$, for every qubit $v \in O\comp$.
	That is, we have
  	\begin{gather}
		\label{eqn:extend-DK06-Premonition}
		\fP_1
		\;=\;
	 	\sqparen{\ordprod[\le]_{u \in O\comp}
			\paren{\prod_{\substack{w \sim f(u) \\ w \ne u}} \Zcorr{w}{\s[u]}}
			\paren{\prod_{\substack{v = f(u) \\ v \ne u}} \Xcorr{v}{\s[u]}}
			\Meas{u}{\theta_u}
		} \Ent{G} \New{I\comp}
	\end{gather}
	for some linear order $\le$ satisfying the conditions of \eqref{eqn:flowPremonition} for the extended successor function $f$.

	Commuting correction operations in $\fP_1$ to the left, we obtain another procedure $\fP_2$, in which the same corrections hold as above for $w \in O$.
	For $w \in O\comp$, there is instead a sign-dependency on $v$ if and only if $w = f(v) \ne v$, and a bit-dependency on $v$ if and only if $w \ne v$ and $w \sim f(v)$.
	Note that the procedure of producing $\fP_2$ from $\fP_0$ is precisely that of standardizing $\fP_0$: then $\bar\fP$ may be obtained by normalizing $\fP_2$.

	Consider how the process of normalization effects dependencies of qubits $w$ on a particular qubit $v$.
	\begin{romanum}
	\item
		The effect of Pauli simplifications are to remove sign-dependencies from qubits in $A_X$ according to \eqref{eqn:pauliSimplifX}, and to change them to bit-dependencies for qubits in $A_Y$ according to \eqref{eqn:pauliSimplifY}.
		As signal shifting does not introduce sign-dependencies for measurements where none previously exist, the measurement of a qubit in $A_X \union A_Y$ will not have any dependencies on previous measurements in a procedure in normal form.

		Applying Pauli simplifications to $\fP_2$ will yield a \oneway\ procedure $\fP_3$: then $\fP_3$ will have no sign-dependencies for qubits $v \in A_X \union A_Y$, and will have signal shift operators $\Shift{v}{\s[u]}$ for any qubits $v = f(u) \in A_Y$, in addition to the shift operators in $\fP_2$.
		We may describe the sign dependencies which remain in $\fP_3$ by a square matrix over $V(G)$, with a $1$ in the row $v$ and column $w$ for $v,w \in V(G)$ when $w$ has a sign-dependency on $v$.
		Accounting for the removals of sign dependencies of $\fP_2$ for qubits $w \in A_X \union A_Y$, the procedure $\fP_3$ then has a sign-dependency of $w$ on $v$ when $w = f(v) \notin A_X \union A_Y$.
		The matrix described in this manner is then the matrix $F$ in \eqref{eqn:signDependMatrix}.

		For the sake of simplicity of discussion, we will assume that the shift operators are not accumulated with any pre-existing shift operators using the relation of \eqref{eqn:accumShifts}.
		However, it is important to note that the shift operators arising from Pauli simplifications will not cancel any operators existing originally in $\fP_2$: if $w = f(v)$, it follows that $w$ is not adjacent to $f(v)$, as $G$ contains no loops by construction.

	\item
		The effect of signal shifting is to propagate a bit-dependency of each qubit $w$ on $v$ (represented by a $\Shift{w}{\s[v]}$ operation just after the measurement of $w$) to the operations which in turn depend on the value of $\s[w]$, according to \eqref{eqn:signalShift}.

		We may describe the effect of signal shifting in terms of walks in a directed graph $D$ of dependencies, where we have arcs $w \arc v$ when there is a shift operator $\Shift{w}{\s[v]}$.
		An operation depending on $w$ may be transformed into one depending on $w$ and all qubits $w$ for which $D$ has an arc $w \arc v$ by shifting of operators $\Shift{w}{\s[v]}$; in turn, the resulting operation may be transformed into one depending also on qubits $u$ for which $D$ has arcs $v \arc u$, ranging over the qubits $v$ of the previous step; and so on.
		The effect is then of transforming bit-dependencies corresponding to arcs in this digraph $D$ (arising from signal shifts) into measurement and correction dependencies corresponding to directed walks in $D$.

		We may represent this digraph as an adjacency matrix over $V(G)$, with a $1$ in the row $v$ and column $w$ for $v,w \in V(G)$ when $\fP_3$ contains a shift operator $\Shift{w}{\s[v]}$, and $0$ if not.
		Considering the shift operations in $\fP_3$ which are also present in $\fP_2$, there will be such a shift operator whenever $w \ne v$ and $w \sim f(v)$; and as a result of Pauli simplifications, there will be such a shift operator when $w = f(v)$ and $w \in A_Y$.
		As there are no other sources of shift operators, these are the only positions in which the adjacency matrix will be non-zero.
		The adjacency matrix described is then the matrix $T$ in \eqref{eqn:shiftDigraphMatrix}.

		As the ordering $\le$ (used in \eqref{eqn:extend-DK06-Premonition} to fix the order of operations) is a linear order, the digraph $D$ described by the shift operations is acyclic: 
		for any standard basis vector $\unit_w \in \Z_2^{V(G)}$, the column-vector $T \unit_w$ is supported only on indices $v \le w$.
		Then we have $T^r = 0$ for some $r \le \card{V(G)}$, so that $T$ is nilpotent (and in particular has no $+1$ eigenvectors).
		Consequently, the operator $\idop - T$ is invertible.
		We may then express the cumulation of dependencies for each operation in $\fP_3$ due signal shifting as follows.
		For two qubits $v$ and $w$, the number of walks in $D$ from $w$ to $v$ of a fixed length $\ell \ge 0$ is given by $\unit_v\trans T^\ell \unit_w $\,: the total number of walks of any length from $w$ to $v$ in $D$ is then
		\begin{gather}
		 		\unit_v\trans \Bigg( \sum_{\ell = 0}^\infty \, T^\ell \Bigg) \unit_w
			\;=\;
				\unit_v\trans \big( \idop - T \big)\inv \unit_w	\;.
		\end{gather}
		If we represent the dependencies of operations in terms of boolean column vectors, we may then represent the effect of signal shifting by multiplying a vector of dependencies by $(\idop - T)\inv$ computed modulo $2$.
	\end{romanum}
	To determine the classical dependencies of each operation in $\bar\fP$, it suffices then to identify the dependencies of the same operations in $\fP_3$, represent these as a vector $\vec d \in \Z_2^{V(G)}$, and compute $(\idop - T)\inv \vec d$.
	We may characterize the classically controllable operations and their dependencies in $\fP_3$ as follows:
	\begin{itemize}
	\item
		As we have noted above, the sign-dependencies for qubits $w \in O\comp$ are represented by the coefficients of the matrix $F$, where $F_{vw} = 1$ if and only if $w$ has a sign dependency on $v$ in $\fP_3$.
		Then $w$ has a sign dependency on $v$ in $\bar \fP$ if and only if $\big[ (\idop - T)\inv F \big]_{vw} = 1$.

	\item
		For $w \in O$, there is no measurement into which corrections on $w$ may be absorbed, and so the corrections in $\fP_3$ are the same as in $\fP_1$: these are $\Xcorr{w}{\s[v]}$ operations when $w = f(v)$, and and $\Zcorr{w}{\s[v]}$ operations when $w \sim f(v)$, where in either case we require $w \ne v$ as well.
		As $O$ is disjoint from both $A_X$ and $A_Y$ by definition, we may then note that for $w \in O$ there is an $\Xcorr{w}{\s[v]}$ operation in $\fP_3$ if and only if $F_{vw} = 1$, and a $\Zcorr{w}{\s[v]}$ operation if and only if $T_{vw} = 1$.
		It follows that after signal shifting, the resulting $\Xcorr{w}{\beta}$ operation depends on $v$ if and only if $[(\idop - T)\inv F]_{vw} = 1$, and the resulting $\Zcorr{w}{\gamma}$ operation depends on $v$ if and only if $[(\idop - T)\inv T]_{vw} = 1$.
 		\qedhere
	\end{itemize}
\end{proof}

\autoref{thm:verifyDependencies} thus characterizes the dependencies which arise in the DKP and RBB constructions, in terms of an (extended) index successor function for a circuit $C$ over the gate-set $\cB_\DKP \union \cB_\RBB$.

\Algorithm{alg:testDependencies} describes an algorithm which, provided a \oneway\ measurement pattern $\fP$ in normal form and a candidate for the extended successor function $f$ for an originating circuit, tests if the dependencies of $\fP$ are consistent with the dependency conditions described in \autoref{thm:verifyDependencies}.
\begin{algorithm}[t]

\PROCEDURE \TestDependencies(\fP, f) :

\Input{%
	$\fP$, a \oneway\ procedure with geometry $(G,I,O)$;
	\\
	$f : O\comp \to I\comp$, a candidate successor function for $\fP$.
}

\Output{%
	A boolean result indicating if the dependency relations in $\fP$ are consistent with \autoref{thm:verifyDependencies}, relative to the function $f$.
}

\BEGIN
	\algline
 	\LET $T$ be a matrix depending on $f$ as in \eqref{eqn:signDependMatrix}\;
 	\LET $F$ be a matrix depending on $f$ as in \eqref{eqn:shiftDigraphMatrix}\;
	\LET $\tilde T \gets (\idop_n - T)$\;
	\FOREACH measurement or correction operation $\Op{}{}$ in $\fP$ \DO
		\algline
		\IF $\Op{}{} = \Meas{w}{\theta;\beta}$ for some $w, \theta, \beta$ \THEN {
			\smallskip
			\LET $\vec d \in \Z_2^{V(G)}$ be the indicator vector for $\beta$\;
			\oneline
			\IF $\tilde T \vec d \ne F \unit_w$ \THEN \RETURN \false\;
			\smallskip
		}
		\ELSEIF $\Op{}{} = \Xcorr{w}{\beta}$ or $\Op{}{} = \Zcorr{w}{\beta}$ for some $w, \beta$ \THEN {
			\smallskip
			\LET $\vec d \in \Z_2^{V(G)}$ be the indicator vector for $\beta$\;
			\short
			\IF \upshape ($\Op{}{} = \Xcorr{w}{\beta}$ \ANDALSO $\tilde T \vec d \ne F \unit_w$) 
					\ORELSE ($\Op{}{} = \Zcorr{w}{\beta}$ \ANDALSO $\tilde T \vec d \ne T \unit_w$)
			\THEN {	\RETURN \false }\relax
			\smallskip
		}\relax
	\ENDFOR

	\smallskip
	\RETURN \true\;
\END

	\caption[An algorithm to test whether the dependencies of a \oneway\ procedure is consistent with those resulting from the DKP and simplified RBB constructions.]{\label{alg:testDependencies}%
		An algorithm to test whether the dependencies of a \oneway\ procedure is consistent with those resulting from the DKP and simplified RBB constructions, relative to a candidate for the (extended) index successor function for the original circuit.}

\end{algorithm}

If we are presented with a \oneway\ procedure $\fP$ in normal form which is \emph{not} known to be a result of the DKP or RBB constructions, we may only verify whether these conditions hold relative to some ``candidate'' successor function $f$.
How such a function $f$ may be constructed for a \oneway\ procedure $\fP$ is the subject of \autoref{sec:combinatorialStruct}.

\paragraph{Run-time analysis of Algorithm~\ref*{alg:testDependencies}.}
\label{sec:testDependenciesAnalysis}

We bound the run-time of \Algorithm{alg:testDependencies} as follows.
In the following, we let $N$ be the number of operations in $\fP$, and for the geometry $(G,I,O)$ underlying $\fP$, we let $n = \card{V(G)}$, $k = \card{O}$, and $m = \card{E(G)}$.
We assume that $f$ can be evaluated in constant time, using an array structure.
We assume that the dependencies of operations $\smash{\Op{}{\beta}}$ (for sign-dependencies or other classical control expressions) are represented as indicator vectors $\vec d \in \smash{\Z_2^{V(G)}}$ representing the presence or absence of the term $\s[v]$ in $\beta$ by $d_v = 1$ (with $d_v = 0$ otherwise).

By definition, $\fP$ will contain $m$ entangling operations, $n - k$ measurement operations, and at most $2k$ correction operations; then, the length of $\fP$ is $N \in O(m + n) \subset O(n^2)$.
Rather than compute matrix products involving $(\idop - T)\inv$, we may verify whether $(\idop - T) \vec d = \vec v$ for $\vec v = F \unit_w$ or $\vec v = T \unit_w$ as appropriate.
The coefficients of $(\idop - T) \vec d$ may be computed by iterating over the neighbors of $f(w)$ for $w \in \supp(\vec d)$: as as $T_{vw} = 1$ only for $w$ adjacent to $f(v)$ (if $v \ne w$) or for $w = f(v)$, there are at most two non-zero coefficients in $T$ for each edge in $G$, so that this iteration requires time $O(m)$ for each vector $\vec d$.
Similarly, the time required to compute column-vectors $F\unit_w$ and $T \unit_w$ for each $w$ is $O(m)$.
As there are $O(n)$ measurement and correction operations, the cumulative time to perform these matrix computations is $O(nm)$; as the remaining operations in $\fP$ require no verification, this is the total time complexity of the \textbf{for} loop.
The execution time of \Algorithm{alg:testDependencies} is then $O(nm)$.

\subsection{Conditions on candidate successor functions}
\label{sec:combinatorialStruct}

In the preceding section, we characterized the dependencies in the DKP and RBB constructions in terms of an extended index successor function, from a stable-index representation of the original circuits.
This leaves open the question of how such a successor function may be obtained.
We examine this problem using the observations made above about the dependencies arising from circuit decompositions without normalization.
We show that these dependencies can be captured in each case by a modification of the \emph{flow} conditions formulated by Danos and Kashefi~\cite{DK06}.
This modification will provide the ``candidate successor functions'' which we require for the verification procedure of \Algorithm{alg:testDependencies}, and which lead to important structural constraints on the geometries $(G,I,O)$ arising from the DKP and RBB constructions.

\subsubsection{Measurement ordering from index successor functions}

As we described in \autoref{sec:onewayStableIndexCorresp}, there is a correspondence between the ``deprecation order'' of indices as described in \eqref{eqn:flowPremonition}, for a stable-index expression for a circuit $C$ over $\cB_\DKP \union \cB_\RBB$, and the order of operations in a well-formed \oneway\ procedure $\Phi(C)$.
This correspondence may be attributed to a sense in which measuring qubits in a \oneway\ procedure simulates the summation over indices in a tensor expression, and thus to evaluating a unitary circuit applied to a state as a sum over computational paths.

When applied to a completely specified pure state, every deprecated index in a stable-index expression becomes bound, in which one evaluates the result by summing over all deprecated indices.
When these indices are represented by qubits in a \oneway\ procedure, one may simulate this summation by measurement of qubits $v$ with respect to the $\ket{+_\theta},\ket{-_\theta}$ basis.
\begin{itemize}
\item 
	Consider a measurement arising from a $\fJ^\theta$ procedure, corresponding to summing over the deprecated index of a $J(\theta)$ operation.
	If the $\ket{+_\theta}$ state is measured (\ie\ if $\s[u] = 0$), the resulting transformation corresponds to summing over the $\ket{0}$ and $\ket{1}$ components, weighted with the with the relative amplitudes phases as those associated with the $J(\theta)$ operation which deprecates the corresponding index.

\item
	For measurements arising from a $\fZz_{v,w}$ procedure, we measure a mediating qubit $a$ in a fixed basis $\ket{\smash{\pm_{\pi\!/\!2}}}$.
	Again, a measurement result of $\ket{\smash{+_{\pi\!/\!2}}}$ corresponds to summing over an auxiliary index $a$ in a stable-index expression: the relative phases of the $\ket{0}$ and $\ket{1}$ components correspond to those of the projection operator $P$.
\end{itemize}
In either case, if the measurement yields the $\ket{-_\theta}$ result instead, the corrections and measurement adaptations are chosen precisely to effect the same transformation which would have occurred for the $\ket{+_\theta}$ result, by influences on the measurement angles due to sign-dependencies or the measurement results due to bit-dependencies on $v$.\footnote{%
	Bit-dependencies, as represented by shift operators \smash{$\Shift{v}{\ast}$} immediately after measurements, may be informally interpreted as flipping the bit $\s[v]$ to produce ``the result which would have been measured if an appropriate $Z$ correction had been performed prior to measurement''.
	As the measurement results are often maximally random whether or not such corrections are performed, this is a somewhat counterfactual interpretation, but one which could perhaps be given an ontological foundation in terms of hidden variables.}

Thus, we may interpret the sign- and bit-dependencies in such \oneway\ computations as potential ``influences'' in evaluating the summation of deprecated indices.
By construction, the measurements are chosen in such a way that the tensor which is produced at the output is related to the tensor at the input by a unitary circuit: the amplitudes arising from each measurement are a function of at most two boolean indices, where the two-index dependencies are signs encoded by the entanglement graph, and whose single-index dependencies are given by the choice of measurement basis.
This may be regarded as indicating that a \oneway\ procedure where such an ordering is possible has a ``circuit-like structure''.

A measurement order which corresponds to a ``deprecation order'' for indices in a circuit $C$ as in \eqref{eqn:flowPremonition} is governed by the index successor function of $C$.
An extended index successor function as defined in \autoref{def:extendedSuccFn} has the effect of imposing constraints on the deprecation of the mediating indices $a$ arising from the substitution of \eqref{eqn:mediatorStarCircuitPremonition} for $\Zz$ gates, which may be interpreted as creating a well-defined notion of when the $\Zz$ operation is performed.
By construction, \oneway\ procedures arising from the DKP and RBB constructions do admit extended successor functions: this can be used to define the domain of a semantic map for these constructions.

\subsubsection{Flow conditions}


The ordering in \eqref{eqn:flowPremonition} was first formulated by Danos and Kashefi~\cite{DK06}, for qubits in the \oneway\ model rather than tensor indices, to describe dependencies arising in the DKP construction without normalization. 
We may formulate these properties in terms of the geometry of a \oneway\ procedure without reference to the DKP construction as follows:
\begin{definition}
	\label{def:flow}
	For a geometry $(G,I,O)$, let $u \sim v$ denote the adjacency relation in $G$.
	A \emph{flow} is a pair $(f, \preceq)$ consisting of an injective function $f: O\comp \to I\comp$,
	and a partial order $\preceq$ on $V(G)$, such that the conditions
	\begin{subequations}
	\label{eqn:flow}%
	\begin{gather}
			\label{flow:a}			v \sim f(v)	\;,
		\\
			\label{flow:b}		\textbalance{v \preceq f(v)\;,}[\;\;\text{and}]
		\\
			\label{flow:c}		w \sim f(v)	\;\implies\; v \preceq w
	\end{gather}
	\end{subequations}
	hold for all $v \in O\comp$ and $w \in V(G)$.
\end{definition}

Examples of geometries with and without flows are illustrated in \autoref{fig:examplesFlow} and \autoref{fig:exampleNoFlow}, respectively.
\begin{figure}[t]
	\begin{center}
			\includegraphics{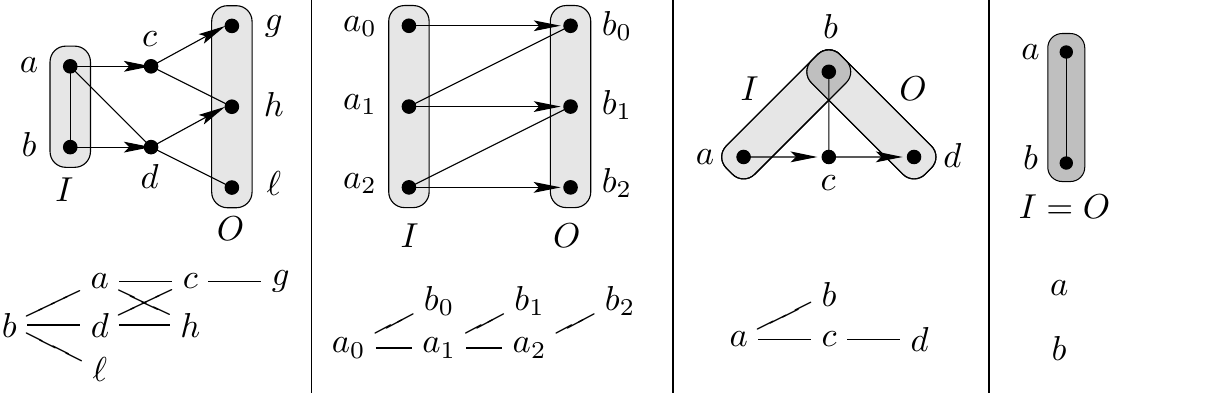}
			\vspace{-1ex}
	\end{center}
	\caption[Examples of geometries with flows.]{\label{fig:examplesFlow}
		Examples of geometries with flows.
		Arrows indicate the action of a function $f: O\comp \to I\comp$ along otherwise undirected edges.
		Corresponding partial orders $\preceq$ for each example are given by Hasse diagrams below the graphs (with minimal elements on the left and maximal elements on the right).
		In the right-most example, the two vertices $a$ and $b$ are incomparable, \ie\ there is no order relation between them.}
		\vspace{1em}
\end{figure}%
\begin{figure}[t]
	\begin{center}
			\includegraphics{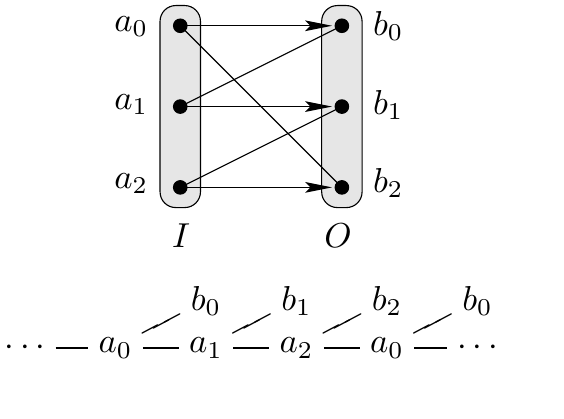}
			\vspace{-1ex}
	\end{center}
	\caption[Example of a geometry with no flow.]{\label{fig:exampleNoFlow}
		Example of a geometry with no flow (\cf\ the second geometry from the left in \autoref{fig:examplesFlow}).
		Arrows indicate the action of an injective function $f: O\comp \to I\comp$ along otherwise undirected edges.
		Also given is a reflexive and transitive binary relation $\preceq$ which satisfies conditions~\eqref{flow:b} and \eqref{flow:c} for this function $f$, but which is not antisymmetric.}
\end{figure}%
The function $f$ and the partial order $\preceq$ capture the essential structure of the dependencies in the DKP construction without normalization: $f$ represents the mapping of qubits to their successors (corresponding to an index successor function for a circuit), and the partial order $\preceq$ represents a suggested order in which the qubits may be measured (or the tensor indices deprecated).

Danos and Kashefi show that any geometry which has a flow underlies some \oneway\ measurement pattern which can be obtained by the DKP construction,\footnote{%
 	The converse of this statement does not necessarily hold: by omitting or manipulating the dependencies in a \oneway\ procedure, a variant procedure which does \emph{not} perform a unitary may be obtained without changing the underlying geometry.}
and which therefore performs a unitary transformation:
\begin{lemma}[{\cite{DK06}, Theorem~1}]
	\label{lemma:DK06}
	Suppose $(f,\preceq)$ is a flow for $(G,I,O)$.
	Let $v \sim w$ denote the adjacency relation of $G$: then for any linear order $\le$ extending $\preceq$, the measurement procedure
	\begin{gather}
	 	\sqparen{\ordprod[\le]_{u \in O\comp}
			\paren{\prod_{\substack{w \sim f(u) \\ w \ne u}} \Zcorr{w}{\s[u]}}
			\Xcorr{f(u)}{\s[u]}
			\Meas{u}{\theta_u}
		} \Ent{G} \New{I\comp}
	\end{gather}
	performs a unitary transformation.
\end{lemma}

The proof presented in~\cite{DK06} is essentially by showing that applying the map $\Phi$ defined in \eqref{eqn:jointHomphm} to a suitably constructed circuit yields the \oneway\ procedure above.
If we can construct a flow for a geometry $(G,I,O)$, we may then use the same circuit construction as in~\cite{DK06} to construct a candidate circuit whose representation as a \oneway\ procedure (via the DKP construction) has the same geometry.
We may then verify whether the dependencies of a \oneway\ procedure are consistent with one arising from the DKP construction, using \Algorithm{alg:testDependencies}.

\subsubsection{Extending to modified flows}

By considering the dependencies described in the proof of \autoref{thm:verifyDependencies} following \eqref{eqn:extend-DK06-Premonition}, we may formulate a similar combinatorial condition which generalizes flows, and also captures the dependencies arising in the (simplified) RBB construction without normalization:\footnote{%
  	A similar generalization of the flow conditions was anticipated in\cite{DK06}, but does explicitly specify combinatorial conditions for this generalization.}
\begin{definition}
	\label{def:mflow}
	Let $(G,I,O,M)$ consist of a geometry $(G,I,O)$ and a subset $M \subset O\comp$, and let $u \sim v$ denote the adjacency relation in $G$.
	A \emph{modified flow} for $(G,I,O,M)$ is an ordered pair $(f, \preceq)$ consisting of an injective function $f: O\comp \to I\comp$,
	and a partial order $\preceq$ on $V(G)$, such that the conditions
	\begin{subequations}
	\label{eqn:mflow}
	\begin{gather}
			\label{mflow:a}		f(v) \sim v	\;, \;\;\text{for $v \in M\comp$}\,;
		\\
			\label{mflow:b}		\text{either~} f(v) = v	\text{~or~} f(v) \sim v \;,\text{~for $v \in M$}\,;
		\\
			\label{mflow:c}		v \preceq f(v)\;;\;\;\text{and}
		\\
			\label{mflow:d}		w \sim f(v)	\;\implies\; v \preceq w
	\end{gather}
	\end{subequations}
	hold for all $v \in O\comp$ and $w \in V(G)$.
\end{definition}
The function $f$ again corresponds essentially to an (extended) successor function for stable-index expression of a unitary circuit, and $\preceq$ again represents an order of measurement corresponding to an order in which the indices of stable-index expression are deprecated.
The set $M$ represents qubits whose measurement angles are $\pion 2$\,: these are qubits which might be mediating qubits $a$ in the $\fZz_{v,w}$ procedures.
We may recover the definition for flows in \autoref{def:flow} if we either set $M = \vide$, or simply require that $v \ne f(v)$ for all $v$.

The following generalization of Theorem~1 of \cite{DK06} to modified flows provides the basis for the approach we will take to define semantic maps for the DKP and RBB constructions:
\begin{lemma}
	\label{lemma:extend-DK06}
	Suppose $(f,\preceq)$ is a modified flow for $(G,I,O,M)$.
	Let $v \sim w$ denote the adjacency relation of $G$: then for any linear order $\le$ extending $\preceq$, the measurement procedure
	\begin{gather}
		\label{eqn:extend-DK06}
	 	\sqparen{\ordprod[\le]_{u \in O\comp}
			\paren{\prod_{\substack{w \sim f(u) \\ w \ne u}} \Zcorr{w}{\s[u]}}
			\paren{\prod_{\substack{v = f(u) \\ v \ne u}} \Xcorr{v}{\s[u]}}
			\Meas{u}{\theta_u}
		} \Ent{G} \New{I\comp}
	\end{gather}
	performs a unitary transformation, provided $\theta_u = \pion 2$ for every $u \in M$.
	Furthermore, for \oneway\ procedures $\fP$ arising from the DKP or (simplified) RBB constructions, the geometry $(G,I,O)$ underlying $\fP$ (together with the set $M$ of qubits with measurement angle $\pion 2$) has a modified flow.
\end{lemma}
\begin{proof}
	We may obtain a procedure as above using a map $\bar\Phi$ extending $\Phi$, as follows.
	Let $\bar\cB = \ens{J(\theta), \cZ, \Zzz_d}_{\theta \in \frac{\pi}{4} \Z,\, d \ge 2}$, where we define the $d$ qubit operator
	\begin{gather}
		\label{eqn:defZzz}
		 	\Zzz_d
		\;=\;
			\exp\big(\!-i\pi Z\sox{d} / 4\big)
		\;=\;
			\e^{-i \pi \smash{[\,\underbrace{\scriptstyle Z \ox \cdots \ox Z}_{\text{\small $d$ times}}\,]}/4}	\;.
		\\[-1em]\notag
	\end{gather}
	We may define the \oneway\ procedure $\fZzz^d$ analogously to $\fZz$, as follows:
	\begin{gather}
		\label{eqn:ZzzProcedure}
			\fZzz^d_{v_1, \ldots, v_d}
		\;=\;
			\paren{\prod_{j = 1}^d \Zcorr{v_j}{\s[a]}} \Meas{a}{\pion 2} \paren{\prod_{j = 1}^d \Ent{av_j}} \New{a}	\;:
	\end{gather}
	by an analysis deferred to \autoref{apx:constructGatesApx-2qubit} (see also Section~3.8 of \cite{BB06}), we may show
	\begin{gather}
		\label{eqn:ZzzCompute}
			\fZzz^d_{v_1, \ldots, v_d}(\rho_S)
		\;=\;
			\big[ \Zzz_d \big]_{v_1, \ldots, v_d} \:\! \rho_S \; \big[ \Zzz_d \big]\herm_{v_1, \ldots, v_d}	\;,
	\end{gather}
	for a system $S$ including the qubits $v_j$, but excluding $a$.
	For a stable-index expression for gates in $\cB$, we may then define:
 	\begin{subequations}
	\label{eqn:extendedCircuitHomphm}%
	\begin{align}
			\bar\Phi\Big(\,	\ket{+}\pseu[:v/]	\Big)
		\;=&\;\,
			\New{v} 	\;;
		\\[1ex]
			\bar\Phi\Big(\,	J(\theta)\pseu[:w/v]	\Big)
		\;=&\;\,
			\fJ^\theta_{w/v}
		\;=\;
			\Xcorr{w}{\s[v]} \Meas{v}{-\theta} \Ent{vw} \New{w}	\;;
		\\[1ex]
			\bar\Phi\Big(\,	\cZ\pseu[v,w]	\Big)
		\;=&\;\,
			\Ent{vw}	\;;
		\\[1ex]
			\bar\Phi\Big(\,	\Zzz_d\pseu[v_1, \ldots, v_d]	\Big)
		\;=&\;\,
			\fZzz^d_{v_1,\ldots,v_d}	\;,
	\end{align}
 	\end{subequations}
	extending this to unitary circuits over $\bar\cB$ in the usual manner.
	For a star circuit $C_\star = \cZ\pseu[v, w_1] \cdots \cZ\pseu[v, w_k] J(\theta_u)\pseu[:v/u]$ as in \eqref{eqn:starCircuitPremonition}, we then have
	\begin{align}
			\bar\Phi(C_\star)
	 	\;\cong&\;\,
			\paren{\prod_{j = 1}^n \Zcorr{w_j}{\s[u]}} \Xcorr{v}{\s[u]} \Meas{u}{-\theta} \Ent{G_\star} \New{v}	
	\end{align}
	after commuting the entangling maps to the right, where $G_\star$ is a star graph with center $v$ and edges from $v$ to each $w_j$ and to $u$; we may define a similar star graph for the entangling maps of $\fZzz^d_{v_1, \ldots, v_d}$ so that we may write
	\begin{align}
			\bar\Phi(\Zzz_d)
	 	\;\cong&\;\,
			\paren{\prod_{j = 1}^d \Zcorr{w_j}{\s[a]}} \Meas{a}{-\theta} \Ent{G_\star} \New{v}	
	\end{align}
	for such a star-graph.

	Given a tuple $(G,I,O,M)$ with a modified flow $(f,\preceq)$, define the circuits $C_u$ for $u \in O\comp$ via the following stable-index expressions:
	\begin{subequations}
	\begin{align}
			C_u
		\;=&\;\,
			\Bigg(\!\prod_{\substack{w \sim f(u) \\ w \notin \ens{\smash{u, f^2(u)}}}} \!\! \cZ\pseu[f(u), w]\Bigg) \, J(\theta_u)\pseu[:f(u)/u]	\,,
		&
			\text{for $u \ne f(u)$};
		\intertext{}
			C_u
		\;=&\;\,
			\big(\Zzz_{d}\big)_{w_1, \ldots, w_d}	\,,
		&
			\mspace{-170mu}\text{for $u = f(u)$ with neighbors $w_1, \ldots, w_n$}.
	\end{align}
	\end{subequations}
	We also define a circuit $C_{\tilde I}$ by
	\begin{gather}
		\label{eqn:inputStarPremonition}
			C_{\tilde I}
		\;\;=\;
			\prod_{\substack{uv \in E(G) \\ u,v \notin \img(f)}}	\cZ\pseu[u,v]	\;:
	\end{gather}
	we may describe the $\cZ$ operations in this circuit by the induced subgraph $G[\tilde I]$, where $\tilde I = V(G) \setminus \img(f)$.
	Consider the circuit $\cC$ described by the stable-index expression consisting of the concatenation of these circuits. 
	Placing the circuit $C_{\tilde I}$ at the beginning, and ordering the the remaining circuits $C_u$ in a linear order $\le$ extending $\preceq$, we obtain an expression in which operators with index $u \ne f(u)$ occur no later than any operator involving $f(u)$ or $w \sim f(u)$, and any operator with index $u = f(u)$ occurs before the indices $w \sim u$ are deprecated.
	Thus, applying the homomorphism $\bar\Phi$ to each sub-circuit $C_u$ in $\cC$ yields a well-formed \oneway\ procedure
	\begin{gather}
		\mspace{-10mu}
		\fP
		\;\cong\;
	 	\sqparen{\ordprod[\le]_{u \in O\comp} \!
			\paren{\prod_{\substack{w \sim f(u) \\ w \notin \ens{\smash{u, f^2(u)}}}} \!\!\!	\Zcorr{w}{\s[u]}}
			\paren{\prod_{\substack{v = f(u) \\ v \ne u}} \!	\Xcorr{v}{\s[u]}}
			\Meas{u}{\theta_u}
		\Ent{G_u} \New{f(u)}	} \Ent{G[\tilde I]}\;,
	\end{gather}
	where $G_u$ is the star graph for each qubit $u$ as described above, and the congruence arises from commutation of preparation maps and entangling maps to the right of each term in the ordered product.

	Note that for every edge $vw \in E(G)$, we have either $v,w \notin \img(f)$, or one of the vertices (without loss of generality, $v$) is of the form $f(u)$ for some $u \in O\comp$.
	In the former case, $vw$ corresponds to a $\cZ$ operation in $C_{\tilde I}$, and thus to an entangling operation in $\bar\Phi(C_{\tilde I})$; in the latter case, we have $w \sim f(u)$, in which case $vw$ corresponds to an entangling operation in the star-graph $G_u$ and thus to an entangling operation in $\bar\Phi(C_u)$.
	Conversely, every entangling operation in $\fP$ above corresponds to some edge in $G$, by construction; the geometry underlying $\fP$ is then $(G,I,O)$.
	Note also that the individual operations in $\fP$ and the procedure described in \eqref{eqn:extend-DK06} are identical, except that corrections of the form $\Zcorr{f^2(u)}{\s[u]}$ do not occur in $\fP$; commuting all the preparation and entangling maps to the right of the expression, these are induced by commuting operations $\Ent{f(u), f^2(u)}$ past the corrections $\Xcorr{f(u)}{\s[u]}$.
	Thus, \eqref{eqn:extend-DK06} is the result of performing this commutation to $\fP$; the \oneway\ procedure in \eqref{eqn:extend-DK06} thus performs the same unitary transformation as $\cC$.

	Finally, we may show that every geometry underlying a \oneway\ procedure obtained from the DKP or (simplified) RBB constructions has a modified flow.
	Consider a stable-index expression $C$ for a circuit over the gate-set $\cB_\DKP \union \cB_\RBB$, whose terms of $C$ are ordered according to the order of performance in the circuit.
	Let $f$ be the extended successor function of $C$, and $\le$ be the order in which the indices are deprecated in $C$.
	Then $f$ will be a mapping from the deprecated indices to the advanced indices, and $(f,\le)$ satisfies \eqref{eqn:flowPremonition} by construction.
	If $(G,I,O)$ is the geometry underlying the procedure $\Phi(C)$, and $M$ is the set of qubits with measurement angle $\pion 2$, by construction $(f,\le)$ will be a modified flow for $(G,I,O,M)$; as the DKP and RBB constructions both involve producing such a stable-index expression $C$, and the process of normalizing a \oneway\ procedure leaves the geometry invariant, the result holds.
\end{proof}
The foregoing proof not only demonstrates that a \oneway\ procedure as in \eqref{eqn:extend-DK06} performs a unitary transformation, but gives a description of a unitary circuit performing the same transformation.
This will form the basis of a similar construction in \autoref{sec:candidateConstr} for constructing circuits without reference to corrections or measurement dependencies.

The proof above also describes a construction of \oneway\ procedures based on a gate-set which subsumes $\cB_\DKP \union \cB_\RBB$, containing operations acting on arbitrarily many qubits.
The extension to this gate-set $\bar\cB$ from $\cB_\DKP$ and $\cB_\RBB$ corresponds to the way in which modified flows extend the pattern of dependencies found in either the DKP or RBB constructions without normalization; 
the inclusion of the operators $\Zzz_d$ for $d > 2$ is essentially a byproduct of the lack of further constraints on the definition of modified flows.

\subsubsection{The natural pre-order, and uniqueness in the case $\card I = \card O$}

Following~\cite{Beaudrap08}, for any geometry (with or without a modified flow) and for any function $f:O\comp \to I\comp$, we may consider an ordering satisfying the conditions of \eqref{eqn:flow} or \eqref{eqn:mflow} as follows:
\begin{definition}
	\label{def:naturalPreorder}
	Let $(G,I,O)$ be a geometry and $f: O\comp \to I\comp$.
	The \emph{natural pre-order} $\preceq$ for $f$ is the reflexive and transitive closure of the conditions
	\begin{subequations}
	\label{eqn:natural}
	\begin{gather}
 			\label{natural:a}
			v \preceq f(v)
		\\
 			\label{natural:b}
			w \sim f(v) \;\;\implies\;\; v \preceq w
	\end{gather}
	\end{subequations}
	for all $v \in O\comp$ and $w \in V(G)$.
\end{definition}
Note that the conditions of~\eqref{eqn:natural} are exactly the conditions \eqref{mflow:c} and \eqref{mflow:d} for modified flows.
The definition of the natural pre-order then relaxes the condition that $\preceq$ be a partial order, \ie\ that $\preceq$ is antisymmetric.

If the natural pre-order $\preceq$ happens also to be a partial order, then $(f,\preceq)$ is a modified flow.
By the definition of the transitive closure, any pre-order satisfying the conditions of \eqref{eqn:natural} must also contain the relations of the natural pre-order $\preceq$; then if $\preceq$ is not antisymmetric, there can be no antisymmetric relation $\le$ satisfying the modified flow conditions with $f$.

Because of the distinguished nature of the natural pre-order, we may prove a useful uniqueness result for modified flows: 
\begin{theorem}
	\label{thm:modFlowUniqueness}
	Let $(G,I,O)$ be a geometry and $M \subset O\comp$.
	If $\card{I} = \card{O}$, then there is at most one function $f$ such that $(f,\preceq)$ is a modified flow for $(G,I,O,M)$, for $\preceq$ the natural pre-order of $f$. 
\end{theorem}
\begin{proof}
 	Suppose that $\card{I} = \card{O}$.
	If $I = O$, then a tuple $(f,\preceq)$ can only be a flow if $f(v) = v$ for every $v \in O\comp$, as otherwise we obtain cycles $v \prec f(v) \prec \cdots \prec f^k(v) = v$ for some $v$.
%
	Then the only function $f$ for which $(f,\preceq)$ may be a modified flow is the identity function on $O\comp$.

	Otherwise, suppose that $I \ne O$, and consider two functions $f, g: O\comp \to I\comp$ such that $f(v) = v$ only for $v \in M$, and similarly for $g$.
	Note that the natural pre-order $\preceq$ of $f$ can only be anti-symmetric if there are no cycles of the form $v \prec f(v) \prec \cdots \prec f^k(v) = v$; thus, we require that $f$ have no such cycles for $k > 1$, and again similarly for $g$.

	Suppose $f$ and $g$ differ, and 
	let $u_0 \in O\comp$ be 
	some vertex for which $f(u_0) \ne g(u_0)$.
	Let $\tilde g: I\comp \to O\comp$ be the inverse function of $g$: we may recursively construct an unbounded sequence of vertices $u_j$ for $j \in \N$ by defining $u_1 = f(v_j)$, and
	\begin{gather}
		 	u_{j+1}
		\;=\;
			\begin{cases}
			 	f(u_j)	\;,			&	\text{if $u_j \ne f(u_{j-1})$}	\\[0.2ex]
				\tilde g(u_j)	\;,	&	\text{if $u_j = f(u_{j-1})$}
			\end{cases}
	\end{gather}
	for $j \ge 1$.
	Note that $\dom(\tilde g) = I\comp = \img(f)$ and $\dom(f) = O\comp = \img(\tilde g)$; then this sequence is well defined for all $j \ge 0$.
	We may also show that $u_{j+1} \ne u_{j-1}$ for each $j \ge 1$ by induction, as follows.
	Suppose that $u_{j+1} \ne u_{j-1}$ for some particular $j \ge 1$:
	\begin{itemize}
	\item 
		Suppose that $u_{j+1} = \tilde g(u_j) \ne f(u_j)$: it follows that $u_j = f(u_{j-1})$, so that $u_{j+2} = f(u_{j+1}) \ne f(u_{j-1}) = u_j$ by the injectivity of $f$.

	\item
		Suppose that $u_{j+1} = f(u_j) \ne \tilde g(u_j)$: it follows that $u_j = \tilde g(u_{j-1}) \ne f(u_{j-1})$, in which case we have $u_{j+2} = \tilde g(u_{j+1}) \ne \tilde g(u_{j-1}) = u_j$, by the injectivity of $g$.

	\item
		Suppose that $u_{j+1} = f(u_j) = \tilde g(u_j)$.
		Because $f$ is injective, we have $u_{j+1} = f(u_j) \ne u_j$ if $u_j \ne u_{j-1}$; and as $u_{j+1} \ne u_{j-1}$, at least one of the inequalities $u_{j+1} \ne u_j$ or $u_j \ne u_{j-1}$ must hold.
		Thus, we have $u_{j+1} = \tilde g(u_j) \ne u_j$ in any case; by the injectivity of $\tilde g$, we then have $u_{j+2} = \tilde g(u_{j+1}) = \tilde g(\tilde g(u_j)) \ne u_j$ by the acyclic property of $g$.
	\end{itemize}
	It then follows by induction that $u_{j+2} \ne u_j$ for all $j \ge 0$.

	We may construct a subsequence $u_{\sigma(0)} \prec u_{\sigma(1)} \prec \cdots$, with $\sigma: \N \to \N$ a monotonically increasing function, as follows.
	We let $\sigma(0) = 0$ and for each $j \ge 0$, and define 
	\begin{gather}
			\sigma(j+1)
		\;=\;
			\begin{cases}
				\sigma(j) + 1	\;,	&	\text{if $u_{\sigma(j) + 2} = f(u_{\sigma(j) + 1}) = f^2(u_{\sigma(j)})$}		\\
				\sigma(j) + 2	\;,	&	\text{otherwise}
			\end{cases}
	\end{gather}
	for all $j \ge 0$.
	Suppose for some $j \ge 1$ that $u_{\sigma(j) + 1} = f(u_{\sigma(j)})$.
	\begin{itemize}
	\item 
		If $u_{\sigma(j) + 2} = f(u_{\sigma(j) + 1})$ as well, we have $u_{\sigma(j)} \prec u_{\sigma(j) + 1} = u_{\sigma(j+1)}$, and we also have $u_{\sigma(j+1) + 1} = f(u_{\sigma(j+1)})$.

	\item
		Otherwise, we have $u_{\sigma(j+1)} = u_{\sigma(j) + 2} = \tilde g(u_{\sigma(j) + 1})$, which implies that either $u_{\sigma(j) + 2} = u_{\sigma(j) + 1} = f(u_{\sigma(j)})$ or $u_{\sigma(j) + 2} \sim u_{\sigma(j) + 1} = f(u_{\sigma(j)})$.
		In either case, we again have $u_{\sigma(j)} \prec u_{\sigma(j+1)}$, and once more we have $u_{\sigma(j+1) + 1} = f(u_{\sigma(j+1)})$.
	\end{itemize}
	As $u_{\sigma(0) + 1} = u_1 = f(u_0)$, by induction we then have $u_{\sigma(k)} \prec u_{\sigma(k+1)}$ for all $k \ge j$.
	As $V(G)$ is finite, there must be integers $N, N'$ such that $N < N'$ and $u_{\sigma(N)} = u_{\sigma(N')}$: then we have $u_{\sigma(N)} \preceq u_{\sigma(N+1)} \preceq \cdots \preceq u_{\sigma(N')} = u_{\sigma(N)}$, where $u_{\sigma(N)} \ne u_{\sigma(N+1)}$ as noted above.
	Then $\preceq$ is not antisymmetric.

	If on the contrary $(f,\preceq)$ is a modified flow, and $(g,\le)$ is also a modified flow for some partial order $\le$, it follows that $f = g$.
	Thus $(G,I,O)$ and $M$ have at most one modified flow $(f,\preceq)$, for $\preceq$ the natural pre-order of $f$.
\end{proof}
The uniqueness of a modified flow $(f,\preceq)$, when we require that $\preceq$ be the natural pre-order of $f$, will play an important role in describing a semantic map for the DKP and RBB constructions.
Different ``flow-functions'' $f$ correspond to different sets of dependencies in a \oneway\ procedure: provided a procedure $\fP$ in which some of the dependencies have been obfuscated or eliminated altogether by Pauli simplifications and signal shifting, it may be difficult to determine whether there exists a successor function which describes the actual dependencies which are present.
This problem is simplified if we are promised that there is at most one successor function $f$ yielding a modified flow $(f,\preceq)$: we may simply attempt to construct $f$, and determine whether the dependendices in $\fP$ are consistent with those described by $f$.

\subsection{Decompositions of geometries}
\label{sec:starDecomp}

The proof of \autoref{lemma:extend-DK06} suggests an approach to determining when a tuple $(G,I,O,M)$ has a modified flow, by means of a decomposition of the geometry $(G,I,O)$ into sub-geometries.
Examining such decompositions will lead to an efficient algorithm for obtaining the candidate successor function $f$ associated with a modified flow, which we may use to implement a semantic map for the DKP and RBB constructions.

\subsubsection{Star geometries and star decompositions}
\label{sec:starDecomposition}

For a geometry $(G,I,O)$, note that the natural pre-order of a injective function $f: O\comp \to I\comp$ is characterized entirely by the transitive closure of local relations of vertices in $G$.
Similar local relations are used in the proof of \autoref{lemma:extend-DK06} to define star graphs which are used to define both star circuits, and their representations as procedures in the \oneway\ model.
For an injective function $f: O\comp \to I\comp$, we may characterize these local relations in terms of the following structures:
\begin{definition}
	\label{def:starGeom}
	Let $(G,I,O)$ be a geometry, and $f: O\comp \to I\comp$ an injective function.
	For $v \in V(G)$ define the vertex-set
	\begin{align}
			V_{v \star}
		\;=&\;
			\ens{w \in V(G) \,\Big|\, w \in \ens{v, f(v)} \text{~or~} \big[w \sim f(v) \text{~and~} w \ne f(f(v)) \big]} \,.
	\end{align}
	We then define the \emph{star geometry 
	rooted at $v$} (or \emph{centered at $f(v)$}) is the geometry $(G_{v\star}, I_{v\star}, O_{v\star})$ given by
	\begin{subequations}
	\label{eqn:starGeometry}
	\begin{gather}
			G_{v \star}
		\;\;=\;\;
			G[V_{v \star}]	\,-\, \ens{xw \,\big|\, x,w \ne f(v)}\,,
		\\[1.5ex]
			I_{v \star}
		\;\;=\;\;
				V_{v \star} \setminus \ens{f(v)}	\,,
		\\[1ex]
			O_{v \star}
		\;\;=\;\;
				V_{v \star} \setminus \ens{v}	\,.
	\end{gather}
	\end{subequations}
\end{definition}

Star geometries correspond to star circuits (as described in \autoref{sec:dependenciesIndexSuccFns}) and $\Zz$ operations, via their representations as \oneway\ procedures.
Consider the \oneway\ procedure $\fP = \Phi(C_\star\pseu[u_1,\ldots,u_n:w/v])$ obtained from a stable-index expression for star circuit $C_\star$ as in \eqref{eqn:starCircuitPremonition}.
Then $\fP$ has a geometry $(G, I, O)$ consisting of a star graph $G$ with center $w \in O \setminus I$, and where $v \in I \setminus O$.
If $f$ is the index successor function of the stable-index expression, we may then show that $(G,I,O)$ is the star geometry with root $v$ and center $w$ as in \autoref*{def:starGeom}.
We may similarly show that $\Phi(\Zz)$ also gives rise to a star geometry, where the mediating qubit $a = f(a)$ is both the center and the root of the star geometry.
(The same also holds for $\bar\Phi(\Zzz_d)$, for the mapping $\bar\Phi$ and unitary $\Zzz_d$ defined in the proof of \autoref{lemma:extend-DK06}.)

Star geometries $(G_{u\star},I_{u\star},O_{u\star})$ are thus precisely those geometries which underlie the \oneway\ procedures $\Phi(C_\star)$, where $C_\star$ are the featured circuits in decompositions described in \autoref{sec:dependenciesIndexSuccFns}.
Extending the results of\cite{DK06}, we may consider how such structures for a geometry $(G,I,O)$ may be used to describe ``candidate'' circuits $C$ for \oneway\ procedures $\fP$, where $C$ may correspond to $\fP$ via the DKP or RBB constructions: we consider how to construct such circuits in \autoref{sec:candidateConstr}.
Motivated by the connection with circuit constructions, we may consider a decomposition of geometries similar to the compositions of simple circuits in \autoref{sec:dependenciesIndexSuccFns}:
\begin{lemma}
	\label{lemma:modifiedStarDecomp}
	Let $(G,I,O)$ be a geometry and $M \subset O\comp$.
	Then $(G,I,O,M)$ has a modified flow $(f,\preceq)$ if and only if there is a decomposition of $(G,I,O)$ into star geometries rooted at each of the vertices $v \in O\comp$, together with the geometry $\big(G\big[\:\!\tilde{I}\:\!\big], I, \tilde{I}\,\big)$, where $\tilde{I} = V(G) \setminus \img(f)$.
\end{lemma}
\begin{proof}
 	We proceed by induction on the size of $\img(f)$.
	If $\img(f) = \vide$, then $\dom(f) = \vide$ as well, in which case $I \subset V(G) = O$; we then have $(G,I,O) = \big(G\big[\,\tilde{I}\,\big], I, \tilde{I}\,\big)$.
	The natural pre-order in this case is the equality relation, which is antisymmetric; then $(f,=)$ is a modified flow.
	The Lemma then holds for $\img(f) = \vide$.

	For the induction step, suppose that $\card{\img(f)} = N+1$ for some $N \in \N$, and that the claim holds for geometries $(G,I,O)$ with flows $(f',\le)$ in which $\card{\img(f')} \le N$.
	\begin{itemize}
	\item
		Suppose that $(G,I,O,M)$ has a modified flow: in particular, the natural pre-order $\preceq$ is a partial order.
 		Because $V(G)$ is finite, there is a vertex $v \in V(G)$ such that $v \prec w$ only for $w$ maximal in the natural pre-order $\preceq$.
		Let $\cG_{v\star} = (G_{v\star}\,, I_{v\star}\,, O_{v\star})$, and $\cG' = (G', I, O')$, where
		\begin{subequations}
		\label{eqn:modifiedGeometryReduce}
		\begin{align}
				G'
			\;=&\;\,
				G \setminus f(v)	\,,\quad\text{and}
			\\
				O'
			\;=&\;\,
				\begin{cases}
					O	\,,	&	\text{if $v = f(v)$}	\\
					O \symdiff \ens{v, f(v)}	\,,	&	\text{otherwise}
				\end{cases}.
		\end{align}
		\end{subequations}
	 	(Again, $A \symdiff B$ denotes the symmetric difference of $A$ and $B$.)
		Note that $O_{v\star} \subset O$ by definition, and either $I_{v\star} = O_{v \star} \symdiff \ens{v, f(v)} \subset O \symdiff \ens{v, f(v)}$ (if $v$ and $f(v)$ differ) or $I_{v\star} = O_{v\star} \subset O$ (if $f(v) = v$).
		Because $V(G_{v\star}) = I_{v\star} \union \ens{f(v)}$ in both cases, we then have $V(G') \inter V(G_{v\star}) \subset I_{v\star} \inter O'$\,: then the composition $(\bar{G}, \bar{I}, \bar{O})  \,=\, \cG_{v\star} \circ \cG'$ is well-defined, and we have
		\begin{subequations}
		\begin{align}
				\bar{G}
			\;\;=&\;\;
				G_{v\star} \union G' 
			\notag\\=&\;\;
				\ens{xw \in G \,\big|\, w = f(v)} \union \Big[ G \setminus f(v) \Big] 
			\;\;=\;\;
				G	\,;
			\\[2ex]
				\bar{I}
			\;\;=&\;\;
				\big( I_{v\star} \setminus O' \big) \union I
			\;\;=\;\;
				\vide \union I
			\;\;=\;\;
				I	\,;
			\\[2ex]
				\bar{O}
			\;\;=&\;\;
				O_{v\star} \union \big( O' \setminus I_{v\star} \big)
			\;\;=\;\;
				O_{v\star} \union \big( O \setminus O_{v\star} \big)
			\;\;=\;\;
				O	\,.
		\end{align}
		\end{subequations}
		Then, we have $(G,I,O) = \cG_{v\star} \circ \cG'$\,.

		The restriction $f'$ of the flow-function $f$ to $V(G) \setminus O'$ does not include $f(v)$ in its image, and therefore has cardinality at most $N$\,; by recursively applying the decomposition above, $\cG'$ has a decomposition into star geometries and the geometry $(G' \setminus \img(f'), I, V(G') \setminus \img(f'))$, where in particular $V(G) \setminus \img(f') = V(G) \setminus \img(f) = \tilde{I}$, by the definition of $f'$\,.
		Thus, for any tuple $(G,I,O,M)$ with a modified flow, $(G,I,O)$ has a decomposition into star geometries as above.

	\item
		For the converse, suppose that $(G,I,O) = \cG_{v\star} \circ \cG'$ for some star geometry $\cG_{v\star} = (G_{v\star}, I_{v\star}, O_{v\star})$ and a residual geometry $\cG' = (G',I,O')$.
		If the latter has a star decomposition, and if $f'$ is the restriction of $f$ to $V(G')$, then $(f', \preceq)$ is a modified flow for $\cG'$ by hypothesis as $\img(f') = \img(f) \setminus \ens{f(v)}$.
		By assumption, $\cG_{v\star}$ and $\cG'$ are composable, in which case we have $V(G_{v\star}) \inter V(G_{v\star}) = I_{v\star} \inter O'$.
		If $v \prec w$ for some $w \in V(G)$, we then have either $w = f(v)$, $w \in O' \setminus{v}$, or $w \succeq w'$ for either $w' = f(v)$ or $w' \sim f(v)$.
		By construction, we have either $O = O' \setminus {v}$ in the case that $v = f(v)$, or $O = (O' \setminus \ens{v}) \union \ens{f(v)}$ otherwise; then for any $w' = f(v)$ or $w' \sim f(v)$, we have $w' \in O$.
		As such $w'$ lie outside of the domain of $f$, they are maximal in the natural pre-order; then $w \succ v$ only if $w \in O$.
		In particular, if the restriction of the natural pre-order to $\cG'$ is such that $(f',\preceq)$ is a modified flow, then $(f,\preceq)$ is also a modified flow.
		Thus, the existence of a star decomposition for $(G,I,O)$ implies that it has a modified flow, by induction.	\qedhere
	\end{itemize}
\end{proof}
Thus, we may decompose a circuit with a modified flow into star geometries.
\autoref{fig:starDecomp} illustrates a decomposition of a geometry with a flow into star geometries in the manner described above; \autoref{fig:starModDecomp} illustrates a more general case for a modified flow.
Also illustrated in these figures are ``candidate circuit'' constructions for unitary circuits, which we describe in \autoref{sec:candidateConstr}.

\begin{figure}[tp]
	\begin{center}
		\includegraphics{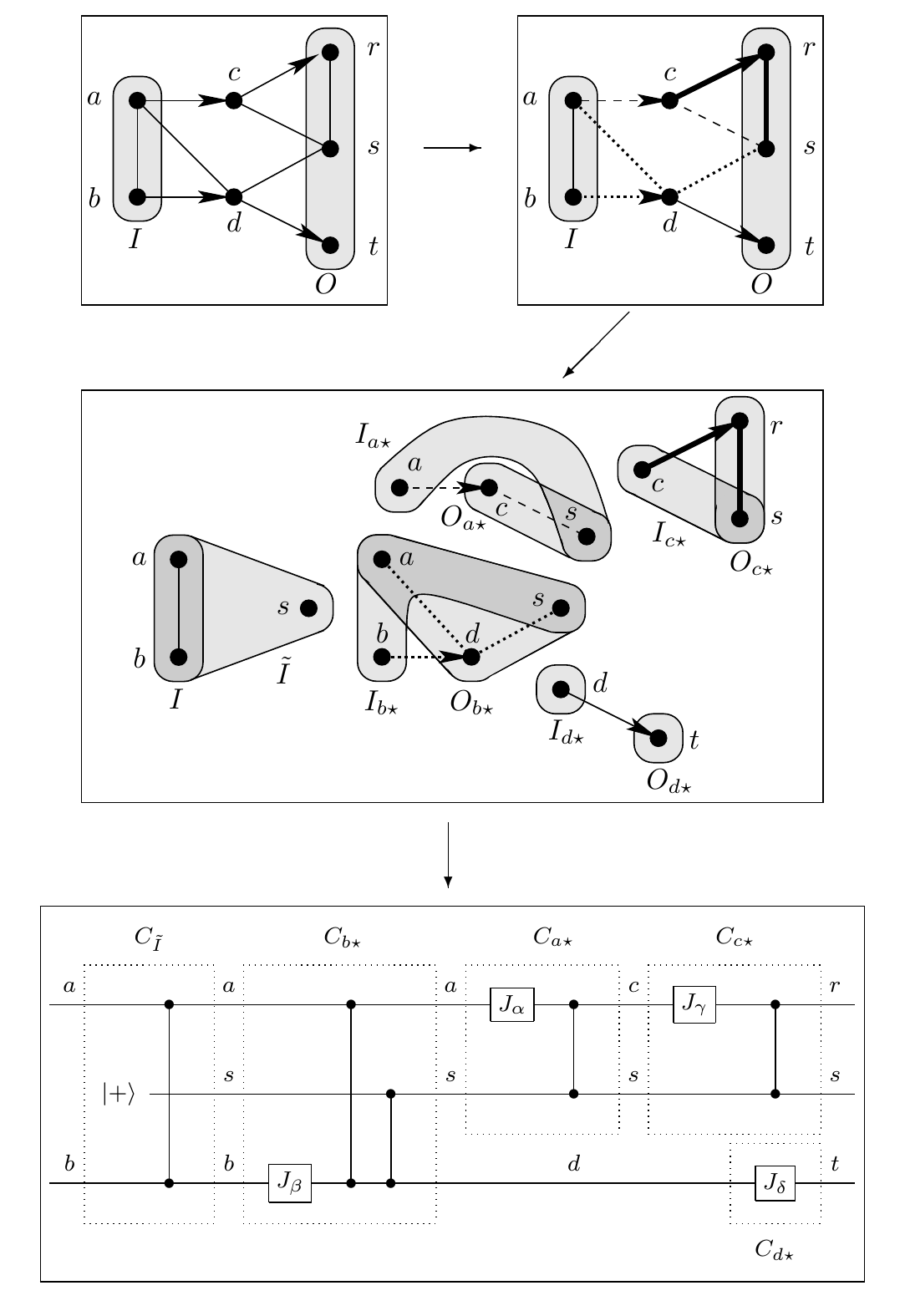}
 		\vspace{-2.2em}
	\end{center}
	\caption[Illustration of a decomposition of a geometry with a flow into star geometries.]{\label{fig:starDecomp}
		Illustration of a decomposition of a geometry $(G,I,O)$ with a flow $(f,\preceq)$ 
		into star geometries for each $v \in O\comp$, and a reduced geometry $\big(G\big[\:\!\tilde I\:\!\big], I, \tilde I\big)$ for $\tilde{I} = V(G) \setminus \img(f)$.
		In the upper three panels, the action of $f$ is represented by arrows; different styles of edges indicate edges belonging to different star geometries.
		The middle panel illustrates each star geometry in the decomposition separately (with vertices repeated between different star geometries), adjacent to the geometries with which they may be composed.
		The bottom panel illustrates the unitary circuits corresponding to these star geometries (including a labelling of wire-segments corresponding to stable tensor indices), and the complete circuit obtained from composing them.}
\end{figure}

\begin{figure}[tp]
	\begin{center}
		\includegraphics{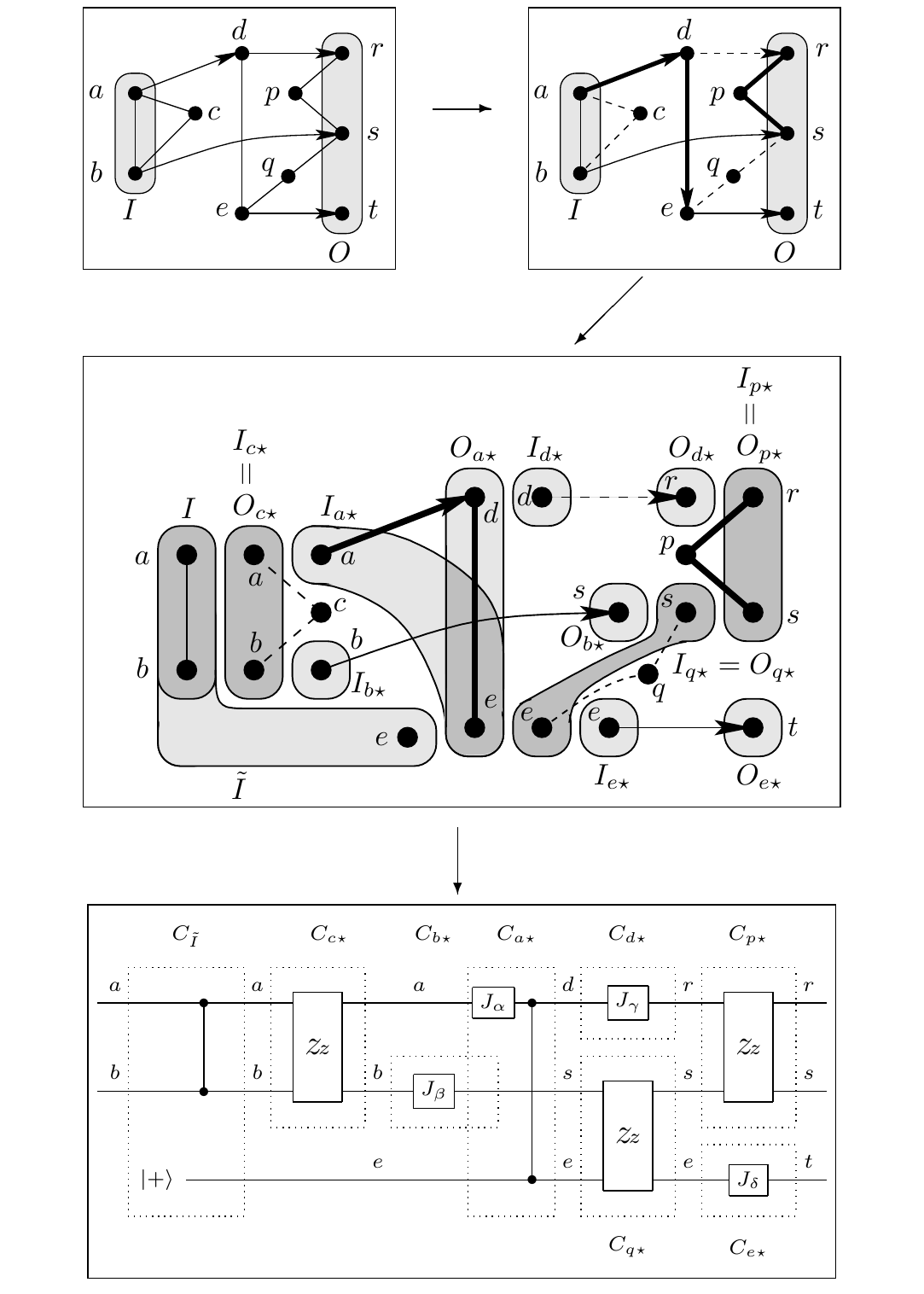}
 		\vspace{-2.2em}
	\end{center}
	\caption[Illustration of a decomposition of a geometry $(G,I,O)$ with a modified flow $(f,\preceq)$ into star geometries.]{\label{fig:starModDecomp}
		Illustration of a decomposition of a geometry $(G,I,O)$ with a modified flow $(f,\preceq)$ into star geometries for each $v \in O\comp$, and a reduced geometry $\big(G\big[\:\!\tilde I\:\!\big], I, \tilde I\big)$ for $\tilde{I} = V(G) \setminus \img(f)$.
		In the upper three panels, the action of $f$ is represented by arrows.
		Different styles of edges indicate edges belonging to different star geometries (as are edges of the same style, but which are disconnected from each other).
		Arrows are omitted for $f(v) = v$.
		The middle panel illustrates each star geometry in the decomposition separately (with vertices repeated between different star geometries), adjacent to the geometries with which they may be composed.
		The bottom panel illustrates the composition of the unitary circuits corresponding to the star geometries (including a labelling of wire-segments corresponding to stable tensor indices).}
\end{figure}

\subsubsection{Structural characterization of star decompositions}

In order to define a semantic map $\cS$ for the DKP and RBB constructions, it will be necessary to attempt to obtain a star decomposition of a geometry which is initially not known to have a modified flow.
We may facilitate this by showing that we may characterize star decompositions independently of modified flows, as follows:
\begin{lemma}
	\label{lemma:flowStarDecompCharn}
	Let $(G,I,O)$ be a geometry and $M \subset O\comp$.
	Suppose that $(G,I,O)$ can be expressed as $\cG_{\ell\star} \circ \cG_{(\ell-1)\star} \circ \cdots \circ \cG_{1\star} \circ \cG_{\tilde I}$, where $\cG_{\tilde I} = (G[\tilde I], I, \tilde I)$ for some subset $\tilde I \subset V(G)$ with $\card{\smash{\tilde I}} = \card{O} = \card{V(G)} - \ell$, and where for each $0 \le j < \ell$ the geometry $\cG_{j\star} = (G_j, I_j, O_j)$ satisfies either
	\begin{romanum}
	\item
		\label{item:sameRootCenterStar}
		$I_j = O_j = V(G_j) \setminus \ens{v_j}$, for some $v_j \in V(G) \inter M$ which is adjacent to every other element of $V(G_j)$;

	\item
		\label{item:diffRootCenterStar}
		$V(G) = I_j \union O_j$, $I_j = (I_j \inter O_j) \union \ens{v_j}$, and $O_j = (I_j \inter O_j) \union \ens{w_j}$, for some distinct $v_j, w_j \in V(G)$ such that $w_j$ is adjacent to every other element of $V(G_j)$.
	\end{romanum}
	Let $f: O\comp \to I\comp$ be defined so that $f(v_j) = v_j$ for any $v_j$ as in the first case above, and $f(v_j) = w_j$ for any $v_j,w_j$ as in the second case; and let $\preceq$ be the natural pre-order for $f$.
	Then $(f,\preceq)$ is a modified flow for $(G,I,O,M)$.
\end{lemma}
\begin{proof}
	Let $f$ be defined as above, and let $\preceq$ be its natural pre-order.
	By construction, each geometry $\cG_j$ is then a star geometry with respect to $f$ in the sense of \autoref{def:starGeom}, whether or not $(f,\preceq)$ is a modified flow.
	We may show by induction on $\ell = \card{V(G)} \setminus \card{O}$ that $\preceq$ is a partial order.
	\begin{itemize}
	\item
		If $\ell = 0$, then $O = V(G)$, in which case $\dom(f) = \vide$.
		Then the natural pre-order is the equality relation, which is antisymmetric and therefore a partial order.

	\item
		Suppose the Lemma holds for all geometries with fewer vertices than $(G,I,O)$, but with an output subsystem of the same size.
		By hypothesis, we have $(G,I,O) = (G_\ell, I_\ell, O_\ell) \circ \cG'$, where we may write
		\begin{gather}
			\cG' \;=\; (G', I, O') \;=\; \cG_{(\ell-1)\star} \circ \cdots \circ \cG_{1\star} \circ \cG_{\tilde I}	\;;
		\end{gather}
		in particular, we have $V(G') \inter V(G_\ell) \subset I_\ell \inter O'$.
		If $f'$ is the restriction of $f$ to $O'$, and $\preceq'$ is the restriction of $\preceq$ to $V(G)$, the pair $(f', \preceq')$ is then a modified flow for $\cG'$.
		\begin{romanum}
		\item
			Suppose $f(v_\ell) = v_\ell$: then $V(G) \setminus V(G') = \ens{v_\ell}$, and we have $w \in O' = O$ for every $w \sim v_\ell$.
			Then the only relations that $\preceq$ contains that $\preceq'$ does not are $u \preceq v_\ell$ for qubits such that $f(u) \in O$, and $v_\ell \preceq w$ for $w \in O$; that is, $u \preceq v_\ell \preceq w$ for some $w$ maximal and $u$ bounded above in $\preceq'$ only by maximal elements.

		\item
			Suppose $f(v_\ell) = w_\ell \ne v_\ell$: then $v_\ell \in O' \setminus O_\ell$, and $w_\ell \in O_\ell \setminus O'$.
			Then the only relations that $\preceq$ contains that $\preceq'$ does not are $u \preceq w_\ell$ for qubits such that $f(u) = w$ such that $w \sim w_\ell$, and $v_\ell \preceq w$ for qubits $w \sim w_\ell$.
			In either case, we have $w_\ell \in O'$; thus we only add additional (maximal) upper bounds for $u$ bounded only bounded by maximal elements, and new upper bounds for $v_\ell$ which are also maximal.
		\end{romanum}
		In either case, $\preceq$ is a partial order; thus $(f,\preceq)$ is a modified flow.	\qedhere
	\end{itemize}
\end{proof}
We will call a decomposition as in \autoref{lemma:flowStarDecompCharn} a \emph{star decomposition} of a geometry $(G,I,O)$, and the associated geometries $\cG_j$ in the decomposition \emph{star geometries}, whether or not they are defined with respect to a corresponding modified flow $(f,\preceq)$.
More generally, we will call a geometry $(G_\star,I_\star,O_\star)$ a star geometry when the input and output subsystems satisfies either property~\ref*{item:sameRootCenterStar} or property~\ref*{item:diffRootCenterStar} in the statement of \autoref{lemma:flowStarDecompCharn} above.
In the case that $I_\star = O_\star$, we say that unique vertex $v \in V(G_\star) \setminus O_\star$ is both the root and the center of $(G_\star,I_\star,O_\star)$; otherwise, we say that the unique vertex $v \in I_\star \setminus O_\star$ is the root of $(G_\star,I_\star,O_\star)$, and the unique vertex $w \in O_\star \setminus I_\star$ is the center.
  
The proof of \autoref{lemma:flowStarDecompCharn} suggests a recursive decomposition of a geometry $(G,I,O)$ into star geometries using ``maximal'' star geometries, in the following sense:
\begin{definition}
	For a geometry $(G,I,O)$, a \emph{maximal star geometry} is a star geometry $(G_\star, I_\star, O_\star)$ such that $(G,I,O) = (G_\star, I_\star, O_\star) \circ (G', I, O')$ for some sub-geometry $(G',I,O')$.
	We call the latter the \emph{residual geometry} of this decomposition.
\end{definition}
\noindent We may characterize a maximal star geometry as follows:
\begin{lemma}
	Let $(G_\star, I_\star, O_\star)$ be a maximal star geometry, with root $v$ and center $w$, for some geometry $(G,I,O)$.
	Then one of the following two properties hold:
	\begin{romanum}
	\item 
		\label{item:sameRootCenterMaxStar}
	 	$v = w \in I\comp \inter O\comp$, and is adjacent in $G$ only to elements of $O$;

	\item
		\label{item:diffRootCenterMaxStar}
		$v \ne w$, $w \in O$, and $v \in O\comp$ is the only neighbor of $w$ in $G$ which is not in $O$.
	\end{romanum}
	Conversely, for a geometry $(G,I,O)$ and vertices $v, w \in V(G)$ for which either of the above hold, the geometry $(G_{v\star}, I_{v\star}, O_{v\star})$ such that
		$G_{v\star} = G[W]$ for $W$ consisting of $w$ and its neighbors in $G$, with input and output subsystems 
%
		$I_{v\star} = V(G_{v\star}) \setminus \ens{w}$ and
%
		$O_{v\star} = V(G_{v\star}) \setminus \ens{v}$,
	is a maximal star geometry for $(G,I,O)$.
\end{lemma}
\begin{proof}
	We first show that any maximal star geometry has either property~\ref*{item:sameRootCenterMaxStar} or property~\ref*{item:diffRootCenterMaxStar}.
	By the definition of a maximal star geometry, we have $O_\star \subset O$; as $v \notin O_\star$, we have $v \in O\comp$.
	If $v = w$, by definition $v$ is adjacent to every element of $I_\star = O_\star$, and every neighbor of $v$ in $G$ is an element of $I_\star$; thus $v$ is adjacent only to elements of $O$.
	Also, in the case that $v = w$, we have $v \notin I_\star$ by definition, in which case $v \in I\comp$ as well.
	Otherwise, we have $w \in O_\star \subset O$ and $v \in I_\star \setminus O\star$, so that $v \in O\comp$.
	Also, if $v \ne w$ by definition, all edges incident to $w$ in $G$ are of the form $uw$, where either $u = v$ or $u \in O_\star \subset O$; then $v$ is the unique element of $O\comp$ to which $w$ is adjacent in $G$.

	Conversely, suppose that $v$ and $w$ are vertices satisfying either property~\ref*{item:sameRootCenterMaxStar} or property~\ref*{item:diffRootCenterMaxStar}, and let $(G_{v\star}, I_{v\star}, O_{v\star})$ be the geometry described.
	Let $G' = G \setminus w$ and $O' = O \setminus \ens{w}$.
	In either case, we may show that $(G,I,O) = (G_{v\star}, I_{v\star}, O_{v\star}) \circ (G',I,O')$, and that $(G_{v\star}, I_{v\star}, O_{v\star})$ is also a star geometry.
	The Lemma then holds.
\end{proof}

By definition, any geometry with a star decomposition has a maximal star geometry; we may then attempt to obtain a star decomposition by recursively constructing maximal star geometries for successive residual geometries.
In the next section, we show how this leads to an efficient algorithm for obtaining a star decomposition (and a modified flow) for an arbitrary geometry $(G,I,O)$ and subset $M \subset O\comp$, when one exists.

\subsubsection{Efficiently decomposing geometries}
\label{sec:efficientDecompose}

In order to be able to test the consistency of measurement dependencies of a \oneway\ procedure $\fP$ with those arising in the DKP or (simplified) RBB construction, we may attempt to determine whether the tuple $(G,I,O,M)$ associated to $\fP$ has a modified flow.
As a special case, we may efficiently determine when a geometry $(G,I,O)$ has a flow, and construct one if one exists.
The most efficient known algorithm to do this is that of Mhalla and Perdrix~\cite{MP08}.
We show how their algorithm may be adapted to produce modified flows when they are present, using an analysis in terms of maximal star geometries.

\paragraph{Maximally delayed flows.}

We may summarize the algorithm of~\cite{MP08} as follows.
To find a flow $(f,\preceq)$ or a geometry $(G,I,O)$, we consider the length of the longest chains of each qubit from the output set, and construct ``layer sets'' consisting of vertices of similar depth:
\begin{definition}[{\cite{MP08}, Definitions~4~\&~5}]
	\label{def:delayLayers}
	For a geometry $(G,I,O)$ and a flow $(f,\le)$, let $V^<_0(G)$ be the set of vertices of $V(G)$ which are maximal in $\le$, and let $V^<_{j+1}(G)$ be the maximal set of vertices of the set
	\begin{gather}
				W^<_j
			\;\;=\;\;
				V(G) \;\setminus\; \paren{\bigcup_{\ell = 0}^j \; V^<_j}
	\end{gather}
	in the order $\preceq$, for any $j > 0$.
	A flow $(f,\preceq)$ is \emph{more delayed} than another flow $(f',\le)$ if we have $\card{W^\prec_j} \le \card{W^<_j}$ for all $j \in \N$\,, where this inequality is strict for at least one such $j$\,; and $(f,\preceq)$ is \emph{maximally delayed} if no flow for $(G,I,O)$ is more delayed.
\end{definition}
A maximally delayed flow is also a flow of minimum depth, \ie\ that the order $\preceq$ of such a flow has the same depth as the natural pre-order for the given flow-function $f$.
For a maximally delayed flow $(f,\preceq)$, let $\ell: V(G) \to \N$ be the function mapping each vertex $v \in V(G)$ to the integer $j$ such that $v \in V^\prec_j$.
Then, the partial order $\le$ characterized by $v \le w \iff [v = w \text{\;or\;} \ell(v) > \ell(w)]$ also yields a minimum-depth flow by construction.

Mhalla and Perdrix show for a maximally delayed flow $(f,\preceq)$ for $(G,I,O)$ that $V^\prec_0 = O$, and $V^\prec_1$ is the set of elements of $W^\prec_0$ such that, for some $w \in O$, $v$ is the only neighbor of $w$ which is contained in $W^\prec_0$.
By induction, augmenting the set $O$ to include successive ``layers'' $V^\prec_j$, one may show that
\begin{gather}
		V^\prec_{j+1}
	\;\;=\;\;
		\ens{v \in W^\prec_j \;\Big|\; \exists w \in V(G) \setminus W^\prec_j\,:\, \ngbh[G]{w} \inter W^\prec_j = \ens{v}}	\;,
\end{gather}
where $\ngbh[G]{w}$ is the set of vertices adjacent to $w$ in $G$.
By a depth-first search from $O$, we may find in a finite number of iterations of $j$ those qubits $w \in V(G) \setminus W^\prec_j$ which have such a neighbor $v$, and accumulate them into new layers $V^\prec_{j+1}$.

\paragraph{Relation to maximal star geometries.}
\label{sec:modFlowAlg}

This decomposition of the geometry $(G,I,O)$ into layers $V^\prec_j$ corresponds strongly to the star decomposition of a geometry $(G,I,O)$.
The maximal elements of $V(G)$ in the natural pre-order are exactly the elements of $O$ (as $z \preceq f(z)$ for any $z \in O\comp$): thus, the root $v$ of a maximal star geometry in $(G,I,O)$ is one such that all neighbors $z \sim f(v)$ are elements of $O$ or equal to $v$ itself.
The set $V^\prec_1$ as described above is the set of precisely those vertices which are the root of a maximal geometry.
Thus, the algorithm of~\cite{MP08} may be described as finding flows by performing a star decomposition of $(G,I,O)$ when it has a flow.

We may easily extend the maximal delay principle of~\cite{MP08} for flows to one which identifies a maximal star geometry for $(G,I,O)$ in the more general case where this geometry (together with a set $M$) has a modified flow.
For such a triple $(G,I,O,M)$ with a modified flow $(f,\preceq)$, we may define the sets $V^\prec_j$ and $W^\prec_j$ as in \autoref{def:delayLayers}.
Then for each $j \in \N$, $V^\prec_{j+1}$ consists of those vertices $v \in W^\prec_j$ such that, for any $w \in V(G)$ such that $w = f(v)$ or $w \sim f(v)$, we have $w = v$ or $w \in V(G) \setminus W^\prec_j$.
For $v \notin M$\,, this implies that (as in the analysis of~\cite{MP08}) the only neighbor of $f(v)$ which lies in $W^\prec_j$ is $v$ itself; for $v \in M$, it may also be possible that $v$ has no neighbors in $W^\prec_j$.
The flow-finding algorithm of~\cite{MP08} can then be easily extended to test for the latter condition, yielding a procedure to produce a star decomposition of $(G,I,O)$ when one exists, as we describe below.

\paragraph{Bounding the search for roots of maximal star geometries.}

The flow-finding algorithm of Mhalla and Perdrix (\cite{MP08}, Algorithm~1) operates on the basis of maintaining a set of ``correcting vertices'' $w$ which, for each new layer $V^\prec_{j+1}$, may potentially be the successor of some vertex $v \in V^\prec_{j+1}$.
These correcting vertices are precisely the centers of maximal geometries, in which case the corresponding $v$ is a root.
We thus take a similar approach to obtain a star decomposition by ``removing'' layers of maximal star geometries, identified with their centers.

For modified flows, we must accommodate vertices $v \in M$ which may be simultaneously the root and the center of a geometry.
To ensure that we do not fruitlessly examine vertices of $M$ at each layer, we maintain a digraph whose arcs point in either direction along each edge of $G$, and from which we delete any arc ending at a vertex which has already been indicated as the root of a star geometry.
Then, if a potential center $w$ has a single neighbor $v$ which has not been matched to a star geometry, $v$ is the root and $w$ the center of a maximal star geometry; and if a vertex $w$ has not yet been identified as the center of a star geometry only neighbors vertices which have already been identified as centers, $w$ is both the center and the root of a maximal star geometry.
As a result, we may easily determine whether a potential center $w$ has a single neighbor $v$ which has not been matched to a star geometry (in which case $v$ is the root and $w$ the center of a maximal star geometry), or whether a vertex which has not yet been identified as the center of a star geometry only neighbors vertices which have already been identified as centers (in which case $w$ is both the center and the root of a maximal star geometry).

\Algorithm{alg:removeStarGeom} defines a subroutine $\RemoveStarGeom$ to perform the necessary arc deletions for the digraph described above whenever we define the layer of a vertex, as well as several other operations useful to an iterative procedure for discovering the maximal star geometries of a star decomposition.

\begin{algorithm}[t]
	\SUBROUTINE \RemoveStarGeom(v,w) :
	
	\Data{%
		$M$, a subset of specially designated vertices
	\\
		$W$, the set of vertices not yet marked as roots of star geometries
	\\
		$\emptyLayer$, a flag indicating if the present layer is empty
	\\
		$D$, a digraph in which vertices $v \in V(G) \setminus W$ have in-degree 0
	\\
		$f: V(G) \partialto V(G)$ a mapping from ``roots'' to ``centers''
	\\
		$S'$, a set of possible centers of star geometries for the next layer
	\\
		$v \in V(G)$, a vertex to mark as the root of a star geometry
	\\
		$w \in V(G)$, a vertex to mark as the center associated to $v$
	}

	\Effect {%
		Assigns $f(v) \gets w$, adds $v$ to the set of potential centers for the next layer and to the ordered list of qubits with identified star geometries, removes all of the inbound-arcs to $v$ in the digraph $D$, and checks the neighbors of $v$ in $D$ for vertices in $M$ with out-degree 0 to mark as a center for the next layer.
	}

	\BEGIN {
		Set $\emptyLayer \gets \false$\;
		Remove $v$ from the set $W$, and set $f(v) \gets w$\;
		Insert $v$ to the beginning of $L$\;
		\FOREVERY neighbor $z$ with an outbound arc to $v$ in $D$ \DO
			Remove the outbound arc $z \arc v$ from $D$\;
			\oneline
			\IF $z \in M \inter W$ \ANDALSO $z$ has no more outbound arcs \THEN add $z$ into $S'$\;	\nllabel{line:mediatorInsert}
		\ENDFOR
	}
	\END
 	\caption[A subroutine to attribute a vertex to the current layer, for use in an iterative procedure for finding the roots of maximal sub-geometries with a modified flow.]{\label{alg:removeStarGeom}%
 		A subroutine to attribute a vertex to the current layer, for use in an iterative procedure for finding the roots of maximal sub-geometries with a modified flow.
 		In particular, this subroutine maintains the arcs of a digraph $D$ (a ``sub-digraph'' of the graph $G$ of a geometry, interpreted as a symmetric digraph) from which we remove any arc inbound to a vertex with a well-defined layer.
 	 }
\end{algorithm}

\begin{algorithm}[p]

	\PROCEDURE \FindModStarDecomp(G,I,O,M):

	\Input{%
		$G$, a graph with subsets $I, O, M$, where $M \subset V(G) \setminus O$ consists entirely of vertices with default measurement angles $\pion 2$.
	}

	\Output{%
		Either $\nil$, if $(G,I,O,M)$ does not have a modified flow; or an ordered pair $(f,L)$ consisting of a partial function $f: V(G) \partialto* V(G)$ mapping roots of star geometries to their respective centers, and a sequence $L$ of the roots representing a composition order for star geometries.
	}

	\BEGIN {
		\LET $f: V(G) \partialto V(G)$ be a mapping from ``roots'' to ``centers''\;
		\LET $L$ be a sequence of vertices (initially empty)\;
		\LET $W \gets \vide$ be a set of potential roots of star geometries\;
		\LET $S \gets \vide$ be a set of potential centers of maximal star geometries\;
		\LET $D$ be the completely disconnected digraph on $V(G)$\;
		\algline
		\FORALL $v \in V(G)$ \DO																						\nllabel{line:eflowInitLoop}
		{
			\oneline
			\FOREVERY $w \in \ngbh[\scriptstyle G]{v} \setminus O$ \DO add the arc $v \arc w$ to $D$ \ENDFOR\;
			\IF $v \in O$ \THEN { 
				Set $f(v) \gets \nil$\;
				\oneline
				\IF $v \notin I$ \THEN add $v$ into $S$\; }
 			\ELSE {
				Add $v$ into $W$\;
			}
		}
		\ENDFOR																												\nllabel{line:eflowInitEnd}
		\medskip

		\LET $\emptyLayer$ be a flag indicating failure to find any new roots\;
		\REPEAT																											\nllabel{line:layerLoop}
			Set $\emptyLayer \gets \true$, initially\;
			\LET $S' \gets \vide$ be the set of potential centers for the next iteration\;

			\FOREACH $w \in S$ \DO {
				\IF $w \in W$ \THEN {																			\nllabel{line:checkYZmeas}
					Call $\RemoveStarGeom(w,w)$\;
				} \ELSEIF $D$ has only one outbound arc $w \arc v$ \THEN {							\nllabel{line:checkSuccessor}
					Call $\RemoveStarGeom(v,w)$\;
					\oneline
					\IF $v \notin I$ \THEN { add $v$ into $S'$\; }										\nllabel{line:newSuccessor}
				} \ELSE {
					Add $w$ back into $S'$\;																	\nllabel{line:propSuccessor}
				}
			} \ENDFOR
			\smallskip
			Set $S \gets S'$ for the next iteration\;															\nllabel{line:successorCopy}
		\UNTIL { $\emptyLayer$ }																					\nllabel{line:layerLoopEnd}

		\medskip
		
		\oneline
		\IF $W = \vide$
			\THEN {\RETURN $(f,L)$; }
			\ELSE { \RETURN \nil }\;
	} \END

	\caption[An iterative algorithm to discover a star decomposition.]{\label{alg:findModStarDecomp}%
		An iterative algorithm to discover a star decomposition inspired by the analysis of~\cite{MP08}, using the subroutine $\RemoveStarGeom$ of \Algorithm{alg:removeStarGeom}.}
\end{algorithm}

\vspace{1em}
\paragraph{Accumulating roots of star geometries.}
Using $\RemoveStarGeom$ as a subroutine, we may easily define a procedure similar to Algorithm~1 of \cite{MP08} which either
\begin{inlinum}
\item
	determines that a star decomposition of the geometry does not exist, implying the non-existence of a modified flow, or

\item
	constructs a modified flow, represented as a tuple $(f,L)$ consisting of a partial function $f: V(G) \partialto V(G)$ mapping the root of each star geometry to the corresponding center, and a list $L$ of the roots corresponding to a composition order for the star geometries.\footnote{%
		The list $L$ constructed by \Algorithm{alg:findModStarDecomp} is a linear order, \ie\ a partial order of \emph{maximum} depth.
		To find one of \emph{minimum} depth as in Algorithm~1 of \cite{MP08}, it suffices to maintain a layer counter and to record the layer for which each vertex $v \in O\comp$ is the root of a maximal star geometry.
		However, as we are ultimately interested in producing a (linearly ordered) circuit expression, it is not necessary to obtain a partial order of smaller depth.}
\end{inlinum}
\Algorithm{alg:findModStarDecomp} presents such a procedure, which we describe below.

We begin by creating the digraph $D$ which governs the search for roots: we copy all of the neighborhood relations of $G$, omitting those arcs ending in $O$, as these cannot be roots of star geometries.
Similarly, as no $w \in O$ is a root, it must lie outside the domain of $f$, so we set $f(w) \gets \nil$; but if $w \in O$ is not an element of $I$, it will be the center of some star geometry, so we insert it to a set $S$ of candidates for centers of \emph{maximal} star geometries.
Any vertex $w \notin O$ is the root of some star geometry in a star decomposition of $(G,I,O)$, provided such a decomposition exists; we then insert such $w$ into a set $W$ of candidates for roots of maximal star geometries.

We then search for layers of new roots, corresponding to the layer-sets in the analysis of~\cite{MP08}.
For each layer, we iterate through the set of potential centers for maximal star geometries, which initially includes no elements of $W$.
We maintain the invariant that $D$ contains all (directed) adjacency relations of $G$, except that it contains no arcs ending at vertices $w \notin W$.
We also maintain the invariant that $v \in W$ becomes an element of $S$ only if it is an element of $M$ and has out-degree zero in $D$, in which case it is both the root and the center of a maximal star geometry.
We indicate the removal of vertices (as a result of decomposing the geometry into star geometries and an as-yet decomposed ``residue'') by absence from both the sets $S$ and $W$; thus $S$ represents at each stage the output subsystem of the residual geometry.

In any iteration, for the center $w$ of any maximal star geometry, either $w$ has a unique neighbor $v \notin W$ in $G$ (in which case $w \arc v$ is the only outbound arc from $w$ in $D$), or $w \in M$ and has no neighbors in $W$ (in which case it has out-degree 0 in $D$).
\begin{itemize}
\item 
	In the latter case, unless $w$ is isolated in $G$, there was some final neighbor $u$ which was removed from $W$ in an earlier iteration; immediately after the removal of $u$ from $W$, we may then identify that $w$ is the root (and center) of a maximal star geometry.
	In the call to $\RemoveStarGeom(u,u')$, we thus examine each of the neighbors $w \in M \inter W$ (restricting our search to qubits in $M$ which have not yet been assigned as the root to a previously removed star geometry) which have an outbound arc $w \arc u$, to see if that arc is the last outbound arc from $w$.
	If this is the case for some $w$, we add $w$ into the set of potential centers $S'$ for the next iteration.

	In \Algorithm{alg:findModStarDecomp}, this is the only condition under which an element of $W$ may be inserted into $S'$; the only other condition under which $S'$ accumulates elements is for a vertex $v$ which has already been removed from $W$.
	Thus, in the subsequent iteration, any element of $S \inter W$ is necessarily the root (and center) of a maximal star geometry, which we may then decompose from the parent geometry.
	This decomposition is represented by the removal of $w$ from $W$, and the omission of $w$ from $S'$.

\item
	In the case where $w$ has a unique outbound arc $w \arc v$ in $D$, the vertex $v$ is necessarily the root of a maximal star geometry for which $w$ is the center.
	We may then decompose this star geometry from the parent geometry, which we represent by the removal of $v$ from $W$ (and the insertion of $v$ into $S'$), and the omission of $w$ from $S'$.
	In the next iteration, as an element of the output subsystem of the residual geometry, $v$ is then a potential center of a maximal star geometry.
\end{itemize}
For any $w \in S$ which is not the center of a maximal geometry as above, we carry $w$ forward for the next layer by including $w$ in $S'$.
We assign the list of new candidate successors to $S$ before the end of the iteration; and we repeat until we cannot find any more roots for some layer.

If we cannot assign any vertices to new layers, this is either because we have exhausted the possible roots (in which case $W = \vide$), or because there are no vertices which are reachable from the existing layers which have viable successors.
In the latter case, there is no modified flow; and in the former, we return the tuple $(f,L)$.

\paragraph{Run-time analysis of Algorithm~\ref*{alg:findModStarDecomp}.}
\label{sec:findModStarDecompAnalysis}

The run-time of \Algorithm{alg:findModStarDecomp} is similar to that of Algorithm~1 of \cite{MP08}, which we may show as follows.
In the following, we let $n = \card{V(G)}$, $k = \card{O}$, and $m = \card{E(G)}$.
We assume that the graph $G$ and digraph $D$ are represented with adjacency lists for each vertex, in the latter case using lists of neighbors from inbound and outbound arcs separately.
The partial function $f$ is represented with an array structure taking the value $\nil$ (indicating a vertex not in the domain) or vertices in $V(G)$. 
We may represent the sets $M$ and $W$ as arrays of pointers to nodes of a linked-list structure, enabling constant-time membership checking and element insertion, as well as linear-time traversal of the elements of the set.
The sets $S$ and $S'$ we represent simply as linked lists, which also provides us with constant-time initialization to the empty set.

In the subroutine $\RemoveStarGeom(v,w)$, the run-time is dominated by the loop over the neighbors $z$ of $v$\,, which requires time at most $O(\deg_G(v))$\,.
As this subroutine is called at most once for each vertex $v$, the total run-time of $\RemoveStarGeom(v,w)$ over the execution of $\FindModStarDecomp$ is $O\big(\sum\limits_v \deg_G(v) \big) = O(m)$.

Initially, $S$ has $k$ elements, none of which are elements of $W$\,; and in each iteration, as every element $w \in S \setminus W$ gives rise to the removal of some vertex $v$ from $W$ and the potential insertion of $v$ into $S'$, there always remain at most $k$ elements of $S \setminus W$ at each iteration of the loop on lines~\ref*{line:layerLoop}\thru\ref*{line:layerLoopEnd}.
The comparisons on lines~\ref*{line:checkYZmeas} and~\ref*{line:checkSuccessor}  can both be performed in constant time, as can the set-inclusions on lines~\ref*{line:newSuccessor} and~\ref*{line:propSuccessor}.
The assignment on line~\ref*{line:successorCopy} can be performed simply by re-attributing the element list of $S'$ to $S$.
There are at most $n$ iterations of the loop on lines~\ref*{line:layerLoop}\thru\ref*{line:layerLoopEnd}, as it only repeats as long as at least one vertex has been assigned to a new layer.
Then, setting aside the work performed for vertices $w \in S \inter W$ in the loop on lines~\ref*{line:layerLoop}\thru\ref*{line:layerLoopEnd} and the work performed by $\RemoveStarGeom$, the work performed by that loop is $O(kn)$.

As elements of $W$ are only included into $S'$ when they are guaranteed to be attributed to the next layer, the total work performed by that loop (again aside from work performed by $\RemoveStarGeom$) for elements of $M$ is then $O(n)$.
The run-time of $\FindFlow(G,I,O,M)$ is then $O(kn + m)$, dominated by the work performed by $\RemoveStarGeom(v,w)$ for each $w \in V(G)$, and by the iteration of the loop on lines~\ref*{line:layerLoop}\thru\ref*{line:layerLoopEnd} for at most $n$ layers.

\subsubsection{Extremal analysis of geometries with modified flows}
\label{sec:extremal}

The run-time of $O(kn + m)$ for \Algorithm{alg:findModStarDecomp} is precisely the complexity of Algorithm~1 of \cite{MP08}.
However, Mhalla and Perdrix also use an extremal result presented in~\cite{BP08} which bounds the number of edges $m$ in a geometry with a flow by $kn - \binom{k+1}{2}$.
This allows any geometry with edges in excess of this to be rejected at the outset, so that the run-time for geometries satisfying this bound may be simplified to $O(kn)$.
By observing the correspondence drawn in \autoref{sec:dependenciesIndexSuccFns} between successor functions and the ordering described in \eqref{eqn:flowPremonition} for indices of stable-index expressions, we may extend this extremal result as follows:
\begin{lemma}
	\label{lemma:extremal}
	Let $(G,I,O)$ be a geometry such that $(G,I,O,M)$ has a modified flow, for some vertex set $M \subset V(G)$.
	If $\card{V(G)} = n$, $\card{E(G)} = m$, and $\card{O} = k$, then $m \le nk - \binom{k+1}{2}$.
	Furthermore, there exists a geometry with a modified flow which saturates this bound.
\end{lemma}
\begin{proof}
	Consider a decomposition $(G,I,O) = (G_{v\star}, I_{v\star}, O_{v\star}) \,\circ\, (G', I, O')$ into a maximal geometry and a residual geometry.
	If $v$ is both the root and the center of the geometry, then it is adjacent to at most $k$ vertices, as it is adjacent only to elements of $O' = O$.
	Otherwise, let $w = f(v)$ be the center of $(G_{v\star}, I_{v\star}, O_{v\star})$; then $w$ is again adjacent only at most $k$ vertices, as it is adjacent to $v$, and otherwise only to elements of $O \setminus \ens{w}$.
	In either case, $(G_{v\star}, I_{v\star}, O_{v\star})$ contributes at most $k$ edges to $(G,I,O)$.
	In either case, we have $\card{O'} = \card{O} = k$.

	By induction, each star geometry in a decomposition of $(G,I,O)$ contributes at most $k$ edges, and has an output subsystem of size $k$.
	Then the output subsystem $\tilde I = V(G) \setminus \img(f)$ of the initial geometry $(G[\tilde I], I, \tilde I)$ in the star decomposition also contains $k$ vertices; this geometry then has at most $\binom{k}{2} = \frac{1}{2}k(k-1)$ edges.
	Furthermore, as $(G[\tilde I], I, \tilde I)$ does not contain the center of any star geometry, and as each star geometry has a center distinct from the others, there are at most $n - k$ star geometries in the decomposition, each with a distinct center.
	The number of edges in $(G,I,O)$ is then bounded by
	\begin{align}
		 	(n-k)k + \tfrac{1}{2}k(k-1)
		\;\;=\;\;
			nk - \tfrac{1}{2}k^2 - \tfrac{1}{2}k
		\;\;=\;\;
			nk - \tbinom{k+1}{2}	\;,
	\end{align}
	as required; this bound is saturated by the extremal construction of \cite{BP08}.
\end{proof}
As a result, we may simplify the run-time of \Algorithm{alg:findModStarDecomp} by $O(kn)$ in this case as well, again by a pre-processing stage which returns $\nil$ for geometries with more than $kn - \binom{k+1}{n}$ edges.

\subsection{Candidate circuit constructions}
\label{sec:candidateConstr}

Following the approach in the proof of \autoref{lemma:extend-DK06}, we may consider how to construct ``candidate'' circuits $C$ for \oneway\ procedures $\fP$, by decomposing of the underlying geometry into star geometries and mapping each star geometry to a corresponding unitary circuit.
This will give rise to a circuit $C$ over the gate set $\cB_\DKP$ or $\cB_\RBB$, such that $\Phi(C)$ is another \oneway\ procedure with the same underlying geometry.
If we can prove that the circuit $C$ performs the same operation as $\fP$, then $C$ describes the semantics of $\fP$ in the unitary circuit model.
We may use this as the final step towards defining a semantic map $\cS$ for the DKP and (simplified) RBB constructions.

\subsubsection{Star circuits corresponding to star geometries}

For a star geometry, we may consider a corresponding star circuit, possibly parameterized (in the case of a star geometry with root $v$ and center $w = f(v) \ne v$) by an angle related to the measurement angle $\theta_v$ of $v$:
\begin{gather}
	\label{eqn:starCircuitFormula-a}
	 	C_{v\star}
	\;=\;
		\paren{\prod_{\substack{u \sim w \\ u \ne v}} \cZ\pseu[u,w]} J(-\theta_v) \pseu[:w/v]	\;.
\end{gather}
For star geometries whose center and root are the same (in the case that $f(v) = v$), corresponding to the mediating qubits of a $\Zzz_d$ operation for some $d \ge 2$ as in \eqref{eqn:defZzz}, the corresponding star circuit is just the operator $\Zzz_d\pseu[u_1, \ldots, u_d]$ acting on the input/output subsystems: this corresponds to the decomposition of $\Zzz_d$ as
\begin{gather}
	\label{eqn:mediatorStarCircuit}
	 	\Zzz_d\pseu[w_1,\ldots,w_d]
	\;\;\propto\;\;
		P \pseu[:/a] \paren{\prod_{j = 1}^d \cZ\pseu[a,w_j] } \ket{+}\pseu[:a/]	,
\end{gather}
generalizing \eqref{eqn:mediatorStarCircuitPremonition}.
We may also correspond the geometry $(G[\tilde I], I, \tilde I)$ to a circuit consisting solely of $\cZ$ gates, as described with \eqref{eqn:inputStarPremonition}:
\begin{gather}
		C_{\tilde I}
	\;\;=\;\;
		\paren{\prod_{\substack{vw \in E(G[\,\tilde I\,])}} \cZ \pseu[v,w] } \paren{\prod_{v \in \tilde I \setminus I} \ket{+} \pseu[:v/] }	\,;
\end{gather}
as described in \autoref{sec:dependenciesIndexSuccFns}, such a circuit arises in star decompositions of circuits as well.
We may then construct a candidate circuit $C$ as in the preceding section by re-composing the circuits corresponding to the geometries in the star decomposition of $(G,I,O)$.
Such re-composed circuits are illustrated in the bottom panels of \autoref{fig:starDecomp} and \autoref{fig:starModDecomp}.
For instance, in \autoref{fig:starDecomp}, we may recover the geometry in the upper left-hand panel by composing
\begin{align}
	\begin{split}
		(G,I,O)
	\;=\;
		(G_{c\star}\,, I_{c\star}\,, O_{c\star})
		\circ
		(G_{d\star}\,, I_{d\star}\,,& O_{d\star})
		\circ
		(G_{a\star}\,, I_{a\star}\,, O_{a\star}) \circ
	\\&
		(G_{b\star}\,, I_{b\star}\,, O_{b\star})
		\circ
		\big(G\big[\:\!\tilde{I}\:\!\big], I, \tilde{I}\big)	\,;
	\end{split}
\end{align}
and similarly, the larger circuit $C$ in \autoref{fig:starDecomp} is obtained by composing
\begin{align}
		C
	\;=\;
		C_{c\star} \circ C_{d\star} \circ C_{a\star} \circ C_{b\star} \circ C_{\tilde I}	\;,
\end{align}
where the circuits $C_{u\star}$ are those corresponding to each star geometry, and $C_{\tilde I}$ corresponds to the geometry  $\big(G\big[\,\tilde{I}\,\big], I, \tilde{I}\big)$.

\subsubsection{The gate model for candidate circuit constructions}

Neither of the circuits described in \autoref{fig:starDecomp} or \autoref{fig:starModDecomp} are generated over the gate set $\ens{H,T,\cZ}$; and the latter is not defined over either $\cB_\DKP$ or $\cB_\RBB$, but instead over the more general gate-set $\bar\cB$ described in the proof of \autoref{lemma:extend-DK06}.
(As we suggested just after the proof of that \autoref*{lemma:extend-DK06}, this gate-set corresponds naturally to the structure of tuples $(G,I,O,M)$ with modified flows; we may refine this statement and describe this set of gates as that which naturally corresponds to star geometries.
In order to use star decompositions to describe a semantic map corresponding to representations of circuits over $\ens{H,T,\cZ}$, we must first translate the circuits from the gate-set $\bar\cB$.

There is one apparent problem: the circuits arising from such a decomposition are more general than those in the domain of $\Phi$.
This may conceivably complicate the task of comparing the circuits arising from a semantic map $\cS$ based on star decompositions, and the circuits which lie in the domain of either $\cR_\DKP$ or $\cR_\RBB$.
We may show that this more general sort of construction will not arise in practise for \oneway\ procedures resulting from either the DKP construction or the RBB construction when $\card I = \card O$.
In this case, a geometry has at most one star decomposition up to commuting terms, by the uniqueness of the modified flow $(f,\preceq)$ shown in \autoref{thm:modFlowUniqueness} for this case.
In particular:
\begin{itemize}
\item 
	for $(G,I,O)$ which underly procedures $\fP$ produced by the DKP construction, there will be no qubits $v \in O\comp$ such that $f(v) = v$, because $(f,\preceq)$ will be a flow;

\item
	for $(G,I,O)$ which underly procedures $\fP$ produced by the RBB construction, the only edges $v w$ in $(G,I,O)$ for which $v \ne f(w)$ and $w \ne f(v)$ are those such that either $v = f(v)$ or $w = f(w)$: \ie\ where one of the qubits is a mediator qubit for a $\Zz$ operation between two logical qubits.
\end{itemize}
Thus, while the candidate circuits for arbitrary geometries with modified flows may be more general, the candidate circuits for \oneway\ procedures $\fP$ arising from a circuit $C$ over $\cB_\DKP$ or $\cB_\RBB$ will be restricted to the same gate-set as $C$ itself when $\card I = \card O$.

For a semantic map defined relative to a particular construction, we may also ensure that we obtain candidate circuits over the appropriate gate-set by imposing restrictions on the type of star geometries which we permit in the decomposition of a particular geometry.
However, as we considering only the special case $\card I = \card O$, this will not be necessary for our analysis.

\subsubsection{Relationships between \oneway\ procedures and candidate circuits}
\label{sec:candidateCircuitCorresp}

Properly speaking, a candidate circuit is a construction from a \oneway\ procedure $\fP$ rather than from a star decomposition of a geometry $(G,I,O)$, as the angular parameters of the $J(\theta)$ gates in a candidate circuit are determined by the measurement angles in $\fP$.
In order to use candidate circuits as a means of defining a semantic map $\cS$, we must consider the relationship between \oneway\ procedures and their candidate circuits.

By \autoref{thm:verifyDependencies}, we may characterize the dependencies of operations in procedures $\fP$ obtained from either the DKP or RBB constructions, relative to a candidate for the extended successor function for a corresponding circuit.
If the geometry $(G,I,O)$ underlying $\fP$ (together with the set $M$ of mediating qubits) has a modified flow $(f,\preceq)$, this defines a candidate circuit $C$, in which the gates $J(\alpha_u)$ arising from star-circuits $C_{u\star}$ (for a qubit $u \ne f(u)$ at the root of a star geometry) have their angles fixed by $\alpha_u = -\theta_u$, where $\theta_u$ is the default measurement angle of $u$ in $\fP$.

By the observations made above, the circuit $C$ will be defined over the gate set $\cB_\DKP$ or $\cB_\RBB$, depending on whether $\fP$ arises from the DKP or RBB construction: we may then consider the \oneway\ procedure $\tilde \fP = \bar\Phi(C)$, which will coincide with $\Phi(C)$.
The procedure $\tilde\fP$ will have the same geometry as $\fP$, and the same default measurement angles: if it also has the same correction and measurement dependencies, this implies that $\fP$ performs the same transformation as the circuit $C$.

For procedures $\fP$ arising from the DKP or RBB constructions in the case that $\card I = \card O$, there is only one modified flow for $(G,I,O)$, which must then be the extended successor function for the circuit over the $\cB_\DKP$ or $\cB_\RBB$ gate sets produced in the course of constructing $\fP$.
It follows then that $C$ is also precisely the circuit arising in the construction of $\fP$, in which case the normal form of $\Phi(C)$ will be congruent to $\fP$.
Thus, to obtain a semantic map for the DKP construction when $\card I = \card O$, it suffices to translate $C$ into the gate-set $\ens{H,T,\cZ}$\,.
For the RBB construction, to obtain a circuit with complexity no greater than the original circuit from which $\fP$ was produced, we must also simplify $C$ by removing products of the form $J(0) J(0)$ and $J(\pion 2) J(\pion 2) J(\pion 2)$ introduced by the construction of \autoref{apx:stableIndexConstructionConstrained}.

In the case that $\card I \ne \card O$, even conditioned on $\fP$ arising from the DKP or RBB constructions, it may be difficult to obtain the modified flow function $f$ corresponding to the circuit arising in the construction of $\fP$, and therefore to produce a suitable candidate circuit.
Without an algorithm that can construct such candidate circuits in the case that $\card I = \card O$, for instance by finding the modified flow $(f,\preceq)$ for which $\fP$ satisfies \autoref{thm:verifyDependencies}, we must therefore restrict ourselves to the DKP and RBB constructions in the case where no fresh qubits are introduced, \ie\ which perform unitary bijections.

\subsubsection{Efficient construction of candidate circuits over $\protect\ens{H,T,\protect\cZ}$}
\label{sec:buildCircuit}

For the sake of completeness, we present a procedure \BuildCircuit\ (extending the remarks made in Section III B of\cite{DK06}) for constructing a stable-index representation for a candidate circuit for $\fP$, represented in the gate set $\ens{H,T,\cZ}$ (rather than the sets $\cB_\DKP$ or $\cB_\RBB$).
The procedure \BuildCircuit\ takes as input a representation of a star decomposition of a geometry $(G,I,O)$: we may represent this in terms of the graph $G$ of the geometry, the successor function $f$ from the modified flow which defines the star decomposition, and a sequence $L$ of the roots of the component star geometries.
This procedure is presented in \Algorithm{alg:buildCircuit}.

\begin{algorithm}[p]

	\PROCEDURE \BuildCircuit(G,f,L,t):

	\Input{%
		$G$, the a graph for the geometry $\cG$ of a \oneway\ procedure
	\\
		$f: V(G) \partialto V(G)$, a modified flow function for $\cG$
	\\
		$L$, a list of roots of star geometries in a decomposition of $\cG$
	\\
		$t: V(G) \partialto \Z$ an array denoting multiples of $\pion 4$
	}

	\Output{%
		A candidate circuit generated over $\ens{H,T,\cZ}$, corresponding to the star decomposition described by $G$, $f$, and $L$, with powers of $T$ gates given by the parameters stored in $t(u)$.}

	\BEGIN {

		\LET $C$ a stable-index representation of a circuit, initially empty\;
		\algline
		\FOREACH $u \in L$ \DO {
			\algline
			\IF $u = f(u)$ \THEN {
				\LET $N \gets \ngbh[\scriptstyle G]{u}$\;
				\FOREACH $w \in N$ \DO {
					\oneline
					\FOREACH $w' \in N$ \DO { add $\cZ\pseu[w,w']$ to $C$\; } \ENDFOR
					Add $T\pseu[w]T\pseu[w]$ to $C$\;
					Remove $w$ from $N$\;
				}
				\ENDFOR
			}
			\ELSE {
				\LET $v \gets f(u)$\;
				\LET $N \gets \ngbh[\scriptstyle G]{u} \setminus \ens{u, f(v)}$\;
				\oneline
				\FOREACH $w \in N$ \DO add $\cZ\pseu[v,w]$ to $C$\; \ENDFOR
				Add $H\pseu[:v/u]$ to $C$\;
				\short
				\IF $t(u) < 0$ \THEN {
					\oneline
					\FOR $j = 1$ \TO $-t(u)$ \DO add $T\herm\pseu[u]$ to $C$ \ENDFOR
				} \ELSE {
					\oneline
					\FOR $j = 1$ \TO $t(u)$ \DO add $T\pseu[u]$ to $C$ \ENDFOR
				}				
			}
		}
		\ENDFOR

		\medskip

		\LET $\tilde I \gets V(G) \setminus \img(f)$\;
		\FOREACH $u \in \tilde I$ \DO {
			\smallskip
			\FOREACH $w \in \ngbh[\scriptstyle G]{u}$ \DO {
				\oneline \IF $w \in \tilde I$ \THEN { add $\cZ\pseu[u,w]$ to $C$\; }\relax
			}
			\ENDFOR
			\oneline
			\IF $u \notin I$ \THEN { add $\ket{+}\pseu[:u/]$ to $C$\; }\relax
			Remove $u$ from $\tilde I$\;
			\smallskip
		}
		\ENDFOR

		\medskip
		
		\RETURN $C$\;
	} \END

	\caption[A procedure to construct a candidate circuit corresponding to a star decomposition of a geometry.]{\label{alg:buildCircuit}%
		A procedure to construct a candidate circuit corresponding to a star decomposition of a geometry, realized over the gate-set $\ens{H,T,\cZ}$.}

\end{algorithm}

For each root $u$ in the list $L$, \BuildCircuit\ determines whether it is the root of a geometry whose root and center are both $u$ (corresponding to a $\Zzz_d$ gate acting on the neighbors of $u$ in $G$), or where the center $v = f(u)$ differs from $u$ (corresponding to a star circuit as in \eqref{eqn:starCircuitFormula-a}.
It then appends to a circuit expression $C$ the appropriate sequence of operations:
\begin{itemize}
\item
	By expanding $Z = \idop - 2\ket{1}\bra{1}$ in the tensor product $Z \ox \cdots \ox Z$, and considering the expansion of $\Zzz_d = \e^{-i\pi Z\sox{d}\!/4}$ in terms of multi-qubit operations $\e^{i ( \ket{1}\bra{1}\sox{j} ) \theta}$ for $0 \le j \le d$, we may show that
	\begin{gather}
			\Zzz_d\pseu[w_1, \ldots, w_d]
		\;\propto\;
			\left[\; \prod_{j = 1}^n	T^2\pseu[w_j] \left]\left[ \prod_{\substack{j,k = 1\\j < k}}^n \cZ\pseu[w_j, w_k] \right]\right.\right.	\;.
	\end{gather}
	We may construct such a gate for $w_j$ ranging over all neighbors of a qubit $u = f(u)$ by obtaining the set $N$ of these neighbors,
	and for each $w \in N$, adding $T$ operations on $w$ to $C$, as well as controlled-$Z$ operations between $w$ and the other elements of $N$.
	Removing $w$ from $N$ after doing so will then prevent controlled-$Z$ operations from being added twice for each pair of neighbors.

\item
	For each qubit $u$ such that $u \ne f(u)$, we may obtain the set of $N$ neighbors of $v = f(u)$ belonging to the output subsystem $O_{u\star}$ of the star geometry $(G_{u\star}, I_{u\star}, O_{u\star})$.
	$N$ will then consist of the neighbors of $v$ excluding $u$ and (if it is defined) $f(v) = f(f(u))$.
	We may then add the star circuit $C_{u\star}$ to $C$ by adding $\cZ$ operations between each $w \in N$ and $v$, followed by the components of the gate $J(\theta)$ for the appropriate angle $\theta$.
	For angles $\theta$ constrained to integer multiples of $\pion 4$, this will consist of $H\pseu[:v/u] T^{t(u)}\pseu[u]$ for some integer array $t$ given as input: if $t(u)$ is negative, we may add products of $T\herm\pseu[u]$ to $C$.
\end{itemize}
After having added the star circuits for each star geometry, we then add the circuit corresponding to the geometry $(G[\tilde I], I, \tilde I)$ for $\tilde I = V(G) \setminus \img(f)$.
After constructing the set $\tilde I$, we iterate through the qubits $u \in \tilde I$, adding to $C$ the $\cZ$ operations acting on $u$ and (if necessary) representing the fact that $u \notin I$ by adding a preparation of a fresh qubit $\ket{+}$ in order to advance $u$ as a fresh index in the circuit $C$.\footnote{%
	While the semantic map we will consider does not involve the introduction of fresh qubits, such fresh qubits do not complicate the process of constructing candidate circuits, and so we include such preparations in \Algorithm{alg:buildCircuit}.}
Subsequently removing $u$ from $\tilde I$ prevents duplicate $\cZ$ operations from being added to $C$.

\paragraph{Run-time analysis of Algorithm~\ref*{alg:buildCircuit}.}
We may bound the run-time of $\BuildCircuit$ as follows.
Let $n = \card{V(G)}$, $m = \card{E(G)}$, and $k = \card{O}$.
We assume a list-like structure for the stable-index expression $C$ and for the sets $\ngbh[\scriptstyle G]{u}$, and an array-structure for $f$, where $f(v) = \nil$ denotes that $v \notin \dom(f)$.
\begin{romanum}
\item
	For each qubit $u = f(u)$, the number of gates to add corresponding to the expression for $\Zzz_{\deg_G(u)}$ above is $O(\deg_G(u)^2)$.
	As the output subsystem of each star geometry in a star decomposition has the same size input and output subsystem, we have $\deg_G(u) \le \card{O} = k$ for each such $u$; then the total work number of gates arising from these qubits is $O(k^2 n)$.

\item
	For each qubit $u \ne f(u)$, the number of gates to add corresponding to the circuit $C_{u\star}$ is $O(\deg_G(f(u))$; summing over all such $u$, we have a total of $O(m)$ gates arising from these qubits.

\item
	For each qubit $u \in \tilde I = V(G) \setminus \img(f)$, we have $\deg_G(u)$ gates to add, plus possibly an initial state $\ket{+}\pseu[:u/]$. Summed over all $u \in \tilde I$, we again have at most $O(m)$ operations arising in this case.
\end{romanum}
As each operation requires time $O(1)$ to add, the run-time of \Algorithm{alg:buildCircuit} is then $O(k^2 n + m)$.

\subsection{Complete algorithm for the semantic map when $\card I = \card O$}
\label{sec:semanticAlg}

We may now describe a single, complete procedure \Semantic\ (presented in \Algorithm{alg:semantic}) to perform a semantic map $\cS$ with respect to the representations $\cR_\DKP$ and $\cR_\RBB$ represented by the constructions of \autoref{sec:constructions}.
\begin{algorithm}[h]

	\PROCEDURE \Semantic(\fP):

	\Input{%
		$\fP$, a \oneway\ procedure in normal form.
	}

	\Output{%
		Either $\nil$, if $\fP$ is not (congruent to) the normal form of a procedure as in \autoref{lemma:extend-DK06}, or a stable-index representation of a circuit $C$ generated over the gate-set $\ens{H,T,\cZ,\Zzz_d}_{d \ge 2}$ which performs the same transformation as $\fP$.}

	\BEGIN {
		\LET $(G,I,O)$ be the geometry underlying $\fP$\;
		\LET $n \gets \card{V(G)}$, $m \gets \card{E(G)}$, $k \gets \card{O}$\;
		\LET $M$ be a set of possible mediator qubits, initially empty\;
		\LET $t: O\comp \to \R$ be an array of multiples of $\pion 4$\;
		\FOREACH operation $\Meas{u}{\theta;\beta}$ in $\fP$ \DO
			Set $t(u) \gets\, -4\theta/\pi$\;
			\oneline
			\IF $\theta = \pion 2$ \THEN add $u$ into $M$\;
		\ENDFOR

		\smallskip
		\oneline
		\IF $m > nk - \binom{k+1}{2}$ \THEN \RETURN \nil\;
		\LET $\starDecomp \gets \FindModStarDecomp(G,I,O,M)$\;
		\short
		\IF $\starDecomp = \nil$ \THEN { \RETURN \nil } \ELSE { \textbf{let} $(f, L) \gets \starDecomp$ }
		\oneline
		\IF $\TestDependencies(\fP,f) = \false$ \THEN { \RETURN \nil }\; 
		\smallskip

		\LET $C \gets \BuildCircuit(G,f,L,t)$\;
		
		\RETURN $\RemoveIdops(C)$\;
	} \END

	\caption[A complete algorithm for a semantic map with respect to either the DKP or simplified RBB constructions.]{\label{alg:semantic}%
		A complete algorithm for a semantic map with respect to either the DKP or simplified RBB constructions, using the procedures defined in \Algorithm{alg:testDependencies}, \Algorithm{alg:buildCircuit}, and \Algorithm{alg:findModStarDecomp}, as well as a procedure \RemoveIdops\ which performs trivial simplifications of products of single-qubit gates.}

\end{algorithm}%
We may describe the operations performed by \Semantic\ in three different phases, which we may describe as follows.

\paragraph{Obtain the underlying structure of $\fP$.}
We first obtain the geometry $(G,I,O)$ underlying $\fP$, as well as the set $M$ of qubits measured with default angles $\pion 2$, in order to obtain a star decomposition.
We also obtain the cardinalities $n = \card{V(G)}$, $m = \card{E(G)}$, and $k = \card{O}$.
For each qubit $u$ measured, we record the integers $t(u) \gets -4\theta/\pi$ which corresponds to the power of $T$ for which $HT^{t(u)} = J(-\theta_u)$; if the measurement at $u$ corresponds to the execution of a $J(\theta)$ gate in the candidate circuit, we may use $t(u)$ to determine the expansion of this gate in the gate-set $\ens{H,T,\cZ}$.

Rather than perform these operations literally as presented in \Algorithm{alg:semantic}, we may obtain this data with a single pass of the operations of $\fP$, constructing $G$, $I\comp$, and $O\comp$ by inserting qubits and edges as each qubit is acted on by preparation maps, entangling operations, and measurements; we may obtain $n$, $m$, and $k$ similarly.
Using arrays of pointers to linked list nodes to represent $I\comp$ and $O\comp$, we may then construct a similar representation of both $I$ and $O$ with a single traversal of $V(G)$.
As we do so, we also assign values to the array $t$ and insert qubits into the set $M$.
As $\fP$ has at most $O(n)$ preparation, measurement, and correction operations, and $O(m)$ entangling operations, it has size $O(m + n)$: this is then the time required for this phase of the procedure.

\paragraph{Examine the geometry and the measurement dependencies of $\fP$.}
As we noted in \autoref{sec:extremal}, a \oneway\ procedure cannot have a modified flow if it has more than $nk - \binom{k+1}{2}$ edges, in which case it lies outside the range of both the DKP and RBB constructions by \autoref{lemma:extend-DK06}; we then return $\nil$.
We then attempt to obtain a star decomposition for $(G,I,O)$ with the set $M$ of potential mediating qubits: if this fails, we also return $\nil$.

By \autoref{thm:verifyDependencies}, we may characterize the dependencies of operations in procedures obtained from either the DKP or RBB constructions, relative to a candidate for the successor function for a corresponding circuit.
For the star decomposition described by the ordered pair $(f,L)$ returned in the previous step, there is a corresponding candidate circuit $C$ for which $f$ is an extended index successor function, as described in \autoref{sec:candidateConstr}.
In the case that $\card I = \card O$, there is no other candidate successor function for $(G,I,O,M)$, by \autoref{thm:modFlowUniqueness}.
Then the candidate circuit $C$ described by $(f,L)$ is the only one which describes a unitary transformation which may be performed by $\fP$.

The \oneway\ procedure $\Phi(C)$ has underlying geometry $(G,I,O)$ by construction; if $\fP$ has the same operations (with the same dependencies) as the normalization of $\Phi(C)$, they then perform the same transformation, and thus $\fP$ is a representation of $C$.
Again by construction, the normalization of $\Phi(C)$ satisfies the conditions of \autoref{thm:verifyDependencies} with respect to $f$.
Then if $\TestDependencies(\fP,f)$ returns $\true$, we may provide the semantics for $\fP$ in terms of unitary circuits by constructing $C$; otherwise, we have no such guarantee, and so we return $\nil$.

We may compare $m$ to $nk - \binom{k+1}{2}$ in time $\polylog(n)$; and the respective run-times of \FindModStarDecomp\ and \TestDependencies\ are $O(kn) \subset O(n^2)$ and $O(kn^2) \subset O(n^3)$, by the analyses of \autoref{sec:efficientDecompose} and \autoref{sec:verifyDependencies} respectively.

\paragraph{Constructing a circuit for $\fP$.}

Provided that the procedure does not return $\nil$ in the previous phase, the procedure $\fP$ performs a unitary transformation, and in particular is a representation of the ``candidate'' circuit $C$ obtained from the star decomposition of $(G,I,O)$.
We therefore construct the circuit $C$, using the star decomposition described by $(f,L)$.
We produce a stable-index representation for $C$, representing it in the gate-set $\ens{H,T,\cZ}$ rather than the set $\bar\cB$ used in \autoref{lemma:extend-DK06}: this takes time $O(k^2 n + m) = O(k^2 n)$, as $m \le kn - \binom{k+1}{2}$ by hypothesis.

In the case of procedures $\fP$ obtained by applying the RBB construction to some original circuit $\tilde C$, we must perform some final simplifications.
We call a final subroutine $\RemoveIdops(C)$ in order to remove superfluous powers of $T$ which may arise from the conversion of operations $\Zz = (T^2 \ox T^2) \cZ$, and which also removes products of the form $H^2 = J(0) J(0)$ and $HT^2 HT^2 HT^2 = J(\pion 2) J(\pion 2) J(\pion 2)$ on individual qubits such as those introduced in the stable-index construction presented in \autoref{apx:stableIndexConstructionConstrained}.
As the dependencies of $\fP$ satisfies the constaints of \autoref{thm:verifyDependencies} relative to the successor function $f$, this circuit provides the semantics for $\fP$ as a unitary circuit, by the analysis of \autopageref{sec:candidateCircuitCorresp}.

The list $L$ produced by $\FindModStarDecomp$ contains roots of star-geometries in a non-increasing order (in terms of the natural pre-order $\preceq$ for the successor function $f$ it produces), and the operations of $C$ are in a similar order from the construction of $\BuildCircuit(G,f,L,t)$.
We may therefore perform this simplification of $C$ in a single pass from right to left, using buffers for each logical qubit of the circuit to maintain accumulated sequences of single-qubit operations performed between pairs of two-qubit operations.
We may then cancel products of single-qubit unitaries that occur at the end of the buffer, as appropriate.
The amount of time required to do so is then a constant factor time the size of the circuit, which by the analysis of \autopageref{sec:buildCircuit} is also $O(k^2 n + m) = O(k^2 n)$.

\paragraph{Remarks on circuit congruence and total run-time.}

In order to consider the correctness of \Algorithm{alg:semantic}, we may consider without loss of generality circuits $\tilde C$ represented in the gate-sets $\cB_\DKP$ or $\cB_\RBB$ which do not contain redundant operations $J(0)J(0)$ or $J(\pion 2)J(\pion 2)J(\pion 2)$, and which do not introduce fresh qubits.
We may proceed as follows.

If we apply the DKP construction to such a circuit $\tilde C$ over the gate-set $\cB_\DKP$, we obtain a procedure $\fP$ for which the qubits of $O\comp$ are labelled by deprecated indices in $\tilde C$, and $I\comp$ by advanced indices; where the entanglement graph is the same as the interaction graph of $\tilde C$; and where the default measurement angles are the additive inverses of the angles of the $J(\theta)$ gates.
Furthermore, the geometry $(G,I,O)$ has a modified flow $(f,\preceq)$, in which $f$ is the index successor function of $\tilde C$.
Correspondingly, the candidate circuit $C$ for $\fP$ has an interaction graph given by the entanglement graph of $\fP$, and a successor function which agrees with $f$.
The single-qubit gates corresponding to the edges $u f(u)$ in the interaction graph correspond to operations $J(-\theta_u)$, where $\theta_u$ is the default measurement angle of $u$; and all other edges correspond to $\cZ$ operations.
We may then define a bijective mapping between the terms of $\tilde C$ and the terms of $C$, via the structures induced in and recreated from $\fP$.

For the RBB construction, the analysis is similar: the only changes are the fact that the edges of the interaction graph of $\tilde C$ which correspond to operations $\Zz$ must be subdivided by a single vertex to yield the entanglement graph of $\fP$, and paths of the form $J(0)J(0)$ or $J(\pion 2)J(\pion 2)J(\pion 2)$ may be inserted to achieve the topological constraint of being an induced subgraph of a two-dimensional grid.
In constructing the candidate circuit $C$, the subdivided edges corresponding to $\Zz$ operations are re-merged, and the superfluous products of single-qubit unitaries are once more removed.
In this case, we may then also define a bijective mapping between the terms of $\tilde C$ and the terms of $C$, via the structures induced in and recreated from $\fP$.

\vspace{\baselineskip}
Thus, $\Semantic$ defines a semantic map for both the DKP and (simplified) RBB constructions, restricted to unitary bijections.
The run-time of $\Semantic$ is dominated by the time required to compare the dependencies of $\fP$ against the criteria of \autoref{thm:verifyDependencies}, which is $O(nm) \subset O(kn^2)$.

\section{Extensions and limitations}

The analysis of \oneway\ procedures presented in \autoref{sec:semanticMaps}, in terms of modified flows and star decompositions, may be extended in a natural way to more general formulations of the \oneway\ model.
At the same time, this approach has limitations in its ability to describe even some of the earliest known constructions for \oneway\ computation.
In this section, we will describe in brief some of these extensions and limitations.

\subsection{Extending beyond measurement in the \XY\ plane}

The model of measurement-based computation which we have focused on in this article (following the treatments in \cite{RB01} and \cite{RBB03}) is one in which all single-qubit measurements are performed in the \XY\ plane, \ie\ with respect to uniform superpositions of the eigenstates of the Pauli $Z$ operator.
A natural extension of this is to allow measurements in any basis in a ``Pauli plane'', \ie\ for any uniform superposition of the eigenstates of $X$ and $Y$ as well.
We may define the notation $\ket{\pm\sfx\sfy_\theta} = \ket{\pm_\theta}$ (denoting the fact that these states lie in the \XY\ plane), and then define
\begin{subequations}
\begin{align}
		\label{eqn:defYZtheta}
			\ket{\pm \sfy \sfz_\theta}
		\;=&\;\,
			\sfrac{1}{2} \Big( \ket{0} + \ket{1} \Big) \;\pm\; \sfrac{\e^{-i\theta}}{2} \Big( \ket{0} - \ket{1} \Big)	\;,
	\\	
		\label{eqn:defXZtheta}
		\ket{\pm\sfx\sfz_\theta}
	\;=&\;\,
			\sfrac{1}{2} \Big( \ket{0} + i \ket{1} \Big) \;\pm\; \sfrac{\e^{i\theta}}{2} \Big( \ket{0} - i \ket{1} \Big)	\;,
\end{align}
\end{subequations}
which denote bases lying in the $\YZ$ and $\XZ$ planes respectively.
Then, we define the measurement operators
\begin{subequations}
\begin{align}	
		\Meas[\XY]{a}{\theta}(\rho)
	\;=&\;\,
		\bra{\splus\sfx\sfy_\theta} \rho\, \ket{\splus\sfx\sfy_\theta}_{\!a} \ox \ket{0}\bra{0}_{\s[a]}
	\notag\\&\qquad\qquad\qquad+\;\;
		\bra{\sminus\sfx\sfy_\theta} \rho\, \ket{\sminus\sfx\sfy_\theta}_{\!a} \ox \ket{1}\bra{1}_{\s[a]}	\,,
\\[1ex]
		\Meas[\YZ]{a}{\theta}(\rho)
	\;=&\;\,
		\bra{\splus\sfy\sfz_\theta} \rho\, \ket{\splus\sfy\sfz_\theta}_{\!a} \ox \ket{0}\bra{0}_{\s[a]}
	\notag\\&\qquad\qquad\qquad+\;\;
		\bra{\sminus\sfy\sfz_\theta} \rho\, \ket{\sminus\sfy\sfz_\theta}_{\!a} \ox \ket{1}\bra{1}_{\s[a]}	\,,
\\[1ex]
		\Meas[\XZ]{a}{\theta}(\rho)
	\;=&\;\,
		\bra{\splus\sfx\sfz_\theta} \rho\, \ket{\splus\sfx\sfz_\theta}_{\!a} \ox \ket{0}\bra{0}_{\s[a]}
	\notag\\&\qquad\qquad\qquad+\;\;
		\bra{\sminus\sfx\sfz_\theta} \rho\, \ket{\sminus\sfx\sfz_\theta}_{\!a} \ox \ket{1}\bra{1}_{\s[a]}	\,.
\end{align}
\end{subequations}
We may then consider extensions of the analysis of \autoref{sec:semanticMaps}, in measurement-based models which include operations of the form above.

\subsubsection{Stabilizers, certificates of unitarity, and local graph structure}
\label{sec:eflow}

Recall from \autoref{lemma:extend-DK06} that we may construct procedures performing unitary transformations in the \oneway\ model by performing the following operation after each measurement $\Meas[\XY]{u}{\theta}$:
\begin{itemize}
\item
	If $\theta \ne \pion 2$, we perform an $X$ operation on some specially designated qubit $v \ne u$ (which we denote $f(u)$ for each $u$), and a $Z$ operation on each of the neighbors of $v$ except for $u$ itself, if the state of $u$ is projected onto $\ket{\sminus \sfx\sfy_\theta}$.

\item
	If $\theta = \pion 2$, either we perform the operation described above, or we perform a $Z$ operation on each neighbor of $u$, if the state of $u$ is projected onto  $\ket{\sminus \sfx\sfy_\theta}$.
\end{itemize}
In either case, the operation which is performed after each measurement may be described by an operator
\begin{gather}
		B_u
	\;\;=\;\;
		X_{f(u)} \paren{ \prod_{w \sim f(u)} Z_w }	\;,
\end{gather}
up to an operation acting on $u$ itself (which we describe as being destroyed by the measurement and which thus cannot be carried out).
This unitary operation is the Pauli group on $V(G)$, and in particular stabilizes the state $\ket{\Psi_{G,I}}$ defined by
\begin{gather}
 		\ket{\Psi_{G,I}}\bra{\Psi_{G,I}}
	\;=\;
		\Ent{G} \New{I\comp} \big( \ket{\psi}\bra{\psi} \big)
\end{gather}
for any $\ket{\psi} \in \cH\sox{I}$.
In particular, as the preparation maps $\New{v}$ adjoin qubits prepared in the $+1$ eigenstate of the $X$ operator, and the $\Ent{vw}$ maps perform two-qubit $\cZ$ operations (which lie in the Clifford group) on pairs of qubits, the state of the system just prior to any measurements is a stabilizer code~\cite{GotPhD} $\sS_{G,I}$ generated by the operators
\begin{gather}
		K_v
	\;=\;
		X_v \prod_{w \sim v} Z_w
	\;=\;
		\Ent{G}(X_v)	\;.
\end{gather}

The modified flow conditions may then be considered as a certificate that there exists an order in which the qubits of $O\comp$ may be measured so that one may use the stabilizer formalism~\cite{GotPhD} to simulate the post-selection of the state $\ket{\splus\sfx\sfy_\theta}$ for each measurement (by applying suitable corrections should the measurement yield $\ket{\sminus\sfx\sfy_\theta}$ instead).

We may observe that the analysis of measurements by a single-qubit observable $\sO_v$ in the stabilizer formalism requires only that $\sO_v$ either commutes or anticommutes with every generator of the subgroup of the Pauli group which stabilizes the state of the system.
We may observe that $\XY$-plane measurements such as those above can be described by measurements with respect to observables
\begin{gather}
	\label{eqn:observableXY}
		\sO^{\XY,\theta}
	\;\;=\;\;
		\ket{\splus\sfx\sfy_\theta}\bra{\splus\sfx\sfy_\theta}
		\;\;-\;\;
		\ket{\sminus\sfx\sfy_\theta}\bra{\sminus\sfx\sfy_\theta}
	\;\;=\;\;
		\cos(\theta) X + \sin(\theta) Y	\;;
\end{gather}
such operators always anticommute with the Pauli $Z$ operator, and also anticommute with $X$ when $\theta = \pion 2$.
Then, a modified flow function $f: O\comp \to I\comp$ for a geometry $(G,I,O,M)$ describes a mapping from $u$ (with an associated observable $\sO^{\XY,\theta_u}_u$) to some qubit $v = f(u)$ such that $K_v$ anticommutes with $\sO^{\XY,\theta_u}_u$\,.
The constraints on the partial order $\preceq$ ensure that, for any linear measurement order extending $\preceq$, the generators which stabilize the state just prior to the measurement of $u$ consist of $K_{f(u)}$, and operators which all commute with $u$.
Then, the evolution of the state-space system under these measurements may be efficiently described via the stabilizer formalism as a transformation of stabilizer codes (although how individual states within these codes transform may be difficult to simulate).
As the modified flow conditions ensure that each measurement observable has a corresponding generator with which it anticommutes, the overall transformation will then be unitary.

\paragraph{A straightforward extension of modified flows.}

The fact that qubits $u \in M$ may have either $u \sim f(u)$ or $u = f(u)$ stems from the fact that $\sO^{\XY,\theta_u} = Y$ in this case: that is, the measurement basis for such qubits lies in the intersection of the \XY\ and \YZ\ planes.
The condition $u \sim f(u)$ in this case is common to all other measurement angles, as this condition suffices for treatment via the stabilizer formalism for observables which anticommute with $Z$.
We may then extend the notion of modified flows as follows to also include other measurement observables which, like $Y$, anticommute with $X$:
\begin{definition}
	\label{def:eflow}
	For a tuple $(G,I,O,M_{\XY},M_{\YZ})$ consisting of a geometry $(G,I,O)$ and sets $M_{\XY}, M_{\YZ} \subset O\comp$ such that $O\comp = M_{\XY} \union M_{\YZ}$, an \emph{extended flow}\footnote{%
		The definition of extended flows here is a refinement of the definition presented in~\cite{BeaudPhD}: the definition there can be recovered by requiring $M_{\XY}$ and $M_{\YZ}$ to be disjoint, in which case the set $T$ in Definition~3-17 of \cite{BeaudPhD} corresponds to $M_{\YZ}$\,.}
	is an ordered pair $(f, \preceq)$ consisting of a function $f: O\comp \to I\comp$, and a partial order $\preceq$ on $V(G)$, such that the conditions
	\begin{subequations}
	\begin{align}
			\label{eflow:a}		f(v) \,&\sim v	\;,											&&\mspace{-50mu}\text{for $v \in M_{\XY} \setminus M_{\YZ}$}\,;
		\\
			\label{eflow:b}		f(v) \,&= v	\;,												&&\mspace{-50mu}\text{for $v \in M_{\YZ} \setminus M_{\XY}$}\,;
		\\
			\label{eflow:c}		f(v) \sim v	\text{~~}&\text{or~~} f(v) = v	\;, 	&&\mspace{-50mu}\text{for $v \in M_{\YZ} \inter M_{\XY}$}\,;
		\\
			\label{eflow:d}		v \preceq&\, f(v)\;; 										&&\mspace{-50mu}\text{and}
		\\
			\label{eflow:e}		w \sim f(v)	\;&\implies\; v \preceq w
	\end{align}
	\end{subequations}
	hold for all $v \in O\comp$ and $w \in V(G)$.
\end{definition}
In the above definition, the sets $M_{\XY}$ and $M_{\YZ}$ correspond to qubits whose default measurement basis lies in the \XY-plane or the \YZ-plane, respectively; the intersection are those qubits whose measurement basis is the eigenbasis of the $Y$ operator, which is common to both.
Modified flows then correspond to the special case where $M_{\YZ} \subset M_{\XY}$.

As with modified flows, an extended flow $(f,\preceq)$ for a tuple $(G,I,O,M_\XY, M_\YZ)$ certifies that any procedure on the geometry $(G,I,O)$ which performs the appropriate types of measurement will perform a unitary transformation, provided also that it performs the appropriate corrections (or measurement adaptations) in response to the measurements.
To this effect, we may prove an extension of \autoref{lemma:extend-DK06}, as follows:
\begin{lemma}
	\label{lemma:extend-again-DK06}
	Suppose $(f,\preceq)$ is a modified flow for $(G,I,O,M_\XY,M_\YZ)$.
	Let $v \sim w$ denote the adjacency relation of $G$, and let $\lambda: V(G) \to \ens{\XY,\YZ}$ indicate the plane of measurement for each qubit, with $\lambda(u) = \XY$ or $\lambda(u) = \YZ$ for $u \in M_\XY \inter M_\YZ$, and $\lambda(u) = P$ for $u \in M_P$ otherwise. 
	Then for any linear order $\le$ extending $\preceq$, the measurement procedure
	\begin{gather}
	 	\sqparen{\ordprod[\le]_{u \in O\comp}
			\paren{\prod_{\substack{w \sim f(u) \\ w \ne u}} \Zcorr{w}{\s[u]}}
			\paren{\prod_{\substack{v = f(u) \\ v \ne u}} \Xcorr{v}{\s[u]}}
			\Meas[\lambda(u)]{u}{\theta_u}
		} \Ent{G} \New{I\comp}
	\end{gather}
	performs a unitary transformation.
\end{lemma}
Again, as we described above for modified flows, this may be proven for extended flows by simulating the evolution of the state-space via the stabilizer formalism, steering the state after each measurement to the one which would arise upon selection of the $\ket{\splus\sfx\sfy_\theta}$ or $\ket{\splus\sfy\sfz_\theta}$ result as appropriate.

\paragraph{Generalization of $\Zzz_d$ operations using extended flows.}

An alternative proof would be to extend the analysis of the proof of \autoref{lemma:extend-DK06} by considering the operation
\begin{gather}
	\label{eqn:defZzzTheta}
	 	\Zzz^\theta_d
	\;=\;
		\exp\big(\!-i\theta Z\sox{d} / 2\big)
	\;=\;
		\e^{-i \theta\smash{[\,\underbrace{\scriptstyle Z \ox \cdots \ox Z}_{\text{\small $d$ times}}]\,}/2}	\;,
	\\[-1em]\notag
\end{gather}
generalizing the operation $\Zzz_d$ defined in \eqref{eqn:defZzzTheta}.
We may define a \oneway\ procedure
\begin{gather}
	\label{eqn:ZzzThetaProcedure}
		\fZzz^{d,\theta}_{v_1, \ldots, v_d}
	\;=\;
		\paren{\prod_{j = 1}^d \Zcorr{v_j}{\s[a]}} \Meas[\YZ]{a}{\theta} \paren{\prod_{j = 1}^d \Ent{av_j}} \New{a}	\;:
\end{gather}
the geometry for this procedure is illustrated in \autoref{fig:zzMany}.
\begin{figure}[t]
	\begin{center}
		\includegraphics{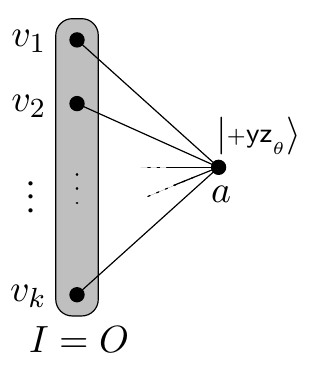}
	\end{center}
	\caption[Illustration of the geometry for the procedure $\fZzz_d^\theta$.]{\label{fig:zzMany}%
		Illustration of the geometry for the procedure $\fZzz_d^\theta$ defined in \eqref{eqn:ZzzThetaProcedure}, together with the ``default state'' $\ket{\splus\sfy\sfz_\theta}$ to be measured for the qubit $a$ in that procedure.
		(In the case that the measurement of $a$ yields $\ket{\sminus\sfy\sfz_\theta}$ instead, each element of $I = O$ is to be acted on with a $Z$ operation.)}
\end{figure}
Following the analysis of Section~3.8 of \cite{BB06}, we then have
\begin{gather}
	\label{eqn:ZzzThetaCompute}
		\fZzz^{d,\theta}_{v_1, \ldots, v_d}(\rho_S)
	\;=\;
		\big[ \Zzz^\theta_d \big]_{v_1, \ldots, v_d} \; \rho_S \; \big[ \Zzz^\theta_d \big]\herm_{v_1, \ldots, v_d}	\;.
\end{gather}
The procedure of \eqref{eqn:ZzzThetaProcedure} is a straightforward generalization of the procedure $\fZzz^d = \fZzz^{d,\pion 2}$\;.
In the same way that we map qubits $u$ such that $u = f(u)$ to a circuit $C_u = \Zzz_d$ acting on the neighbors of $u$ for a modified flow as a part of a circuit construction, we may map qubits such that $u = f(u)$ to a circuit $\tilde C_u = \Zzz^{\theta_u}_d$, where $\ket{\spm\sfy\sfz_{\theta_u}}$ is the default basis for the measurement of $u$.

\paragraph{Candidate circuits and star decompositions for extended flows.}

Just as geometries with modified flows correspond to circuits over the gate-set $\ens{J(\theta), \cZ, \Zzz_d}_{\theta \in \frac{\pi}{4}\Z,\, d\ge2}$ via \autoref{lemma:extend-DK06}, we may define a map from circuits over $\ens{\smash{J(\theta), \cZ, \Zzz^\theta_d}}_{\theta \in \frac{\pi}{4}\Z,\, d\ge2}$ to the \oneway\ model.
For an extended flow, we may apply the same definition for star geometries as for modified flows: the only extension required in the correspondence between star decompositions and extended flows is that qubits $u \in M_\YZ \setminus M_\XY$ are necessarily both the root and the center of a star geometry with respect to an extended flow.
It is then possible to show that star decompositions with respect to extended flows may be obtained with only minor modifications to the procedure \FindModStarDecomp\ defined in \Algorithm{alg:findModStarDecomp}, and its subroutine \RemoveStarGeom\ defined in \Algorithm{alg:removeStarGeom}.
In the former case, the set $M$ on \hyperref[line:mediatorInsert]{line~\ref*{line:mediatorInsert}} plays the role of $M_\YZ \subset M_\XY$ of qubits which may be both the root and center of a star geometry: it suffices them to replace $M$ with $M_\YZ$.
In \FindModStarDecomp, we must refine the condition on \hyperref[line:checkSuccessor]{line~\ref*{line:checkSuccessor}}, as there may not be qubits $w \in M_\YZ \setminus M_\XY$ which cannot be the center of a star geometry without also being the root: we must then revise the condition to require that
\begin{gather}
	\textit{$D$ has only one outbound arc $w \arc v$ \ANDALSO $v \in M_\XY$},
\end{gather}
the second condition of which was guaranteed for the special case of modified flows.
Applying the same analysis as for \Algorithm{alg:findModStarDecomp}, this modified algorithm then finds a star decomposition for $(G,I,O)$, with respect to the constraints described by $M_\XY$ and $M_\YZ$ on the possible relationship between the roots and centers of each star geometry.
In particular, as the proof of the extremal result of \autoref{lemma:extremal} carries forward for extended flows, the run-time of the modified algorithm is also $O(kn)$.

\subsubsection{Generalized certificates of unitarity based on stabilizers}
\label{sec:pauliFlows}

In the preceding section we observed one sense in which modified flows may be generalized, by observing that they essentially provide sufficient information to certify via the stabilizer formalism that a \oneway\ procedure performs a unitary transformation.
This is done by attributing some generator $K_v$ to each qubit $u \in O\comp$ to be measured, where the measurement on $u$ may be described by some observable $\sO_u$ which anticommutes with $K_v$, but commutes with $K_w$ for all $w \preceq u$.

Similar, but more general, criteria that extended flows were defined in~\cite{BKMP07} in \emph{generalized flows} and \emph{Pauli flows}, which extend further to arbitrary generators of the initial stabilizer group $\gen{K_w}_{w \in I\comp}$ (and thus to non-trivial products of the operators $K_w$) when $\XY$, $\YZ$, and $\XZ$-plane measurements are all allowed.
We may paraphrase the definition of Pauli flows as follows:
\begin{definition}
	\label{def:pauliFlow}
	Let $(G,I,O)$ be a geometry, and let $M_\XY, M_\YZ, M_\XZ \subset O\comp$ such that $O\comp = M_\XY \union M_\YZ \union M_\XZ$.
	Let $\cP$ denote the power-set function, and $\odd_G(S)$ the set of vertices in $G$ which are adjacent to an odd number of elements of a set $S \subset V(G)$.
	Then a \emph{Pauli flow} for $(G,I,O,M_\XY, M_\YZ, M_\XZ)$ is a tuple $(g,\preceq)$ consisting of a function $g: O\comp \to \cP(I\comp)$ and a partial order $\preceq$ which satisfy the following conditions:\footnote{%
		We may recover the definition of Pauli flows found in \cite{BKMP07} by defining a function $\lambda$ from $O\comp$ to $\ens{\pX, \pY, \pZ, \XY, \YZ, \XZ}$, such that
		\begin{itemize}
		\item 
			$\lambda(v) = \pX$ if $v \in M_\XY \inter M_\XZ$,
		\item
			$\lambda(v) = \pY$ if $v \in M_\XY \inter M_\YZ$,
		\item
			$\lambda(v) = \pZ$ if $v \in M_\XZ \inter M_\XZ$, and
		\item
			$\lambda(v) = P$ if $v \in M_{\!P}$ otherwise; 
		\end{itemize}
		the fact that $O\comp = M_\XY \union M_\YZ \union M_\XZ$ then implies that the conditions of \autoref{def:pauliFlow} and Definition~5 of \cite{BKMP07} are equivalent.}
	\begin{subequations}
	\begin{align}
			w \in g(v) \union \odd(g(v)) \;\implies\;& v \preceq w ,
		\\
			v \in \odd(g(v)) \setminus g(v)	\;\implies\;&	v \in M_\XY	\;,
		\\
			v \in g(v) \setminus \odd(g(v))	\;\implies\;& v \in M_\YZ	\;,
		\\
			v \in g(v) \inter \odd(g(v))	\;\implies\;& v \in M_\XZ	\;.
	\end{align}
	\end{subequations}
\end{definition}

\vspace{1ex}\noindent
From this definition, we recover extended flows if we require that the sets $g(u)$ be singleton sets for all $u$.
(The definition of generalized flows may be recovered by requiring that $M_\XY$, $M_\YZ$, and $M_\XZ$ be disjoint.)

Based on these conditions, we may obtain a further extension of \autoref{lemma:extend-DK06}, which we paraphrase from Theorem~4 of \cite{BKMP07} as follows:
\begin{lemma}
	\label{lemma:extend-DK06-PauliFlows}%
	Suppose $(g,\preceq)$ is a Pauli flow for $(G,I,O,M_\XY,M_\YZ,M_\XZ)$.
	Let $v \sim w$ denote the adjacency relation of $G$, and $\lambda: V(G) \to \ens{\XY,\YZ,\XZ}$ indicate the plane of measurement for each qubit so that for each qubit, $\lambda(u) = P$ only if $u \in M_P$.
	Then for any family $\ens{\theta_u}_{u \in O\comp}$ of measurement angles satisfying
	\begin{subequations}
	\begin{align}
		\begin{cases}[r@{}c@{}l]
			u \in \XZ &\inter& \YZ , ~~\text{or}									\\	
			u \in \XY \inter \XZ	~~&\text{and}&~~	\lambda(u)	=	\XY		\\
		\end{cases}
		\;\implies&\;\;	\theta_u \in \ens{0,\pi}	\;,
	\\[1ex]
		\begin{cases}[r@{}c@{}l]
			u \in \XY &\inter& \YZ , ~~\text{or}								\\	
			u \in \XY \inter \XZ	~~&\text{and}&~~	\lambda(u)	=	\XZ	\\
		\end{cases}
		\;\implies&\;\;	\theta_u \in \ens{\mpion 2, \pion 2}	\;,
	\end{align}
	\end{subequations}
	and for any linear order $\le$ extending $\preceq$, the measurement procedure
	\begin{gather}
	 	\sqparen{\ordprod[\le]_{u \in O\comp}
			\paren{\prod_{\substack{w \in \odd(g(u)) \\ w \ne u}} \Zcorr{w}{\s[u]}}
			\paren{\prod_{\substack{v \in g(u) \\ v \ne u}} \Xcorr{v}{\s[u]}}
			\Meas[\lambda(u)]{u}{\theta_u}
		} \Ent{G} \New{I\comp}
	\end{gather}
	performs a unitary transformation.
\end{lemma}
This represents a further extension of the possible range of \oneway\ procedures which may be identified as performing unitary transformations; the proof of the above Lemma may again be obtained by considering how each measurement may be certified to transform the state-space unitarily conditioned on a preferred measurement result, where the transformation of the state-space and the corrections to apply are described by the stabilizer formalism.

In order to use this to provide semantics for such \oneway\ procedures in terms of unitary circuits, we may consider how the concept of star geometries generalize from extended flows to Pauli flows.
The natural approach, rather than attempt to retain the star-graph structure, would be to consider geometries with a single root $u$, multiple ``centers'' $v \in g(u)$, and whose output subsystem consists of the qubits $[g(u) \union \odd(g(u))] \setminus \ens{u}$.
However, it is not as clear what the circuit corresponding to such a geometry would be, nor what results may be shown for the decomposition of a geometry $(G,I,O)$ into such geometries, as these lack the simple local structure of star geometries.
How a Pauli flow structure may be exploited to produce more general semantic maps for \oneway\ procedures in unitary circuit models is an open question.

\subsection{A simple \oneway\ construction lacking a star decomposition}
\label{sec:extendedConstructions}

\autoref{lemma:extend-DK06-PauliFlows} above presents a \oneway\ procedure which performs a unitary embedding, but for which the techniques of \autoref{sec:semanticMaps} do not apply.
Another example is the procedure for performing reversals of a linear array of logical qubits presented in Section~IV.A of \cite{RBB03}.
\begin{figure}[t]
	\begin{center}
		\includegraphics[width=\textwidth]{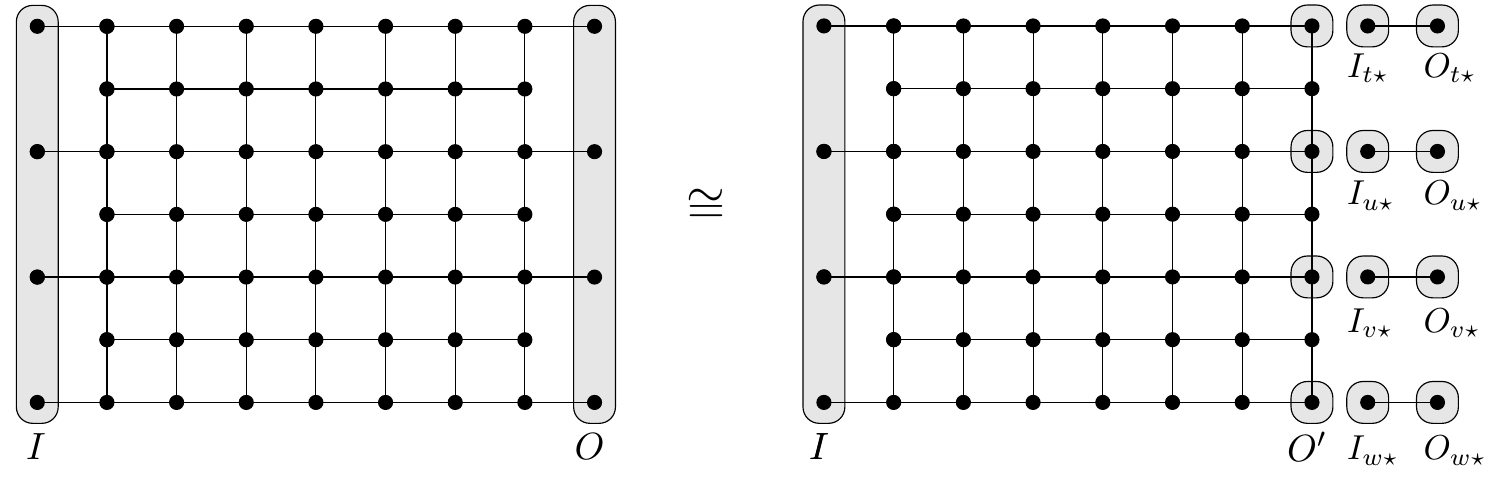}
	\end{center}
	\caption[Illustration of the geometry for a qubit-reversal procedure.]{\label{fig:swapRBB}%
		Illustration of the geometry for a qubit-reversal procedure of Section~IV.B of \cite{RBB03}, and a partial decomposition into star geometries.
		Each qubit in $V(G) \setminus O$ is to be measured in the \XY\ plane with default angle $0$: when this yields the measurement result $\s[u] = 0$ for each $u \in V(G) \setminus O$, the result is to perform a reversal of the vertical ordering of the logical qubits: input states $\ket{\psi}_{a,b,c,d}$ of the input subsystem (with $a$--$d$ ordered from top to bottom) are mapped to $\ket{\psi}_{w,v,u,t}$ of the output subsystem (with $t$--$w$ ordered from top to bottom).
		Notice that the residual geometry in the decomposition on the right has no maximal star geometries, as every element of $O'$ is adjacent to more than one element of $V(G) \setminus O'$, and every element of $V(G) \setminus O'$ is adjacent to at least one other element of $V(G) \setminus O'$.
		Thus, this geometry has no star decomposition, and therefore no modified flow.}
\end{figure}
\autoref{fig:swapRBB} illustrates the geometry underlying this \oneway\ procedure: the entire procedure may be obtained from the specification that each qubit in $O\comp$ is to be measured in the $\XY$ plane with the default angle $0$, and noting that as this represents selection of the $+1$ eigenstate of the observables $X_u$ for $u \in O\comp$, from which suitable correction operations may be obtained by the stabilizer formalism.
However, as \autoref{fig:swapRBB} also illustrates, this procedure has no star decomposition, and so cannot be mapped to its meaning as a reversal of a linear array of qubits using the techniques of \autoref{sec:semanticMaps}.

Despite this, we may gain some information about the qubit-reversal pattern using modified flows by an appropriate extension of the input and output systems, as exhibited in~\cite{CLM2005}.
It is easy to see that if we include the left-most element of each row into $I$, and the right-most into $O$, the resulting altered geometry has a flow with a successor function which maps each vertex to the one immediately to the right.
Chains in the natural pre-order $\preceq$ have monotonically increasing column numbers: and therefore $\preceq$ is antisymmetric.
Given that the default measurement angle for each qubit is zero, we may obtain a candidate circuit for this procedure, as illustrated in \autoref{fig:swapRBBredux}.
\begin{figure}[pt]
	\vspace{0.5em}
	\begin{center}
		~\hfill\includegraphics{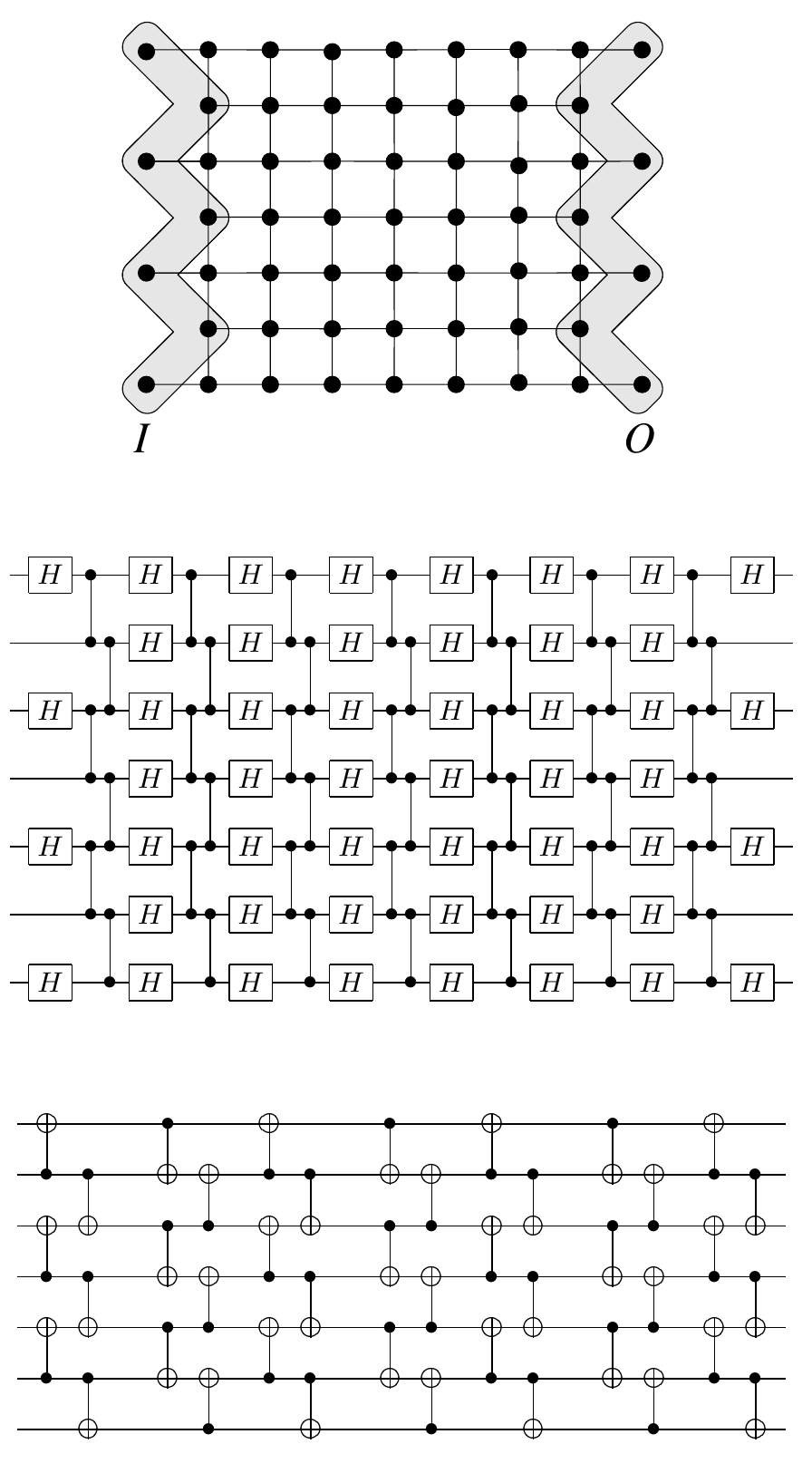}\hfill~
	\end{center}
	\caption[Illustration of the qubit-reversal procedure with augmented input/output subsystems, and corresponding candidate circuits.]{\label{fig:swapRBBredux}%
		Illustration of the qubit-reversal procedure of Section~IV.B of \cite{RBB03}, with augmented input/output subsystems, and corresponding candidate circuits (represented in the middle in the gate-set $\ens{H,T,\cZ}$, and in the bottom with \textsc{cnot} gates).}
	\vspace{0.5em}
\end{figure}
The final circuit illustrated there can easily be described as a reversible classical computation: for seven input bits $(v_0, \ldots, v_6)$, the computation performed by that circuit is
\begin{align}
	\lefteqn{(v_0,\; v_1,\; v_2,\; v_3,\; v_4,\; v_5,\; v_6)}\qquad\qquad&\notag\\
	\mapsto&\quad
	(v_6,\; v_4 \oplus v_5 \oplus v_6,\; v_4,\; v_2 \oplus v_3 \oplus v_4,\; v_2,\; v_0 \oplus v_1 \oplus v_2,\; v_0)	\;.
%
\end{align}
If we apply this circuit instead to qubits, with the odd-indexed qubits initialized to the $\ket{\splus}$ state (\ie\ a uniform superposition of the bit-values described above), those qubits remain in the $\ket{\splus}$ state, and the even-indexed qubits are reversed, as required.

We may then understand the qubit-reversal pattern in~\cite{RBB03} as corresponding to a circuit which prepares auxiliary qubits in the $\ket{\splus}$ state, and subsequently removes them by an $X$ observable measurement.
The fact that we are able to obtain the circuit in this case is possible because the highly structured geometry makes it easier to guess which qubit measurements correspond to terminal measurements in a unitary circuit.
Without exploiting this structure, however, there are no known approaches to \emph{algorithmically} provide semantics for such a measurement pattern in terms of unitary circuits.

\section{Review and open problems}
\label{sec:flowsReview}

In this article, we have presented an algorithm for obtaining the ``meaning'' of \oneway\ procedures $\fP$ in terms of unitary circuit models, using a combinatorial analysis of the structures underlying procedure $\fP$ arising from standard constructions of \oneway\ procedures.
In doing so, we have identified a correspondence between elementary combinatorial structures (\emph{star geometries}) and elementary circuits (\emph{star circuits}) which, given a \oneway\ procedure whose geometry which may be decomposed into these elementary structures, allows the construction of a corresponding circuit.
Together with a characterization of the dependencies of operations on classical measurement outcomes, this enables us to efficiently construct a circuit $C$ for a \oneway\ procedure $\fP$, and to determine whether $C$ and $\fP$ perform the same operation, when $\fP$ has input and output subsystems of the same size.

We have used the gate model $\ens{H,T,\cZ}$ as a reference model of circuits, using the related gate-sets $\cB_\DKP = \ens{J(\theta), \cZ}_{\theta \in \frac{\pi}{4}\Z}$ and $\cB_\RBB = \ens{J(\theta), \Zz}_{\theta \in \frac{\pi}{4}\Z}$ to describe constructions of \oneway\ procedures.
However, essentially all of the results of this article generalize to $\ens{H,R(\theta),\cZ}_{\theta \in \R}$, where
\begin{gather}
	R(\theta) \;=\; \e^{-i \theta Z / 2} \;=\;
	\begin{bmatrix}
		\e^{-i\theta/2}	&	\;\;0\;\;	\\	\;\;0\;\;	&	\e^{i\theta/2}
	\end{bmatrix}\,:
\end{gather}
we then have $J(\theta) = HR(\theta)$, and the gate sets $\ens{H,R(\theta),\cZ}_{\theta \in \R}$\,, $\ens{J(\theta), \cZ}_{\theta \in \R}$\,, and $\ens{J(\theta), \Zz}_{\theta \in \R}$ are all parsimonious in the sense of \autoref{def:parsimonious}.

There are three natural open problems: how this analysis may be extended to the case of unitary embeddings in general, how this may be extended to admit an analysis for quantum circuits with ancillas (qubits which are prepared, operated on, and subsequently removed without creating mixtures of states in the other qubits), and how we may extend this analysis to include stabilizer certificates such as Pauli flows.

\paragraph{General unitary embeddings.}
The problem presented by the case of a unitary embedding with $\card I < \card O$ is that a star decomposition of the geometry may not be unique.
Consequently, there may be many modified flow functions against which the dependencies of a \oneway\ procedure must in principle be tested.
We may relieve the situation if we can use the dependencies themselves to narrow the search space for a single modified flow; this seems a likely approach to extending the analysis of this article to a semantic map for the DKP and (simplified) RBB constructions, without further constraints.

\paragraph{Quantum circuits with ancillas.}
The qubit-reversal procedure described in \autoref{sec:extendedConstructions} presents a situation where the introduction and removal of ancillas can be used to obtain a unitary circuit which corresponds to a \oneway\ procedure.
Combinatorial tools for the ``detection'' of instances in which a qubit may represent the state of an ancilla qubit would yield more general tools for presenting the semantics of \oneway\ procedures in terms of unitary circuits, possibly permitting a semantic map to be defined for a construction employing all of the sub-procedures described in Section~IV of \cite{RBB03}, and for \oneway\ procedures in general arising from circuits involving partial trace operations on qubits which are disentangled from the rest of the system.

\paragraph{Extension to more general stabilizer certificates.}
As noted in \autoref{sec:eflow} and \autoref{sec:pauliFlows}, modified flows essentially represent a means via the stabilizer formalism to certify that, for a given geometry $(G,I,O)$, there is an efficiently describable \oneway\ procedure having that geometry $(G,I,O)$ and which performs a unitary transformation.
We can then consider how to compare an arbitrary  procedure with geometry $(G,I,O)$ to this canonical geometry.
In \autoref{sec:eflow}, we showed how such certificates may be extended to more general gate-sets: how we might exploit the definition of Pauli flows in \autoref{def:pauliFlow} to obtain more general semantic maps is an open problem.

\paragraph{Acknowledgements.}

This work is based substantially on results and remarks made in \cite{BeaudPhD}, which was written while the author was affiliated with the Institute for Quantum Computing at the University of Waterloo, and funded in part by MITACS and DTO-ARO.
The extension of the results to the simplified RBB construction was written during the author's current affiliation with the Unitersit\"at Potsdam, and was funded in part by EPSRC.
Finally, the author would to thank Ashwin Nayak, Andrew Childs, Michele Mosca, Debbie Leung, John Watrous, and Dan Browne for remarks on the presentation of \cite{BeaudPhD}.


\bibliographystyle{unsrt}
\bibliography{circConversion}

\appendix
\section{Stable-index tensor notation}
\label{apx:stableIndex}

\subsection{Sum-over-paths interpretation of stable-index tensors}
\label{apx:stableIndexPathIntegral}

The informal semantics of stable-index tensor notation is to provide a simplified description of circuits in terms of computational paths, as follows.

Stable-index notation for circuits is partially inspired by the \emph{path-labelling} of circuits by Dawson \etal~\cite{DHHMNO04}, which is explicitly concerned with descriptions of circuits in terms of computational paths: the tensor indices of the stable-index expressions correspond to labels of ``wire segments'' of~\cite{DHHMNO04}, which are separated by Hadamard gates (or more generally, other operations which do not preserve the standard basis).
For a given operator, a stable index $s$ then represents a qubit for which the standard basis states $\ket{0}$ and $\ket{1}$ remain unchanged (\ie\ the density operators $\ket{s}\bra{s}$ are unaffected) by the operation for $s \in \ens{0,1}$.
Over a single ``wire-segment'', the reduced state of such a qubit will then remain unchanged, even if one measures this state in the standard basis at the beginning of the wire-segment.

As with tensors in Einstein notation, we may describe a ``path'' as a mapping of each index (whether bound or free) to some fixed value; the amplitude associated with a path is the product of the corresponding coefficient of each factor involved in the expression.
The convention of summing over bound indices (whose possible values may be thought to represent intermediate configurations which were not observed, and therefore not distinguished from one another) then represents a summation over indistinguishable paths.
The simplification introduced by stable-index expressions corresponds to explicit recognition when operations leave some of the the ``dynamical variables'' associated to a path (represented by tensor indices) unaffected.

\subsection{Constructing stable-index representations}
\label{apx:stableIndexConstruction}

We may automatically generate a stable-index representation $S$ for a unitary circuit from a representation as a composition of operators $C$ on specified qubits in an efficient manner.
We consider two distinct methods: one which is appropriate for a network without any topological constraints on the operators, and one which is motivated by the constraints described in \autoref{sec:topologicalConstraints}.

\subsubsection{Construction for topologically unconstrained circuits}
\label{apx:stableIndexConstructionUnconstrained}

Let $L$ be the set of qubit labels in $C$.
We may then transliterate the terms in $C$, defining a set of tensor indices incrementally as we do so, as follows:
\begin{enumerate}
\item
	For each qubit $v \in C$, we define the tensor index $v_0$\,.
	We designate $v_0$ as \emph{the current tensor index} for $v$.
	(Throughout the transliteration of the circuit, the index designated as ``current'' for a qubit will be updated.)
	We also start with an expression $S$ consisting of the null sequence of operators: this may be regarded as representing a stable-index representation for the the empty product (the identity operator).

\item
	We translate each operator $U_{a,b,\ldots}$ in $C$ taken in sequence, 
	as follows.
	For each qubit $v \in \ens{a,b,\ldots}$ on which $U$ acts (again, in sequence), let $v_j$ be the current index for $v$, and perform the following:
	\begin{romanum}
	\item
		If $U_{a,b,\ldots}$ introduces $v$ as a ``fresh'' qubit (or a new qubit which is possibly entangled with the others), then $v_0$ is still the current index for $v$.
		We then use $v_0$ as the advanced index corresponding to $v$, without a corresponding deprecated index.

	\item
		\label{item:constructionNonStableIndex}%
		If $U_{a,b,\ldots}$ not preserve each of the standard basis states on $v$ (that is, if $U_{a,b,\ldots} \ket{s}\bra{s}_v U_{a,b,\ldots}\herm \ne \ket{s}\bra{s}_v$ for some $s \in \ens{0,1}$), then we define a new tensor index $v_{j+1}$ for $v$.
		We then use $v_j$ as the deprecated index corresponding to $v$, and $v_{j+1}$ as the corresponding advanced index.
		We update the current index for $v$ to $v_{j+1}$.

	\item
			If $U_{a,b,\ldots}$ preserves each of the standard basis states on $v$ (that is, for each $s \in \ens{0,1}$ we have $U_{a,b,\ldots} \ket{s}\bra{s}_v U_{a,b,\ldots}\herm = \ket{s}\bra{s}_v$), then we use $v_j$ as a stable index corresponding to $v$, without defining any new indices.\footnote{%
			In the definition of stable-index notation in \autoref{sec:stableIndexNotation}, all stable indices must occur as a contiguous block at the beginning of the linear ordering of the qubits operated on.
			Generalizing this definition, we may easily alternate between sequences of stable indices and sequences of advanced and deprecated indices.
			If we choose not to generalize stable-index notation in this way, we may only represent a qubit $v$ with a stable index if all preceding qubits on which $U$ operates are also represented by stable indices; otherwise, we must apply case \ref*{item:constructionNonStableIndex} above regardless of how $U$ acts on $v$.}

	\item
		If $U_{a,b,\ldots}$ removes a qubit $v$ (\eg\ describing an operation conditioned on a a particular measurement outcome), then we use $v_j$ as a deprecated index corresponding to $v$, without a corresponding advanced index.
	\end{romanum}
	This produces sequences of stable indices $\vec s$ and sequences of matched advanced and deprecated indices $\vec a$ and $\vec d$ (possibly with placeholders); we then translate $U_{a,b,\ldots}$ as $U\pseu[ \vec s : \vec a / \vec d ]$, and append it to $S$.

	\item
		Having translated each of the terms in $C$, the sequence $S$ is a stable-index tensor expression corresponding to $C$.
\end{enumerate}
Note that in particular, every new index which is advanced by a gate $U\pseu[\vec s:\vec a/\vec d]$ must be distinct from any other index being advanced, as well as from all of the indices which have occurred to that point.
This ensures that the resulting product is acyclic if the original network $C$ was acyclic.

Note that the index successor function $f$ (\autoref{def:indexSuccessorFn}) can be easily described from this construction: $f$ has as a domain all those indices $v_j$ such that $v_{j+1}$ is well-defined, and $v_{j+1} = f(v_j)$ for all such indices.

\subsubsection{Construction suitable for use in the simplified RBB construction}
\label{apx:stableIndexConstructionConstrained}

Motivated by the constraints described in \autoref{sec:topologicalConstraints} on \oneway\ procedures whose entanglement graph is an induced subgraph of the grid, we also present a construction for stable-index expressions for circuits with a linear-nearest neighbor topology with the gate set $\cB_\RBB = \ens{J(\theta), \Zz}_{\theta \in \frac{\pi}{4}\Z}$, in which the two-qubit gates
\begin{romanum}
\item
	act ``stably'' on both their indices,

\item
	can only act on indices of equal ``depth'' from the input --- that is, can only act on pairs of indices $v_j$ and $w_j$, where $v_j = f^j(v_0)$ and $w_j = f^j(w_0)$ for some indices $v_0, w_0 \notin \dom(f)$, where $f$ is the successor function of the stable-index expression; and

\item
	can only act on a pair of indices $v_j, w_j$ if there is no two-qubit gate acting on $v_{j-1}$ and $w_{j-1}$.
\end{romanum}
Given a circuit $C$ using the gates of $\cB_\RBB$, we may produce such a stable-index representation (and at the same time explicitly construct its interaction graph, applying \autoref{def:interactionHypergraph} in the special case of interactions of at most two qubits) as follows, adapting the algorithm of \autoref{apx:stableIndexConstructionUnconstrained} above.

\begin{enumerate}
\item 
	For each qubit $v$ acted on by $C$, we set the current index to $v_0$.
	We also initialize the graph $G$ to the totally disconnected graph whose vertices are the indices $v_0$\,.
	We start with an expression $S$ consisting of the null sequence of operators: this may be regarded as representing a stable-index representation for the the empty product (the identity operator).

\item
	We translate each operator $J(\theta)_a$ or $\Zz_{ab}$ in $C$ taken in sequence, as follows.
	For each qubit $v$ on which the gate acts, let $v_j$ be the current index for $v$.
	For a gate $J(\theta)$, we append $J(\theta)\pseu[:v_{j+1}/v_j]$ to $S$, add the index $v_{j+1}$ and the edge $v_j v_{j+1}$ to $G$, and set the current index to $v_{j+1}$.
	Otherwise, for a gate $\Zz_{vw}$, we perform the following operations:
	\begin{romanum}
	\item
		We ensure that the constraints on successive applications of two-qubit gates is satisfied as follows.
		If the current indices for $v$ and $w$ are $v_j$ and $w_j$ for a common value of $j$, and there is an edge of the form $v_{j-1} w_{j-1}$ in $G$, we append the expression
		\begin{align}
				J(0)\pseu[:v_{j+2}/v_{j+1}] J(0)\pseu[:v_{j+1}/v_j]
				J(0)\pseu[:w_{k+2}/w_{k+1}] J(0)\pseu[:w_{k+1}/w_k] 			 
		\end{align}
		to $S$, and add the edges $v_j v_{j+1}$\,, $v_{j+1}v_{j+2}$\,, $w_jw_{j+1}$\,, and $w_{j+1}w_{j+1}$ (along with the indicated indices) to $G$.
		We then set the current indices of $v$ and $w$ to $v_{j+2}$ and $w_{j+2}$ respectively.

	\item
		We then recursively perform the following, letting $v_j$ and $w_k$ stand for the current indices of $v$ and $w$ at each iteration:
		\begin{itemize}
		\item
			If $\abs{k - j} \ge 3$ or $\abs{k - j} = 2$, we append either
			\begin{align}
				J(0)\pseu[:v_{j+2}/v_{j+1}] J(0)\pseu[:v_{j+1}/v_j]
			&&
				\text{or}
			&&
				J(0)\pseu[:w_{k+2}/w_{k+1}] J(0)\pseu[:w_{k+1}/w_k] 			 
			\end{align}
			to $S$, according to whether $k > j$ or $j > k$ respectively, adding the appropriate indices and edges to $G$.
			We set the current index either of $v$ or $w$ to $v_{j+2}$ or $w_{k+2}$ accordingly, and repeat for the new current indices.

		\item
			If $\abs{k - j} = 3$ or $\abs{k - j} = 1$, we append either
			\begin{align}
			\begin{split}
				&J(\pion2)\pseu[:v_{j+3}/v_{j+2}] J(\pion2)\pseu[:v_{j+2}/v_{j+1}] J(\pion2)\pseu[:v_{j+1}/v_j]
			\quad
				\text{or}
			\\[1ex]
				&J(\pion2)\pseu[:w_{k+3}/w_{k+2}] J(\pion2)\pseu[:w_{k+2}/w_{k+1}] J(\pion2)\pseu[:w_{k+1}/w_k]			 
			\end{split}
			\end{align}
			to $S$, according to whether $k > j$ or $j > k$ respectively, adding the appropriate indices and edges to $G$.
			We set the current index either of $v$ or $w$ to $v_{j+3}$ or $w_{k+3}$ accordingly, and repeat for the new current indices.

		\item
			If $j = k$, we append $\Zz[v_j , w_j]$ to $S$, add the edge $v_j w_j$ to $G$, and stop.
		\end{itemize}
	\end{romanum}

	\item
		Having translated each of the terms in $C$, the sequence $S$ is a stable-index tensor expression corresponding to $C$, and $G$ its interaction graph.
\end{enumerate}
In the above construction, note that the final index $v_\ell$ for each qubit $v$ may not have the same value of $\ell$; and that such stable-index expressions may both begin and end with operations $\Zz\pseu[v_j, w_j]$.
These degrees of freedom may complicate the procedure of composing two such circuits.
These issues may both be mitigated by appending appropriate products of $J(\theta)$ operators which evaluate to the identity, of differing lengths for each qubit; this can be done to both uniformize the value $\ell$ of the final indices $v_\ell$, and also ensure that any $\Zz$ operation is a distance of at lest $2$ from the end of the circuit.
However, for the purposes of this article, such operations will not be necessary, as we do not consider the composition of two such expressions.

\subsection{Congruence of circuits and isomorphic stable-index expressions}
\label{apx:stableIndexCongruence}

In this section, we prove \autoref{thm:stableIndexCongruence} (page~\pageref*{thm:stableIndexCongruence}), and expand on the role of parsimonious sets of gates in this result.

\begin{lemma}
	\label{lemma:stableIndexIsomorphicCongruent}
	Let $C_1$ and $C_2$ be two sequences of unitary operators on a common set of qubits, composed from a common set of unitaries.
	If the stable[index tensor representations of $C_1$ and $C_2$ are isomorphic, then $C_1$ and $C_2$ are congruent up to re-ordering of commuting operators.	
\end{lemma}
\begin{proof}
	Consider the coarsest equivalence relation $\cong$ on circuits which consist of permutations of the gates of $C_1$, such that $C \cong C'$ if $C$ differs from $C'$ by a transposition of commuting gates.
	By definition, the equivalence classes of $\cong$ are sets of circuits which are congruent up to re-ordering of commuting operators.
	Consider stable-index tensor expressions $S$ and $S'$ arising from two unitary circuits $C$ and $C'$, where $S$ acts on the same index set as $S'$, and $S$ differs from $S'$ only by a transposition of two terms $U_1\pseu[\vec s_1: \vec a_1/\vec d_1]$ and $U_2\pseu[\vec s_2:\vec a_2/\vec d_2]$, where without loss of generality the former precedes the latter in $C$, and vice-versa in $C'$.
	Then, the corresponding gates $U_1$ and $U_2$ in $C$ occur with $U_1$ preceding $U_2$, and occur in $C'$ in the opposite order.
	Consider the index sets for these two operations in the stable-index tensor expression:
	\begin{itemize}
	\item
		If the sequences of indices $(\vec s_1; \vec a_1; \vec d_1)$ and $(\vec s_2; \vec a_2; \vec d_2)$ do not overlap, the operations $U_1$	and $U_2$ act on disjoint sets of qubits, and so they commute.
	
	\item
		Suppose the sequences of indices $(\vec s_1; \vec a_1; \vec d_1)$ and $(\vec s_2; \vec a_2; \vec d_2)$ overlap, and let $v$ be an index occurring in both sequences.
		In the construction of $S$, because $v$ occurs in $U_2\pseu[\vec s_2:\vec a_2/\vec d_2]$, it cannot be a deprecated index in $U_1\pseu[\vec s_1:\vec a_1/\vec d_1]$\,; and because it occurs in $U_1\pseu[\vec s_1:\vec a_1/\vec d_1]$, it cannot be an advanced index in $U_2\pseu[\vec s_2:\vec a_2/\vec d_2]$.
		By the construction of $S'$, we similarly have that $v$ cannot be a deprecated index in $U_2\pseu[\vec s_2:\vec a_2/\vec d_2]$ or an advanced index in $U_1\pseu[\vec s_1:\vec a_1/\vec d_1]$.
		Then, $v$ is a stable index of both terms.
	\end{itemize}
	Therefore, the only qubits $v$ on which $U_1$ and $U_2$ both act are those whose standard basis states are preserved by both $U_1$ and $U_2$.
	Then, $U_1$ and $U_2$ commute, in which case $C \cong C'$.

	Let $S_1$ and $S_2$ be the stable-index tensor representations of $C_1$ and $C_2$, where the isomorphism is given by an index relabelling map $\lambda$ and mapping of the terms $\tau$.
	Let $S'_2$ be the representation obtained by applying $\lambda$ to all of the indices in $S_2$\,: then $S_1$ and $S'_2$ consist of a product of the same terms in different orders.
	Without loss of generality, we will then suppose that $S_2$ and $S_1$ differ only by a reordering of terms.
	Because the terms of $S_1$ differ from those of $S_2$ by a permutation, by induction on the decomposition of this permutation into transpositions, we therefore have $C_1 \cong C_2$.
\end{proof}
Whatever set of gates is used in a stable-index tensor expressions, rearranging the terms in the expression not only preserves the meaning (by summation over indices which are both advanced and deprecated) as a unitary operator, but also preserves the relevant structural information about the circuit --- preserving the necessary order information in the way described following \autoref{def:indexSuccessorFn}, while discarding such information for each pair of commuting operators.
As a consequence, we may regard permuting terms of such an expression as a way of exploring the elements of the congruence class of a circuit.

However, in the case where two gates commute without preserving the standard basis states of their common operands, we may obtain congruent circuits which do not have isomorphic tensor index expansions:
\begin{example}
	Let $U \in \U(2)$ be a non-diagonal operator, and define the operation $\ctrl U$ as in \eqref{eqn:controlFunctor}.
	Consider the two circuits $\ctrl U_{c,b} \; \ctrl U_{a,b} \; H_a$ and $\ctrl U_{a,b} \; \ctrl U_{c,b} \; H_a$ on three qubits $a,b,c$ (where the subscripts here denote the linear order of the qubits as operands of $\ctrl U$ and $H$): these have the respective stable-index expressions
	\begin{align}
			\ctrl U\pseu[c:b_3/b_2] \ctrl U\pseu[a_1:b_2/b_1] H\pseu[:a_1/a_0]
		&&
			\text{and}
		&&
			\ctrl U\pseu[a_1:b_3/b_2] \ctrl U\pseu[c:b_2/b_1] H\pseu[:a_1/a_0]	\;.
	\end{align}
	These two expressions are non-isomorphic: in the first expression, the leftmost factor contains the indices $c$ (which is neither advanced nor deprecated in the entire expression) and $b_3$ (which is not deprecated in the entire expression); the second expression does not contain any factor with these properties.
\end{example}
In the example above, the inequivalence of the two stable-index expressions is due to the fact that the two operations $\ctrl U_{a,b}$ and $\ctrl U_{b,c}$ preserve the eigenbasis of $U$ on $b$, rather than the basis $\ket{0}, \ket{1}$.
The restrictions imposed on parsimonious sets of unitary operations prevent precisely this problem from arising, which allows us to prove the following:
\begin{lemma}
	\label{lemma:stableIndexCongruentIsomorphic}
	Let $C_1$ and $C_2$ be two unitary circuits composed from a common parsimonious set of unitaries.
	If $C_1$ is congruent to $C_2$ up to re-ordering of commuting operators, then the stable-index representations of $C_1$ and $C_2$ will be isomorphic.
\end{lemma}
\begin{proof}
	Let $K_0 \cong K_1 \cong \cdots \cong K_\ell$ be a sequence of circuits which differ only by a transposition of two commuting gates, where $C_1 = K_0$ and $C_2 = K_\ell$.
	We may show by induction on $\ell$ that the stable-index tensor expressions for $C_1$ and $C_2$ are isomorphic; it suffices in this case to prove the case $\ell = 1$.

	Let $S_0$ and $S_1$ be the stable-index representations of $K_0$ and $K_1$.
	Up to a relabeling, we may assume that the indices of $S_0$ and $S_1$ are all of the form $v_0, v_1, \ldots$ for different qubits $v$ acted on by $K_0$ and $K_1$: we set $v_0$ to be the first index corresponding to a given qubit in either $S_0$ or $S_1$, and $v_{j+1}$ to be an advanced index corresponding to each deprecated index $v_j$ for $j \in \N$.
	Let $U_1$ and $U_2$ be the pair of gates whose order differs between the two circuits; and consider the truncation $S'_0$ of $S_0$ just prior to the terms corresponding to $U_1$ and $U_2$, and similarly for $S'_1$.
	Because $K_0$ and $K_1$ differ only by a transposition of commuting gates $U_1$ and $U_2$, $S'_0 = S'_1$.
	Let $S''_0$ and $S''_1$ be the truncations of $S_0$ and $S_1$ just after the terms corresponding to $U_1$ and $U_2$\,:
	\begin{itemize}
	\item
		If $U_1$ and $U_2$ act on disjoint sets of qubits, the terms corresponding to $U_1$ and $U_2$ in $S''_0$ may differ from those in $S''_1$ only by the indices which are advanced in each; and as the indexing scheme we have fixed is the same for each, and the indices which will be advanced for $U_1$ and $U_2$ are independent of each other in both expressions, $S''_0$ differs from $S''_1$ only by a transposition of those two terms.

	\item
		If $U_1$ and $U_2$ both have the standard basis as their eigenbases, their corresponding terms in $S_0$ and $S_1$ only have stable indices, and therefore do not advance any new indices.
		Then the term in $S_0$ and in $S_1$ corresponding to $U_1$ are the same, and similarly for $U_2$\,; the difference between $S''_0$ and $S''_1$ is a transposition of those two terms, and so they are isomorphic.

	\item
		If $U_1$ and $U_2$ perform the same single-qubit operation, there will be some common tensor index $u_j$ which is current for that qubit prior to the application of $U_1$ and $U_2$ in both $S'_0$ and $S'_1$. 
		Then, the terms corresponding to $U_1$ and $U_2$ in $S_0$ will be $U_2\pseu[:u_{j+2}/u_{j+1}]U_1\pseu[:u_{j+1}/u_j]$, and the corresponding terms in $S_1$ will be $U_1\pseu[:u_{j+2}/u_{j+1}]U_2\pseu[:u_{j+1}/u_j]$ for some indices $v''$ and $w''$.
		Given that $U_1$ and $U_2$ perform the same operation, we then have $S''_0 = S''_1$.
	\end{itemize}
	Because the operations of $K_0$ and $K_1$ are identical after $U_1$ and $U_2$, the differences between $S_0$ and $S_1$ after the terms corresponding to $U_1$ and $U_2$ can only arise from the indices advanced by each term; because we have fixed the indexing scheme, the only difference between $S_0$ and $S_1$ is then the (possibly trivial) transposition of those two terms.
	Thus, $S_0$ and $S_1$ are isomorphic.
\end{proof}
\autoref{thm:stableIndexCongruence} then follows from \autoref{lemma:stableIndexIsomorphicCongruent} and \ref*{lemma:stableIndexCongruentIsomorphic}.
Using this, we may identify congruent circuits simply by relabelling their indices in a uniform way (\eg\ as in the construction for stable-index expressions in \autoref{apx:stableIndexConstruction}), and determining whether there is a perfect matching between the two expressions defined by equality of terms.

\section{Details for constructions of \oneway\ procedures}
\label{apx:constructions}

\subsection{Construction of the gate-sets $\cB_\DKP$ and $\cB_\RBB$}
\label{apx:constructGatesApx}

In this section, we provide details on the construction of the \oneway\ procedures $\fJ^\theta$ and $\fZz$, as well as a generalization of the latter operation to a many-qubit diagonal operation.

\subsubsection{Two-qubit operations}
\label{apx:constructGatesApx-2qubit}

As described in \eqref{eqn:ZzPattFormula}, we may implement $\Zz$ in the \oneway\ model with the procedure
\begin{align}
	\label{eqn:ZzPattFormula-a}
		\fZz_{v,w}
	\;=\;
		\Zcorr{v}{\s[a]} \Zcorr{w}{\s[a]} \Meas{a}{\pion2} \Ent{av} \Ent{aw} \New{a}	\;.
\end{align}
For the analysis of \autoref{sec:semanticMaps}, it will be convenient to generalize this to a procedure $\fZzz^d_{v_1, \ldots, v_d}$ acting symmetrically on many qubits:
\begin{gather}
	\label{eqn:ZzzProcedure-a}
		\fZzz^d_{v_1, \ldots, v_d}
	\;=\;
		\paren{\prod_{j = 1}^d \Zcorr{v_j}{\s[a]}} \Meas{a}{\pion 2} \paren{\prod_{j = 1}^d \Ent{av_j}} \New{a}	\;.
\end{gather}
We may compute the effect of this procedures as follows.
For standard basis vectors $\ket{x_j}, \ket{x'_j}$ for each $1 \le j \le d$, ignoring the measurement result $\s[a]$ after the correction operations, we may obtain
\begin{subequations}
\begin{align}
	\label{eqn:ZzzCompute'}
	\mspace{20mu}&\mspace{-20mu}
		\fZzz^d_{v_1,\ldots,v_d}\Bigg( \bigotimes_{j = 1}^d \ket{x_j}\bra{x'_j}_{v_j} \Bigg)
	\notag\\[1ex]=&\;
		\paren{\prod_{j = 1}^d \Zcorr{v_j}{\s[a]}} \Meas{a}{\pion 2}
		\Bigg( \Big[ Z^{x_1+\cdots+x_d} \ket{+}\bra{+} Z^{x'_1+\cdots+x'_d} \Big]_a \ox \Bigg( \bigotimes_{j = 1}^d \ket{x_j}\bra{x'_j}_{v_j} \Bigg) \Bigg)
	\notag\\[1ex]
	=&\;\;\;
		\sigma^+ \Bigg( \bigotimes_{j = 1}^d \ket{x_j}\bra{x'_j}_{v_j} \Bigg)
 	\;\,+\;\;
		\sigma^- \Bigg( \bigotimes_{j = 1}^d \Big[ Z \ket{x_j}\bra{x'_j} Z\Big]_{v_j} \Bigg)	\,,
\end{align}
where we define the non-negative scalars $\sigma^\pm$ by
\begin{align}
 		\sigma^+
	\;=&\;\,
		\bra{\smash{+_{\pi\!/\!2}}} Z^{x_1+\cdots+x_d} \ket{+}\bra{+} Z^{x'_1+\cdots+x'_d} \ket{\smash{+_{\pi\!/2\!}}}	\;,
	\\
 		\sigma^-
	\;=&\;\,
		\bra{\smash{-_{\pi\!/\!2}}} Z^{x_1+\cdots+x_d} \ket{+}\bra{+} Z^{x'_1+\cdots+x'_d} \ket{-_{\smash{\pi\!/\!2}}}	\;.
\end{align}
\end{subequations}
We may consider the terms of~\eqref{eqn:ZzzCompute'} as outer products of standard-basis vectors: we may verify that 
\begin{subequations}
\label{eqn:projectPlusZzz-a}%
\begin{align}
		\Big[ \bra{\smash{+_{\pi\!/\!2}}} Z^{x_1+\cdots+x_d} \ket{+} \Big] &\ox \ket{x_1}_{v_1} \ox \cdots \ox \ket{x_d}_{v_d} 
	\notag\\[1ex]=&\;\;
		\sfrac{1}{2} \big[ 1 - i(-1)^{x_1 + \cdots + x_d} \big] \ket{x_1 \cdots x_d}_{v_1 \cdots v_d}
	\notag\\[1ex]=&\;\;
		\sfrac{1}{\sqrt 2} \, \e^{i\pi\big[2(x_1+\cdots+x_d) - 1\big]\!/4} \ket{x_1 \cdots x_d}_{v_1 \cdots v_d}
	\notag\\=&\;\;
		\sfrac{1}{\sqrt 2} \Big[ \e^{i \pi [Z \ox \cdots \ox Z]/4} \ket{x_1 \cdots x_d} \Big]_{v_1 \cdots v_d} \;,
	\intertext{}
		\Big[ \bra{\smash{-_{\pi\!/\!2}}} Z^{x_1+\cdots+x_d} \ket{+} \Big] &\ox \Big[ Z_v \ket{x_1}_{v_1} \Big] \ox \cdots \ox \Big[ Z_w \ket{x_d}_{v_d} \Big]
	\notag\\[1ex]=&\;\;
		\sfrac{1}{2} \big[ (-1)^{x_1 + \cdots + x_d} + i \big] \ket{x_1 \cdots x_d}_{v_1 \cdots v_d}
	\notag\\[1ex]=&\;\;
		\sfrac{i}{2} \big[ 1 - i(-1)^{x_1 + \cdots + x_d} \big] \ket{x_1 \cdots x_d}_{v_1 \cdots v_d}
	\notag\\=&\;\;
		\sfrac{i}{\sqrt 2} \Big[ \e^{i \pi [Z \ox \cdots \ox Z]/4} \ket{x_1 \cdots x_d} \Big]_{v_1 \cdots v_d} \;.
\end{align}
\end{subequations}
By taking outer products, and extending~\eqref{eqn:ZzzCompute'} linearly over tensor products of the rank-1 projectors $\ket{x_1}\bra{x'_1} \ox \cdots \ox \ket{x_d}\bra{x'_d}$ with linear operators on an external system $S'$, we then obtain
\begin{align}
	\label{eqn:ZzzCompute-a}
		\fZzz^d_{v_1,\cdots,v_d}\big( \rho_S \big)
	\;=&\,\;
		\e^{-i\pi [ Z_{v_1} \ox \cdots \ox Z_{v_d} ] / 4}
		\; \rho_S \;\; \e^{i\pi [ Z_{v_1} \ox \cdots \ox Z_{v_d} ] / 4}	\;,
\end{align}
for arbitrary operators $\rho_S$,  acting on a system $S$ which includes the qubits $v_j$ but not the qubit $a$.
In the special case $d = 2$, we then recover
\begin{align}
	\label{eqn:ZzCompute-a}
		\fZz_{v,w}\big( \rho_S \big)
	\;=&\,\;
		\e^{-i\pi [Z_v \ox Z_w] / 4}
		\; \rho_S \;\; \e^{i\pi [Z_v \ox Z_w] / 4} 
	\;=\;
		\Zz_{vw} \; \rho_S \; \Zz_{vw} \;,
\end{align}
as claimed in \eqref{eqn:ZzCompute}.

\subsubsection{Single qubit operations}
 
We may describe \oneway\ procedures for the operators $J(\theta)$ with a similar analysis as for the $\Zz$ operator above.
For an arbitrary angle $-\pi < \theta \le \pi$ and qubits $v$ and $w$, define the CPTP map $\fJ^\theta_{w/v}$ by
\begin{gather}
	\label{eqn:JpattFormula-a}
		\fJ^\theta_{w/v}(\rho_S)
	\;=\;
		\Xcorr{w}{\s[v]} \Meas{v}{-\theta} \Ent{vw} \New{w}	\;,
\end{gather}
for $\rho$ a state supported on a set of qubits $S$ which includes $v$ but not $w$.
We may characterize the effect of $\fJ^\theta_{w/v}$ as follows.
For $\ket{\psi} \in \cH$ arbitrary, we have
\begin{align}
		\fJ^\theta_{w/v}\big(\ket{\psi}\bra{\psi}_v\big)
	\;=&\;\,
		\Xcorr{w}{\s[v]} \Meas{v}{-\theta} \Big[ \cZ_{vw} \Big( \ket{\psi}\bra{\psi}_v \ox \ket+\bra+_w \Big) \cZ_{vw} \Big]	\;.
\end{align}
Ignoring the measurement result $\s[v]$ after the correction operation, we may characterize the compound operation $\Xcorr{w}{\s[v]} \Meas{v}{\varphi}$ on a joint pure state $\ket{\Psi}_{v,w}$ by
\begin{align}
		\Xcorr{w}{\s[v]} \Meas{v}{\varphi} \Big(\ket{\Psi}\bra{\Psi}_{v,w}\Big)
	 \;=&\;\;
		\Big( \bra{+_\varphi}_v \ket{\Psi}_{v,w} \Big)\Big( \bra{\Psi}_{v,w}\ket{+_\varphi}_v \Big)
	\notag\\&\quad
		+\;
		\Big( X_w \bra{-_\varphi}_v \ket{\Psi}_{v,w} \Big)\Big( \bra{\Psi}_{v,w} \ket{-_\varphi}_v X_w \Big) \,.
\end{align}
In the case $\ket{\Psi}_{v,w} = \cZ_{v,w} \ket{\psi}_v \ket{+}_w$, the terms in this equation may again be expressed as outer products of state vectors:
\begin{subequations}
\begin{align}
		\bra{+_\varphi}_v \Big[ \cZ_{v,w} \ket{\psi}_v \ket{+}_w \Big]
	\;\;=&\;\;
		\Big[ \bra{+_\varphi}_v \cZ_{v,w} \ket{+_\varphi}_w \Big] \ket{\psi}_v
	\notag\\=&\;\;
		J(-\varphi)_{w/v} \ket{\psi}_v	\;,
	\\[2ex]
		X_w \bra{-_\varphi}_v \Big[ \cZ_{v,w} \ket{\psi}_v \ket{+}_w \Big]
	\;\;=&\;\;
		\Big[ X_w \bra{-_\varphi}_v \cZ_{v,w} \ket{+_\varphi}_w \Big] \ket{\psi}_v
	\notag\\=&\;\;
		\Big[ \bra{+_\varphi}_v \cZ_{v,w} \ket{+_\varphi}_w \Big] \ket{\psi}_v
	\notag\\=&\;\;
		J(-\varphi)_{w/v} \ket{\psi}_v	\;,
\end{align}
where by $J(-\varphi)_{w/v}$ we denote the linear operator $\cH[v] \to* \cH[w]$ which maps state-vectors $\ket{\psi}_v \,\mapsto\, \big[\, J(-\varphi) \ket{\psi} \,\big]_w$\,.
\end{subequations}%
By linearity, we may extend this to arbitrary entangled states $\ket{\psi}_v \otimes \ket{\phi}_{S'}$ with a system $S'$ which does not include $v$.
Then, for a system $S$ which includes $v$ and excludes $w$, we have
\begin{align}
		\fJ^\theta_{w/v}\Big(\ket{\Psi}\bra{\Psi}_S\Big)
	\;=&\;\,
		\Xcorr{w}{\s[v]} \Meas{v}{-\theta} \Big[ \cZ_{vw} \Big( \ket{\Psi}\bra{\Psi}_S \ox \ket+\bra+_w \Big) \cZ_{vw} \Big]
	\notag\\\;=&\;\,
		J(\theta)\big._{w/v} \ket{\Psi}\bra{\Psi}_S J(\theta)\herm_{v/w}	\;.
\end{align}

\subsection{Efficiency of the DKP and simplified RBB constructions}
\label{apx:constructionEfficiency}

Both the DKP construction and the (simplified) RBB construction, as described in \autoref{sec:constructionAlgorithms}, can be performed efficiently in the size of the input circuit.
We may illustrate this for both procedures, as follows.
In both cases, let $k$ be the number of qubits which $C$ acts upon, $N$ be the number of one-qubit gates, and $M$ be the number of two-qubit gates: we may assume that $k \le N + M$ if we suppose that $C$ acts non-trivially on each qubit.

\subsubsection{Obtaining a circuit in normal form}
	We may easily construct a normal form in each case in a single pass of the expression of the circuit $C$.
	For each qubit $v$ acted on by $C$, we may maintain a buffer $B_v$ of single-qubit operations which are performed.
	These are initially empty for each qubit; we accumulate $T$ and $T\herm$ operations (cancelling them as appropriate) of each qubit as we encounter them; and upon encountering a Hadamard operation on a qubit $v$, we insert the contents of the buffer immediately before the Hadamard, resetting the buffer to zero after the Hadamard operation has been traversed.
	In both cases, we may also cancel $\cZ$ operations by storing in the buffers $B_v$ indications of which qubits $w$ the qubit $v$ has interacted with via a $\cZ$: initially this is set to zero, and we toggle the presence of each qubit in the set with each $\cZ$ operation between the two qubits.

	When we empty a buffer for a qubit $v$, we also produce the corresponding $\cZ$ gates between $v$ and other qubits $w$.
	When we clear a buffer $B_v$ prior to a Hadamard, we may count the number of $\cZ$ gates between $v$ and each qubit $w$, and produce a $\cZ$ operation when the total is odd.
	For each $\cZ$ gate between $v$ and some other qubit $w$, we also remove a $\cZ$ gate in the buffer $B_w$ between $w$ and $v$ to avoid double-counting the $\cZ$ gates.
	For the RBB construction, we may elaborate this further by converting $\cZ$ operations to $\Zz$ operations by performing the substitution $\cZ \mapsto \Zz(T\herm \ox T\herm)^2$ when we encounter such an operation, accumulating the resulting $T\herm$ gates in the buffers $B_v$ and $B_w$ when we do so.

	The total amount of work performed by this routine per operation in $C$ (counting both insertions and removals of terms corresponding to these operations, to or from the buffers) is $O(1)$.
	Using list-structures to implement the buffers $B_v$ and adjacency-list representations for the $\cZ$ operations, we may initialize the buffers in time $O(k)$.
	The amount of work required by this phase is then $O(k + N + M)$.

\subsubsection{Conversion to stable-index representation}
	This phase may be performed simultaneously with the previous phase: with each gate produced, we may produce a stable-index representation of the gate, and we produce $J(\theta)$ gates rather than $H$, $T$, and $T\herm$ gates when we clear the buffers.
	In the case of the RBB construction, an additional amount of work may be required for the insertion of $J(\theta)$ gates to achieve uniform wire lengths: the amount of additional work per $\Zz$ gate is bounded to within a constant by the number of single-qubit operations performed on either wire, which is at most $N$.
	The amount of work required to do this is then $O(k + N + M)$ for the DKP construction, and $O(k + NM)$ for the RBB construction.

\subsubsection{Translation of individual operations to the \oneway\ model}
	This phase can also be performed simultaneously with the previous phases: rather than produce a stable-index expression $C''$, we may produce the \oneway\ procedure $\Phi\big(U\pseu[\cdots]\big)$ for each term $U\pseu[\cdots]$ produced for the stable-index expansion.
	As each such \oneway\ procedure is of constant length, this requires no additional overhead asymptotically.

\subsubsection{Obtaining a \oneway\ procedure in normal form}
	\label{apx:obtainingOneWayNormalForm}
	Let $\fP_0$ be the resulting \oneway\ procedure from the previous phase.
	Let $n$ be the number of qubits acted on by $\fP_0$, and and let $m$ be the number of operations $\fJ^\theta$, $\Ent{}$, or $\fZz$ in $\fP_0$.
	(We may easily show that $n = k + N$ for the DKP construction and $n = k + N + M \in O(N + M)$ for the RBB construction; we have $m = N + M$ for the DKP construction and $m \le N + NM \in O(NM)$ for the RBB construction, essentially by the arguments given for the run-time for the production of the stable-index expression $C''$.)
	As each operation $\fJ^\theta$, $\Ent{}$, and $\fZz$ consists of a constant number of elementary operations and contains only one or two entangling operations, the procedure $\fP_0$ also has $O(m)$ elementary operations and $O(m)$ entangling operations.

	\paragraph{Standardizing $\fP_0$.}
	Let $G$ be the entanglement graph associated to $\fP_0$.
	In the process of standardizing $\fP_0$, for each qubit $v$ which is part of a $\fJ_{v/u}$ operation for some $u$, we must commute a single operation $\Xcorr{v}{\s[u]}$ past entangling operations $\Ent{vw}$ for all $w$ adjacent to $v$ in the graph $G$.
	Rather than explicitly commuting operations, we may standardize $\fP_0$ by preparing a new \oneway\ procedure $\fP_1$ as follows:
	\begin{enumerate}
	\item
		We initialize $\fP_1$ by assigning to it the procedure consisting of all preparation maps $\New{v}$ in $\fP_0$, followed by all entangling maps $\Ent{vw}$.
		(Note that the sets of prepared qubits and the graph $G$ of entangling operators can be obtained with no additional asymptotic cost, during the phase of translation to the \oneway\ model in which $\fP_0$ is produced.)

	\item
		Note that every deprecated index $v_j$ in the stable-index expansion of $C''$ has a corresponding advanced index $v_{j+1}$.
		As the procedures $\fJ^\theta_{v/u}$ are the only operations in the DKP construction which allocate or discard qubits, each qubit $u$ which is discarded in $\fP_0$ has a corresponding qubit $v$ which is allocated, and this correspondence arises precisely from the presence of an operation $\fJ^\theta_{v/u}$ in $\fP_0$.
		If $f$ is the index successor function of $C''$, we then have $v = f(u)$ for such qubits $u$ and $v$.
		For the RBB construction, we similarly have a gate $\fJ^\theta_{v/u}$ for qubits with $f(u) = v$; but we also have qubits $a$ which are allocated and discarded in the operations $\fZz$, which lie outside of the domain of $f$.
		We may then extend $f$ to a function $\tilde f$ such that $\tilde f(a) = a$ for such mediating qubits, and $\tilde f(v_j) = f(v_j) = v_{j+1}$ for $v_j \in \dom(f)$.
		(Such a function $\tilde f$ corresponds to an extended successor function as defined in \autoref{def:extendedSuccFn}.)
		For the DKP construction, we will use $\tilde f = f$.
		Note that we may again construct $\tilde f$ at no additional cost during the phase in which $\fP_0$ is produced from the stable-index expression $C''$.
	
		We may use the function $\tilde f$ to append the measurements of $\fP_0$ to the left-hand side of $\fP_1$ in order, with the correct sign- and bit-dependencies included.
		We may make a single pass of $\fP_0$, and in doing so maintain a list of correction operations for each of the qubits in $\fP_0$, as follows.
		Initially, no qubit has any corrections.
		For any measurement $\Meas{u}{\ast}$ encountered, we perform the following:
		\begin{romanum}
		\item
			If $u$ does not have any accumulated corrections, we append $\Meas{u}{\theta}$ to the left-hand side of $\fP_1$ without any dependencies;

		\item
			If $u$ has accumulated corrections $\Xcorr{u}{\beta}$ and $\Zcorr{u}{\gamma}$, we append $\Shift{u}{\gamma} \Meas{u}{\theta;\beta}$ to the left-hand side of $\fP_1$.
			(If $\gamma = 0$, we may omit the shift operation.)

		\item 
			We add an $\Xcorr{v}{\s[u]}$ correction to $v = \tilde f(u)$, provided that $v \ne u$, corresponding to the $X$ correction arising from an operation $\fJ^\theta_{v/u}$.
			(Note that the condition $u \ne \tilde f(u)$ is satisfied for all $u$ in the DKP construction; it is satisfied if and only if $u \in \dom(f)$ in the RBB construction.)

		\item
			We add $\Zcorr{w}{\s[u]}$ corrections for every $w$ adjacent to $v = \tilde f(u)$ in the graph $G$, provided $w \ne u$.
			(In the DKP construction, and for qubits $u \in \dom(f)$ in the RBB construction, these corrections arise from commuting the $\Xcorr{v}{\s[u]}$ operation past the entangling maps $\Ent{vw}$ for $w \ne u$; for qubits $a \notin \dom(f)$ in the latter construction, this corresponds to the $\Zcorr{w}{\s[a]}$ corrections in the $\fZz$ procedure, in which case the condition $w \ne u$ will be satisfied for all $w$.)
		\end{romanum}
		We may simplify the expressions for any accumulated corrections using the relations $\Xcorr{v}{\beta} \Xcorr{v}{\s[u]} = \Xcorr{v}{\beta + \s[u]}$ and $\Zcorr{w}{\beta} \Zcorr{w}{\s[u]} = \Zcorr{w}{\beta + \s[u]}$.

		The amount of work performed for each measurement $\Meas{u}{\ast}$ in $\fP_0$ is at most $O\big(\deg[\tilde f(u)]\big)$, where $\deg(v)$ is the number of neighbors of $v$ in $G$.
		The total number of operations performed in this stage is then $O\big(\smash{\sum\limits_v \deg(v)}\big) =$\linebreak $O\big(\card{E(G)}\big) = O(m)$.
		
	\item
		After we have read all of the measurements, we append the corrections for each unmeasured qubit (computed in the previous step) to to the end of $\fP_1$.
	\end{enumerate}
	The resulting procedure $\fP_1$ is then the same that would result from standardizing $\fP_0$, and can be obtained in time $O(m)$. 

	\vspace{1em}
	\paragraph{Normalizing $\fP_1$.}
	Obtaining a normal form $\bar\fP$ from $\fP_1$ can be done by checking for Pauli simplifications for each measurement, and then performing signal shifting by accumulating shift operators to the left (without commuting shift operators past each other).
	We may check for and apply Pauli simplifications as a part of the standardization process above, with no additional overhead asymptotically.
	We may also perform signal shifting simultaneously with the standardization process, by immediately commuting the shift operation $\Shift{u}{\gamma}$ past any corrections $\Xcorr{v}{\s[u]}$ and $\Zcorr{w}{\s[u]}$ which might have been introduced by the measurement on a given qubit $u$: if $v \sim w$ denotes the adjacency relation in $G$, we have
	\begin{align}
			\mspace{50mu}&\mspace{-50mu}
	 		\paren{\prod_{\substack{w \sim \tilde f(v) \\ w \ne u}} \Zcorr{w}{\s[u]}}
			\paren{\prod_{\substack{v = \tilde f(v) \\ v \ne u}} \Xcorr{v}{\s[u]}} 
			\Shift{u}{\gamma}
		\notag\\[0.5ex]=&\;\;
			\Shift{u}{\gamma}
	 		\paren{\prod_{\substack{w \sim \tilde f(v) \\ w \ne u}} \Zcorr{w}{\s[u] + \gamma}}
			\paren{\prod_{\substack{v = \tilde f(v) \\ v \ne u}} \Xcorr{v}{\s[u] + \gamma}}	\;. 
	\end{align}
	After performing this commutation, there are no classically controlled operations to be inserted to the left of the shift operator which will depend on $\s[u]$; we may then commute it past all additional operations and remove it from the expression.
	Then, rather than inserting shift operations and the corrections described, we may accumulate the corrections described on the right-hand side after the measurement on $u$.
	The amount of work required for each measurement in this alternative procedure is then dominated by the work required to update the corrections, which is $O(n \deg[\tilde f(u)])$; summed over all vertices $u$, obtaining a normal form may then be done in time $O(nm)$.

\subsubsection{Total run-time complexity}

Comparing the complexity of the phases above (and recalling the relationship between $n$ and $m$ with $N$, $M$, and $k$ at the beginning of \autoref{apx:obtainingOneWayNormalForm}, the run-time of both the DKP and RBB constructions are then dominated by the phase of obtaining a \oneway\ procedure in normal form.
This requires time $O(nm)$ for both constructions, which evaluates to $O((k + N)M)$ for the DKP construction and $O((N+M)NM)$ for the RBB construction.
This bound on the run-time coincides with that of the semantic map $\cS$ described in \autoref{sec:semanticAlg}.

\end{document}

%% file: preamble.tex
\usepackage{hyperref}
\usepackage{amsthm}
\usepackage{graphicx}
\usepackage{jrnbsymb}
\usepackage[mathscr]{euscript}
\usepackage{upgreek}
\usepackage{textcomp}
\usepackage{nomencl}
\usepackage{times}
\usepackage{color}
\usepackage{aliascnt}
\usepackage{enumitem}
\usepackage{scalefnt}

\definecolor{oldcolor}{rgb}{0.25,0,0}

\definecolor{fluxcolor}{rgb}{0,0.25,0}

\definecolor{todo}{rgb}{0.8,0,0}

\def\todo[#1]{\textcolor{todo}{\textbf{[#1]}}}

\renewcommand\vec\mathbf
\newcommand\unit[1][e]{\vec{\hat{#1}}}

\def\Kdelta[#1][#2]{\delta_{#1,#2}}

\newcommand\sfrac[2]{\frac{\text{\raisebox{-0.5ex}{$\textstyle \smaller{#1}$}}}{\text{\raisebox{0.3ex}{$\textstyle \smaller{#2}$}}}}

\newcommand\ssstyle{\scriptscriptstyle}

\makeatletter
\newcommand\ordprod[1][<]{%
	\@ifnextchar_{\@ordprod[#1]}{\@ordprod[#1]_{}}}
\def\@ordprod[#1]_#2{%
	\@ifnextchar^{\@@ordprod[#1]_{#2}}{\@ordprod[#1]_{#2}^{}}}
\def\@@ordprod[#1]_#2^#3{%
	\Bigg.\smash{\prod\limits_{#2}^{#3}}\big.^{#1}\:\:}
\makeatother

\newlength\spacewidth
\input{custalgorithm.tex}

\makeatletter
	\def\texttilde{\raise.2ex\hbox{\m@th$\scriptstyle\sim$}}
	\def\@nbsp{~}
	\newcommand\myurlstyle[1]{{\small \ttfamily #1}}
	
	\renewcommand\url[1]{[\href{#1}{\myurlstyle{#1}}]}
\makeatother

\let\P\relax
\let\H\relax

\makeatletter
\newcommand\smaller{%
	\ifmmode\def\@tempa{\small@scalemath}\else\def\@tempa{\small@scaletext}\fi
	\@tempa}

\newcommand\larger[1][6/5]{%
	\ifmmode\def\@tempa{\@scalemath[#1]}%
	\else\def\@tempa{\@scaletext[#1]}\fi
	\@tempa}

\newdimen\scaled@size

\long\def\small@scalemath#1{\mathchoice{%
	\mbox{\small@scaletext{$\displaystyle #1$}}}{%
	\mbox{\small@scaletext{$\textstyle #1$}}}{%
	\mbox{\small@scaletext{$\scriptstyle #1$}}}{%
	\mbox{\small@scaletext{$\scriptscriptstyle #1$}}}}

\long\def\small@scaletext#1{%
	\scaled@size=\f@size pt\relax
	\ifdim\scaled@size<6pt\scaled@size=5pt\relax\else
		\ifdim\scaled@size<7pt\scaled@size=6pt\relax\else
			\ifdim\scaled@size<8pt\scaled@size=7pt\relax\else
				\ifdim\scaled@size<9pt\scaled@size=8pt\relax\else
					\ifdim\scaled@size<10pt\scaled@size=9pt\relax\else
						\multiply\scaled@size by 5\relax
						\divide\scaled@size by 6\relax
					\fi
				\fi
			\fi
		\fi
	\fi
	\edef\@tempa{\curr@fontshape}%
	\edef\@tempb{\strip@pt\scaled@size}%
	\bgroup
		\edef\font@name{%
			\csname\@tempa/\@tempb\endcsname}%
		\edef\f@size{\@tempb}%
	   \pickup@font
	   \font@name
  		\@@enc@update
		#1%
	\egroup}

\makeatother

\newcommand\ssfStyle[1]{\textmd{\textup{\textsf{#1}}}}
\newcommand\srmStyle[1]{\textup{\textrm{\smaller{#1}}}}
\newcommand\parrmStyle[1]{\textrm{\textmd{\textup{(#1\kern0.1ex)}}}}
\newcommand\paritStyle[1]{\textrm{\textmd{\textup{(\textit{#1}\kern0.1ex)}}}}
\newcommand\parbfitStyle[1]{\textrm{\textbf{\textup{(\textit{#1}\kern0.1ex)}}}}

\newcommand\complexityStyle[1]{\ssfStyle{#1}}
\newcommand\gateStyle[1]{\textsc{#1}}
\newcommand\problemStyle[1]{\srmStyle{\bf #1}}
\newcommand\pauliStyle[1]{\ssfStyle{#1}}

\newcommand\setStyle[1]{\textrm{\mdseries\upshape{#1}}}

\NewNameCommandStyle{f}{\ensuremath{\mathfrak{#1}}}

\NewNameCommandStyle{Set}{\setStyle{#1}}

\NewNameCommandStyle{ComplexityClass}{\complexityStyle{#1}}
\NewNameCommandStyle{GateSet}{\gateStyle{#1}}
\NewNameCommandStyle{Problems}{\problemStyle{#1}}
\NewNameCommandStyle{Planes}{\pauliStyle{#1}}
\NewNameCommandStyle{p}{\pauliStyle{#1}}
\NewNameCommandStyle{rm}{\parrmStyle{#1}}
\NewNameCommandStyle{it}{\paritStyle{#1}}
\NewNameCommandStyle{bfit}{\parbfitStyle{#1}}
\NewNameCommandStyle{s}{\mathscr{#1}}

\let\L\relax

\namecommands{\cA, \cB, \cC, \cE, \cF, \cG, \cK, \cM, \cP, \cR, \cS, \cU, \cW}
\namecommands{\sI, \sT, \sG, \sD, \sS, \sO, \sP}
\namecommands{\sfC, \sfT, \sfF, \sfH, \sfJ, \sfP, \sfProb, \sfN, \sfQ, \sfR, \sfS, \sfz, \sfx, \sfy, \sfn, \sfm, \sfs}
\namecommands{\fC, \fP, \fJ, \fZz, \fZzz}

\namecommands[bb]{\A}
\namecommands[sf]{\E, \L, \D, \J}
\namecommands[Set]{\U, \SU, \Herm, \RBB, \DKP, \GL}

\namecommands[ComplexityClass]{\BQP, \BPP, \PP, \P, \NP, \co}
\namecommands[Problems]{\MPI, \GMPI}
\namecommands[GateSet]{\M, \Gates, \Phases, \Clifford, \Cluster, \Graph, \OneWay, \Jacz, \CoherCtrl, \Cond}

\namecommands[Planes]{\XY, \YZ, \XZ}
\namecommands{\pX, \pY, \pZ}

\NameMathOperators{\poly, \polylog, \Inp, \Outp, \odd, \Lower}

\namecommands{\iti, \itii, \itiii, \bfiti, \bfitii, \bfitiii, \bfitiv, \rma, \rmb, \rmc, \rmd, \rme, \rmf}

\let\oldtr\tr
\renewcommand\tr[1][]{\oldtr_{\text{\raisebox{-0.3ex}{$#1$}}}}

\makeatletter
\newcommand\Zz{\mathpalette\@Zz{}}
\def\@Zz#1{\mbox{$#1\text{\scalefont{1.7}$#1z$}\!z$}}
\newcommand\Zzz[1][]{\mathpalette\@Zzz{}{#1}}
\def\@Zzz#1#2{\mbox{$#1\text{\scalefont{1.7}$#1z$}\!z\!z^{#2}$}}
\makeatother

\newcommand\deffunction[1]{\SetKwFunction{#1}{#1}}

\deffunction{MaximalStarVertex}
\deffunction{TestDependencies}
\deffunction{Semantic}

\deffunction{FindFlow}
\deffunction{FindModStarDecomp}
\deffunction{GetLayers}
\deffunction{RemoveIdops}
\deffunction{BuildCircuit}

\deffunction{RemoveStarGeom}

\deffunction{TestNaturalPreorder}
\deffunction{FindInfima}
\deffunction{InitializePath}

\deffunction{VisitVertex}
\deffunction{ObtainMaxFamilyPaths}
\deffunction{AlternPathSearch}

\SetKw{LET}{Let}
\SetKw{SUCHTHAT}{such\,that}
\SetKw{NEXT}{skip to next}
\SetKw{RETURN}{return}

\SetKw{TO}{to}
\SetKw{ANDALSO}{and}
\SetKw{ORELSE}{or}
\SetKw{NOT}{not}

\ResetInOut{Input}
\SetKwInOut{Data}{Data}
\SetKwInOut{Effect}{Effect}

\namecommands[tt]{
	\true, \false,
	\iter,
	\foundpath, 
	\fixed, \nil, \marked,
	\status, \untouched, \incomplete, \complete,
	\foundcircuit,
	\nextpath,
	\partialOrd,
	\signDep, \xDep, \zDep,
	\xType, \zType, 	
	\Inf, \Path, \Dist,
	\starDecomp,
	\flow, \maximalStar,
	\emptyLayer, \layer, \neighbors
 }

\namecommand[c]{\H}	
\newcommand\cH[1][2]{\H_{#1}}

\newcommand\symdiff{\text{\raisebox{0.2ex}{\,$\scriptstyle \triangle$\,}}}

\newcommand\splus{\smaller+}
\newcommand\sminus{\smaller-}
\newcommand\spm{\smaller\pm}

\def\tpm{\text{\tiny$\pm$}\!\!\:}

\newcommand\eg{\textit{e.\nolinebreak\hspace{-0.1ex}g.}\nolinebreak}
\newcommand\ie{\textit{i.\nolinebreak\hspace{-0.1ex}e.}\nolinebreak}
\newcommand\cf{\textit{c.\nolinebreak\hspace{-0.1ex}f.}\nolinebreak}

\newcommand\etal{\textit{et~al.\nolinebreak}}
\newcommand\oneway{\textup{1WQC}}

\newcommand\fn{\emph{\small f$^{\text{\raisebox{0.09ex}{\kern0.32ex\underline{\kern-0.2ex{n}}}}}$}}
\newcommand\fns{\emph{\small f$^{\text{\raisebox{0.09ex}{\kern0.32ex\underline{\kern-0.2ex{n\:\!\!s}}}}}$}}

\makeatletter
	\newcommand\partialto{%
		\@ifstar{\rightharpoonup}{\DOTSB\relbar\joinrel\rightharpoonup}}
\makeatother

\newcommand\arc{\protect\to*}

\newcommand\thru{\,--\:\:\!}

\newcommand\minus{\text{--}}

\newcommand\sox[1]{^{\mspace{2mu} \ox #1}}
\newcommand\comp{^{\textsf{c}}}				

\newcommand\New[1]{\:\!\mathsf{N}_{#1}\:\!}
\newcommand\Ent[1]{\:\!\mathsf{E}_{#1}\:\!}

\newcommand\Xcorr[2]{\:\!\mathsf{X}_{#1}^{#2}\:\!}
\newcommand\Zcorr[2]{\:\!\mathsf{Z}_{#1}^{#2}\:\!}
\newcommand\Shift[2]{\:\!\mathsf{S}_{\s[#1]}^{#2}\:\!}
\newcommand\Op[2]{\:\!\mathsf{Op}_{#1}^{#2}\:\!}

\makeatletter
\newcommand\Meas[3][]{\@Meas[#1]{#2}{#3}}

\def\@Meas[#1]#2#3{%
	\:\!%
	\def\@tempa{}
	\@ifempty{#1}{%
		\@tempa{\mathsf{M}_{#2}^{#3}}%
	}{%
		\@tempa{\mathsf{M}_{#2}^{(#1, #3)}}%
}\:\!}

\newcommand\s[1][]{%
	\mathsf{s}\@ifempty{#1}{}{\smaller{[#1]}}}

\let\ringaccent\r
\def\r{\@ifstar{\ringaccent}{\meas@result}}
\newcommand\meas@result[1][]{%
	\mathsf{r}\@ifempty{#1}{}{{\mathchoice
		{\text{\footnotesize $[#1]$}}%
		{\text{\small $[#1]$}}%
		{\scriptscriptstyle [#1]}%
		{\text{\tiny $[#1]$}}%
}}}
\makeatother

\newcommand\ctrl{\text{\raisebox{0.3ex}{\small $\wedge$}}\mspace{-1.5mu}}	
\newcommand\cZ{\ctrl Z}																		

\makeatletter
\newcommand\pion[2][]{
	\mathchoice%
		{{\text{\small\:\!\raisebox{0.3ex}{\raisebox{0.5ex}{\smaller{$#1\pi$}}$\!\!\!\;/\!_{#2}$}}}}
		{{\text{\small\:\!\raisebox{0.3ex}{\raisebox{0.5ex}{\smaller{$#1\pi$}}$\!\!\!\;/\!_{#2}$}}}}
		{{\scriptscriptstyle \text{\raisebox{0.2ex}{\raisebox{0.5ex}{\smaller{$#1\pi$}}$\!\!\!\:/\!_{#2}$}}}}
		{\big.\text{\raisebox{0.2ex}{\raisebox{0.5ex}{\smaller{$#1\pi$}}$\!\!\!\:/\!_{#2}$}}}}

\newcommand\mpion{\pion[\text{--}\:\!]}
\makeatother

\newcommand\paren[1]		{\left( #1 \right)}		     						 	
\newcommand\sqparen[1]	{\left[ #1 \right]}		   					  		
\newcommand\ens[1]		{\left\{#1\right\}}										
\newcommand\ket[1]		{\left\lvert#1\right\rangle\mspace{-1.5mu}}		
\newcommand\bra[1]		{\mspace{-1.5mu}\left\langle#1\right\rvert}		
\newcommand\card[1]		{\left\lvert #1 \right\rvert}							

\newcommand\anglebr[1]	{\left\langle #1 \right\rangle}
\newcommand\gen\anglebr

\newcommand\ngbh[2][]{N_{\ssstyle #1}(#2)}

\let\abs\card
\newcommand\supp[1]		{^{(#1)}}

\makeatletter
\def\pseu[#1]{%
	\mspace{-5mu}%
	\mathchoice
		{\@pseu\@iden\scriptstyle[#1:]}%
		{\@pseu{\big.\smash}\scriptstyle[#1:]}%
		{\@pseu{\big.\smash}\scriptscriptstyle[#1:]}%
		{\@pseu{\big.\smash}{\scriptscriptstyle\smaller}[#1:]}%
	\mspace{2mu}}
\def\@pseu#1#2[#3:#4]{%
	\def\@tempa{#4}
	\ifx\@tempa\@empty
		\left[\mspace{2mu}#1{\text{\raisebox{0.2ex}{$#2{#3}$}}}\mspace{2mu}\right]%
	\else
		\@@pseu{#1}{#2}[#3:#4]%
	\fi
}
\def\@@pseu#1#2[#3:#4/#5:]{%
	\def\@tempa{#3}%
	\ifx\@tempa\@empty\else
		\def\@tempa{\raisebox{0.2ex}{$#2{#3}$}\,}%
	\fi
	\left[#1{\@tempa\tfrac{#2{#4}}{#2{#5}}}\right]%
}
\makeatother

\makeatletter
	\def\matrix@check{\@gobble}
\makeatother

\renewenvironment{cases}[1][c]{%
 	\left\{ \begin{array}{#1@{\quad}l}%
}{%
	\end{array} \right\}%
}

\theoremstyle{definition}
\newtheorem{definition}{Definition}[section]

\newtheorem*{notation*}{\adhoclabel}

\theoremstyle{plain}
\newtheorem{theorem}{Theorem}[section]

\newaliascnt{lemma}{theorem}
\newtheorem{lemma}[lemma]{Lemma}
\aliascntresetthe{lemma}

\newaliascnt{proposition}{theorem}

\aliascntresetthe{proposition}

\theoremstyle{plain}
\newtheorem{example}{Example}%

\newcommand\Algorithm[1]{\hyperref[#1]{Algorithm~\ref*{#1}}}

\providecommand\autopageref[1]{\hyperref[#1]{page~\pageref*{#1}}}

\makeatletter

\@addtoreset{equation}{section}
\makeatother


\makeatletter
\newcommand\textbalance[2][]{%
	\@ifnextchar[{\@textbalance[#1]{#2}}{\@textbalance[#1]{#2}[]}}
\def\@textbalance[#1]#2[#3]{\phantom{#3}#1#2#3\phantom{#1}}
\makeatother

\makeatletter
\newcounter{inlinum}
\renewcommand\theinlinum{\textup{(\alph{inlinum})}}

	\newenvironment{inlinum}{%
		\bgroup
			\def\item{\refstepcounter{inlinum}\theinlinum~\ignorespaces}%
 			\catcode`\^^M=10%
			\setcounter{inlinum}{0}%
			\ignorespaces
	}{%
		\egroup\aftergroup\@after@inlinum
	}
	\def\@after@inlinum{%
		\catcode`\^^M=10%
		\gdef\@tempa{\catcode`\^^M=5}%
		\expandafter\@tempa\ignorespaces}%
\makeatother

\newenvironment{romanum}{%
	\begin{enumerate}[label=\textup{(\textit{\roman*})}]%
}{%
	\end{enumerate}%
}

\newcounter{descenum}
\newenvironment{descenum}{%
	\let\olditem\item
	\def\item[##1]{\olditem[\refstepcounter{descenum}\thedescenum. ##1]}
	\setcounter{descenum}{0}%
	\begin{description}[font=\mdseries\itshape,topsep=1ex]%
}{%
	\end{description}%
}

\newlist{subfigenum}{enumerate}{1}
\setlist[subfigenum,1]{labelindent=0.5em,leftmargin=2em,label=\textup{(\textit{\alph*})}}
\newenvironment{subfiglist}{%
	\begin{minipage}{\textwidth}%
 	\begin{subfigenum}%
}{%
 	\end{subfigenum}%
	\end{minipage}%
}

\newcounter{alphanum}
	{\end{list}\vspace{-1.4ex}}

\setdescription{font=\mdseries\itshape,topsep=1ex,itemsep=0.8em}

%% file: custalgorithm.tex
\usepackage[linesnumbered, boxed, longend]{algorithm2e}

\newcommand\captionstyle{\small}

\setalcapskip{0ex}
\SetAlTitleFnt{captiontitlestyle}
\SetAlCapFnt{\captionstyle}

\Setnlsty{}{\textit\bgroup}{\egroup \,\,\,}

\makeatletter
\renewcommand\thealgocf{\arabic{algocf}}
\renewcommand{\algocf@endline}{\string.}
\makeatother

\SetInd{0.2em}{0.8em}

\makeatletter

\renewcommand{\algocf@caption@boxed}{\vskip 1.5em\box\algocf@capbox}%
\renewcommand{\algocf@captiontext}[2]{#1\algocf@typo. \AlCapFnt{}#2} 

\makeatother

\makeatletter
\newif\if@alg@oneline@
\@alg@oneline@false
\newcommand\oneline{\@alg@oneline@true}
\newif\if@alg@short@
\@alg@short@false
\newcommand\short{\@alg@short@true}

\newcommand\algline{\SetLine}

\SetCommentSty{textit}
\SetKwComment{Comment}{--- }{}

\SetKwInOut{Input}{Input}
\SetKwInOut{Output}{Output}

\SetKwFor{For}{for}{\!:}{endfor}
\SetKwFor{ForAll}{For\,all}{\!,~do:}{endfor}
\SetKwFor{ForEach}{For\,each}{\!,~do:}{endfor}
\SetKwFor{ForEvery}{For~every}{\!,~do:}{endfor}
\SetKwFor{While}{While}{\!,~do:}{endwhile}

\SetKwRepeat{Repeat}{Repeat:}{until}

\SetKwIF{@For}{alg@e}{alg@f}{For}{\!,~do}{}{}{}
\SetKwIF{@ForAll}{alg@g}{alg@g}{For\,all}{\!,~do}{}{}{}
\SetKwIF{@ForEach}{alg@h}{alg@i}{For\,each}{\!,~do}{}{}{}
\SetKwIF{alt@ForEach}{alg@h}{alg@i}{For\,each}{\!:}{}{}{}
\SetKwIF{@While}{alg@j}{alg@k}{While}{\!,~do}{}{}{}
\SetKwIF{Short@If}{alg@a}{alg@b}{If}{}{else if}{else}{}
\SetKwIF{Short@Then}{alg@c}{alg@d}{then}{}{else}{}

\SetKwIF{If}{ElseIf}{Else}{If}{\!,~then}{else, if}{else}{endif}

\SetKwBlock{Begin}{begin}{end.}

\long\def\INDENTIF#1\THEN#2#3{%
	\ifx #3\ELSE
		\def\@tempa##1{%
			\bgroup
				\SetNoline\dontprintsemicolon
				\uShort@If{#1}{\lShort@Then{}{#2; \\ \lElse {##1.}}}%
			\egroup\@alg@short@false}\else
		\def\@tempa{%
				\bgroup
					\SetNoline\dontprintsemicolon
					\uShort@If{#1}{\lShort@Then{}{#2.}}%
			\egroup\@alg@short@false #3}\fi
	\@tempa}

\long\def\INDENTELSEIF#1\THEN#2#3{%
	\ifx #3\ELSE
		\def\@tempa##1{%
			\bgroup
				\SetNoline\dontprintsemicolon
				\uElseShort@If{#1}{\lShort@Then{}{#2; \\ \lElse {##1;}}}%
			\egroup\@alg@short@false}\else
		\def\@tempa{%
				\bgroup
					\SetNoline\dontprintsemicolon
					\uElseShort@If{#1}{\lShort@Then{}{#2;}}%
			\egroup\@alg@short@false #3}\fi
	\@tempa}

\long\def\IF#1\THEN#2#3{%
	\if@alg@oneline@
		\def\@tempa{\lIf{#1}{#2}#3}%
		\ifx #3\ELSE\else\@alg@oneline@false\fi
	\else
		\if@alg@short@\def\@tempa{\INDENTIF{#1}\THEN{#2}#3}\else		
			\ifx #3\ELSE
				\def\@tempa{\uIf{#1}{\SetNoline #2}#3}\else
				\ifx #3\ELSEIF
					\def\@tempa{\uIf{#1}{\SetNoline #2}#3}\else
					\def\@tempa{\If{#1}{\SetNoline #2}#3}%
				\fi
			\fi
		\fi
	\fi
	\@tempa}
\long\def\ELSEIF#1\THEN#2#3{%
	\if@alg@oneline@
		\def\@tempa{\lElseIf{#1}{#2}}%
		\ifx #3\ELSE\else\@alg@oneline@false\fi
	\else
		\if@alg@short@\def\@tempa{\INDENTELSEIF{#1}\THEN{#2}#3}\else		
			\ifx #3\ELSE
				\def\@tempa{\uElseIf{#1}{\SetNoline #2}#3}\else
				\ifx #3\ELSEIF
					\def\@tempa{\uElseIf{#1}{\SetNoline #2}#3}\else
					\def\@tempa{\ElseIf{#1}{\SetNoline #2}#3}%
				\fi
			\fi
		\fi
	\fi
	\@tempa}

\long\def\ELSE#1{%
	\if@alg@oneline@
		\def\@tempa{\lElse{#1}}\else
		\def\@tempa{\Else{\SetNoline #1}}\fi
	\@alg@oneline@false\@tempa}

\long\def\FOR#1\DO#2\ENDFOR{%
	\if@alg@oneline@
		\def\@tempa{\lFor{#1}{#2}}\else
		\if@alg@short@
			\def\@tempa{\SetNoline\u@For{#1}{#2}}\else
			\def\@tempa{\For{#1}{\SetNoline #2}}\fi\fi
	\@alg@oneline@false\@tempa}

\long\def\FORALL#1\DO#2\ENDFOR{%
	\if@alg@oneline@
		\def\@tempa{\lForAll{#1}{#2}}\else
		\if@alg@short@
			\def\@tempa{\SetNoline\u@ForAll{#1}{#2}}\else
			\def\@tempa{\ForAll{#1}{\SetNoline #2}}\fi\fi
	\@alg@oneline@false\@alg@short@false\@tempa}

\long\def\FOREACH#1\DO#2\ENDFOR{%
	\if@alg@oneline@
		\def\@tempa{\lalt@ForEach{#1}{#2}}\else
		\if@alg@short@
			\def\@tempa{\SetNoline\u@ForEach{#1}{#2}}\else
			\def\@tempa{\ForEach{#1}{\SetNoline #2}}\fi\fi
	\@alg@oneline@false\@alg@short@false\@tempa}

\long\def\FOREVERY#1\DO#2\ENDFOR{%
	\if@alg@oneline@
		\def\@tempa{%
		\bgroup
			\SetKwFor{ForEvery}{For\,every}{\!:}{endfor}
			\lForEvery{#1}{#2}%
		\egroup}
	\else
		\if@alg@short@
			\def\@tempa{\SetNoline\u@ForEvery{#1}{#2}}\else
			\def\@tempa{\ForEvery{#1}{\SetNoline #2}}\fi\fi
	\@alg@oneline@false\@alg@short@false\@tempa}

\long\def\WHILE#1\DO#2\ENDWHILE{%
	\if@alg@oneline@
		\def\@tempa{\lWhile{#1}{#2}}\else
		\if@alg@short@
			\def\@tempa{\SetNoline\u@While{#1}{#2}}\else
			\def\@tempa{\While{#1}{\SetNoline #2}}\fi\fi
	\@alg@oneline@false\@alg@short@false\@tempa}

\long\def\REPEAT#1\UNTIL#2{%
	\if@alg@oneline@
		\def\@tempa{\lRepeat{#2}{#1}}\else
		\def\@tempa{\Repeat{#2}{\SetNoline #1}}\fi
	\@alg@oneline@false\@tempa}

\def\BEGIN{\@ifstar{\@BEGIN[\SetNoline]}{\@BEGIN[\SetLine]}}

\long\def\@BEGIN[#1]#2\END{\par\bgroup #1\
Begin{\SetNoline #2 {}}\egroup}

\SetKwIF{Procedure}{a'}{a''}{procedure}{:}{a}{b}{c}
\SetKwIF{Subroutine}{b'}{b''}{subroutine}{:}{a}{b}{c}

\long\def\PROCEDURE#1:#2\BEGIN#3\END{%
	\ResetInOut{Output}%
	#2%
	\medskip\uProcedure{$#1$}{\smallskip #3}}

\long\def\SUBROUTINE#1:#2\BEGIN#3\END{%
	\ResetInOut{Effect}%
	#2
	\medskip\uSubroutine{$#1$}{\smallskip #3}}

\makeatother

%% file: circConversion.bbl
\begin{thebibliography}{10}

\bibitem{RB01}
R.~Raussendorf and H.~J. Briegel.
\newblock A one-way quantum computer.
\newblock {\em Physical Review Letters}, 86(5188), 2001.

\bibitem{RBB03}
R.~Raussendorf, D.~E. Browne, and H.~J. Briegel.
\newblock Measurement-based quantum computation on cluster states.
\newblock {\em Physical Review A}, 68(2):022312, Aug 2003.
\newblock \arXiv[quant-ph/0301052].

\bibitem{DKP06}
V.~Danos, E.~Kashefi, and P.~Panangaden.
\newblock Robust and parsimonious realisations of unitaries in the one-way
  model.
\newblock {\em Physical Review A}, 72:064301, 2006.

\bibitem{BKMP07}
D.~E. Browne, E.~Kashefi, M.~Mhalla, and S.~Perdrix.
\newblock Generalized flow and determinism in measurement-based quantum
  computation.
\newblock {\em New Journal of Physics}, 9(8):250, 2007.
\newblock \arXiv[quant-ph/0702212].

\bibitem{GotPhD}
D.~Gottesman.
\newblock {\em Stabilizer codes and quantum error correction}.
\newblock {PhD} in {P}hysics, California Institute of Technology, 1200 East
  California Blvd, Pasadena CA, 91125, 1997.
\newblock \arXiv[quant-ph/9705052].

\bibitem{GC99}
D.~Gottesman and I.~Chuang.
\newblock Demonstrating the viability of universal quantum computation using
  teleportation and single-qubit operations.
\newblock {\em Nature}, 402, 1999.
\newblock Preprint version, entitled ``Quantum Teleportation is a Universal
  Computational Primitive'' available as \arXiv[quant-ph/9908010v1].

\bibitem{FGGS00}
E.~Farhi, J.~Goldstone, S.~Gutmann, and M.~Sipser.
\newblock Quantum computation by adiabatic evolution.
\newblock Pre-print: \arXiv[quant-ph/0001106], 2000.

\bibitem{FGG07}
E.~Farhi, J.~Goldstone, and S.~Gutmann.
\newblock A quantum algorithm for the hamiltonian {NAND} tree.
\newblock Pre-print: \arXiv[quant-ph/0702144], 2007.

\bibitem{BeaudPhD}
J.~R.~N.~{de} Beaudrap.
\newblock {\em Theory of measurement-based quantum computing}.
\newblock {PhD} in {M}athematics, University of Waterloo, 2008.
\newblock Available online at \url{hdl.handle.net/10012/4133} and as
  \arXiv[0812.2869].

\bibitem{DKP07}
V.~Danos, E.~Kashefi, and P.~Panangaden.
\newblock The measurement calculus.
\newblock {\em J. ACM}, 54(2):8, 2007.
\newblock \arXiv[0704.1263].

\bibitem{CLM2005}
Andrew~M. Childs, Debbie~W. Leung, and Michael~A. Nielsen.
\newblock Unified derivations of measurement-based schemes for quantum
  computation.
\newblock {\em Physical Review A}, 71(3):032318, 2005.
\newblock \arXiv[quant-ph/0404132].

\bibitem{DK06}
V.~Danos and E.~Kashefi.
\newblock Determinism in the one-way model.
\newblock {\em Physical Review A}, 74:052310, 2006.
\newblock \arXiv[quant-ph/0506062].

\bibitem{Beaudrap08}
N.~{de} Beaudrap.
\newblock Finding flows in the one-way measurement model.
\newblock {\em Physical Review A}, 77:022328, 2008.
\newblock \arXiv[quant-ph/0611284].

\bibitem{MP08}
M.~Mhalla and S.~Perdrix.
\newblock Finding optimal flows efficiently.
\newblock In {\em Lecture Notes in Computer Science}, volume 5125, pages
  857--868, 2008.
\newblock \arXiv[0709.2670].

\bibitem{BP08}
N.~{de} Beaudrap and M.~Pei.
\newblock An extremal result for geometries in the one-way measurement model.
\newblock {\em Quantum Information \& Computation}, 8(5):430--437, 2008.
\newblock \arXiv[quant-ph/0702229].

\bibitem{Diestel}
R.~Diestel.
\newblock {\em Graph Theory}.
\newblock Springer-Verlag, Heidelberg, New York, 3rd edition, 2005.
\newblock Electronic version freely available on the author's website, at
  \url{www.math.uni-hamburg.de/home/diestel/books/graph.theory}.

\bibitem{NC00}
M.~A. Nielsen and I.~L. Chuang.
\newblock {\em {Q}uantum {C}omputation and {Q}uantum {I}nformation}.
\newblock Cambridge University Press, Cambridge, 2000.
\newblock See errata~\cite{NC00errata}.

\bibitem{Einstein1916}
A.~Einstein.
\newblock The foundation of the general theory of relativity.
\newblock {\em Annalen Phys.}, 49:769--822, 1916.

\bibitem{BB06}
D.~E. Browne and H.~J. Briegel.
\newblock One-way quantum computation --- a tutorial introduction, 2006.
\newblock \arXiv[quant-ph/0603226].

\bibitem{NC00errata}
M.~A. Nielsen and I.~L. Chuang.
\newblock Errata list for \emph{Quantum Computation and Quantum Information}.
\newblock Available online at
  \url{web.squint.org/qci/errata/errata/errata.html}; see~\cite{NC00}.

\bibitem{DHHMNO04}
C.~M. Dawson, H.~L. Haselgrove, A.~P. Hines, D.~Mortimer, M.~A. Nielsen, , and
  T.~J. Osborne.
\newblock Quantum computing and polynomial equations over \protect{$\mathbb
  Z_2$}.
\newblock {\em Quantum Information \& Computation}, 5(2):102, 2004.
\newblock \arXiv[quant-ph/0408129].

\end{thebibliography}
